%% file: Main.tex
\documentclass[11pt,english]{article}
\usepackage{amsmath,amssymb,amsthm}
\usepackage{multicol}

\usepackage{graphicx,color}
\usepackage{enumitem}
\usepackage{fullpage}
\usepackage[noblocks]{authblk}
\usepackage{tcolorbox}
\usepackage{babel}
\usepackage{wrapfig}
\usepackage{MnSymbol}
\usepackage{mdframed}
\usepackage{mathtools}
\usepackage{floatrow}
\usepackage[utf8]{inputenc}
\usepackage[ruled,linesnumbered,vlined]{algorithm2e}
\usepackage{tablefootnote}
\usepackage{thm-restate}
\usepackage{bbding}
\usepackage{pifont}
\usepackage{nicefrac}
\usepackage{wrapfig}
\usepackage{tabularx}
\usepackage{booktabs}

\definecolor{Darkblue}{rgb}{0,0,0.4}
\definecolor{Brown}{cmyk}{0,0.61,1.,0.60}
\definecolor{Purple}{cmyk}{0.45,0.86,0,0}
\definecolor{Darkgreen}{rgb}{0.133,0.543,0.133}

\usepackage[colorlinks,linkcolor=Darkblue,filecolor=blue,citecolor=blue,urlcolor=Darkblue,pagebackref]{hyperref}
\usepackage[nameinlink]{cleveref}

\usepackage[colorinlistoftodos,prependcaption,textsize=tiny]{todonotes}

\newif\ifdraft 
\draftfalse

\newcommand{\namedref}[2]{\hyperref[#2]{#1~\ref*{#2}}}
\newcommand{\propref}[1]{\hyperref[#1]{property~(\ref*{#1})}}

\newtheorem{theorem}{Theorem}
\newtheorem{lemma}{Lemma}
\newtheorem{definition}{Definition}
\newtheorem{claim}{Claim}
\newtheorem{observation}{Observation}
\newtheorem{corollary}{Corollary}

\newtheorem{property}{Property}
\newtheorem{question}{Question}

\newtheorem{remark}{Remark}

\DeclareMathOperator{\MST}{\mathrm{MST}}

\newcommand{\diam}{\mathrm{diam}}

\newcommand{\nd}{\delta}
\newcommand{\md}{\mathrm{maxdeg}}
\newcommand{\lvl}{\mathrm{lvl}}

\newcommand{\weight}{\text{\bf w}}

\newcommand{\dist}{\mbox{\bf dist}}

\newcommand{\ball}{\mathbf{B}}

\newcommand{\eps}{\epsilon}
\newcommand{\Cr}{\mathrm{Cross}}
\newcommand{\Aug}{\mathrm{Aug}}
\newcommand{\Pa}{\mathrm{Pa}}
\newcommand{\LCA}{\mathrm{LCA}}
\newcommand{\rad}{\mathrm{rad}}
\newcommand{\rep}{\mathrm{rep}}

\definecolor{forestgreen}{rgb}{0.13, 0.55, 0.13}

\SetCommentSty{mycommfont}

\DeclareMathAlphabet{\mathpzc}{OT1}{pzc}{m}{it}

\newlength{\dhatheight}

\newcommand {\ignore} [1] {}

\newcommand{\initOneLiners}{%
	\setlength{\itemsep}{0pt}
	\setlength{\parsep }{0pt}
	\setlength{\topsep }{0pt}
}

\SetKwInOut{KwIn}{Input}
\SetKwInOut{KwOut}{Output}

\title{Optimal Fault-Tolerant Spanners in Euclidean and Doubling Metrics: Breaking the $\Omega(\log n)$ Lightness Barrier}

\author{Hung Le}
\affil{University of Massachusetts at Amherst}

\author{Shay Solomon}
\affil{Tel Aviv University}

\author{Cuong Than}
\affil{University of Massachusetts at Amherst}

\date{}
\begin{document}
	\maketitle
	\begin{abstract}

\input{abstract_shorten}
	\end{abstract}
	 \thispagestyle{empty}
	\clearpage\maketitle
	
	
	
	\setcounter{tocdepth}{2} 
	\tableofcontents
\thispagestyle{empty}
	\clearpage
	\pagenumbering{arabic}
	\input{intro_v1}
	\input{Preliminaries}
	\input{VFT-optimal}
	\input{Degree-Analysis}
	\input{Lightness}
	\input{Connectivity}
	\input{fast_implementation}
	
	\paragraph{Acknowledgement.~} {Hung Le is supported by the NSF CAREER Award No.\ CCF-2237288 and an NSF Grant No.\ CCF-2121952. 
	Shay Solomon is funded by the European Union (ERC, DynOpt, 101043159).
	Views and opinions expressed are however those of the author(s) only and do not necessarily reflect those of the European Union or the European Research Council.
	Neither the European Union nor the granting authority can be held responsible for them.
	Shay Solomon is also supported by the Israel Science Foundation (ISF) grant No.1991/1.
	Shay Solomon is also supported by a grant from the United States-Israel Binational Science Foundation (BSF), Jerusalem, Israel, and the United States National Science Foundation (NSF).
	Cuong Than is supported by the NSF CAREER Award No.\ CCF-2237288 and an NSF Grant No.\ CCF-2121952.}

	\bibliographystyle{alphaurlinit}
	\bibliography{RamseyTreewidthBib, RPTALGbib}

\end{document}

%% file: abstract_shorten.tex
An essential requirement of spanners in many applications is to be fault-tolerant: a $(1+\epsilon)$-spanner of a metric space is called (vertex) $f$-fault-tolerant ($f$-FT) if it remains a $(1+\epsilon)$-spanner (for the non-faulty points) when up to $f$ faulty points are removed from the spanner. Fault-tolerant (FT) spanners for Euclidean and doubling metrics have been extensively studied since the 90s.

For low-dimensional Euclidean metrics, Czumaj and Zhao in SoCG'03 [CZ03] showed that the optimal guarantees $O(f n)$, $O(f)$ and $O(f^2)$ on the size, degree and lightness of $f$-FT spanners can be achieved via a greedy algorithm, which na\"{\i}vely runs in $O(n^3) \cdot 2^{O(f)}$ time. The question of whether the optimal bounds of [CZ03] can be achieved via a fast construction has remained elusive, with the lightness parameter being the bottleneck. Moreover, in the wider family of doubling metrics, it is not even clear whether there exists an $f$-FT spanner with lightness that depends solely on $f$ (even exponentially): all existing constructions have lightness $\Omega(\log n)$ since they are built on the net-tree spanner, which is induced by a hierarchical net-tree of lightness  $\Omega(\log n)$.

In this paper we settle in the affirmative these longstanding open questions. Specifically, we design a construction of $f$-FT spanners that is optimal with respect to all the involved parameters (size, degree, lightness and running time): For any $n$-point doubling metric, any $\epsilon > 0$, and any integer   $1 \le f \le n-2$,
our construction provides, within time $O(n  \log n + f n)$, an $f$-FT $(1+\epsilon)$-spanner with size $O(f n)$, degree $O(f)$ and lightness $O(f^2)$.

To break the $\Omega(\log n)$ lightness barrier,  we introduce a new geometric object --- the \textit{light net-forest}.
Like the net-tree, the light net-forest is induced by a hierarchy of nets.
However, to ensure small lightness, the light net-forest is inherently less ``well-connected'' than the net-tree, which, in turn, makes the task of achieving fault-tolerance significantly more challenging. Further, to achieve the optimal degree (and size) \textit{together} with optimal lightness, and to do so within the optimal running time --- we overcome several highly nontrivial technical challenges. 

%% file: intro_v1.tex
\section{Introduction} \label{sec1}

\vspace{4pt}
\subsection{Euclidean Spanners} \label{euc}
Let $P$ be a set of $n$ points in $\mathbb R^d$ and let $\eps > 0$ be any parameter.
A spanning subgraph  $H = (P,E,\|\cdot \|)$  of the complete Euclidean graph induced by $P$
is called a (Euclidean) \emph{$(1+\eps)$-spanner} 
for the point set $P$ if $\forall p,q \in P$, there is 
a \emph{$(1+\eps)$-spanner path} in $H$ between $p$ and $q$, i.e., a path of {\em weight} at most $(1+\eps) \cdot \|p-q\|$,
where the weight of a path is the sum of all edge weights in it and $\|p-q\|$ denotes the Euclidean distance between $p$ and $q$.
Euclidean spanners find applications in various areas, including in geometric approximation algorithms, network topology design and distributed systems, and they have been studied extensively since the 80s \cite{Chew86,Clark87,KG92,ADDJS93,ADMSS95,AWY05};
see also the book by Narasimhan and Smid \cite{NS07}
titled ``Geometric Spanner Networks'', which is devoted to Euclidean spanners and their applications.

A natural requirement from a spanner, which is essential for 
real-life applications, is to be {\em robust against failures}, so that even when part of the network fails, 
we still have a good spanner for the functioning part of the network.
Formally, a Euclidean spanner $H$ for point set $P$ is called a \emph{(vertex) $f$-FT $(1+\eps)$-spanner}, for $1 \le f \le n-2$, if for any
$F \subseteq P$ with $|F| \le f$, the graph $H \setminus F$ (obtained by removing from $H$ the vertices of $F$ and their incident edges) is a $(1+\eps)$-spanner for $P \setminus F$.\footnote{We shall restrict the attention to vertex faults, but that does not lose generality:
Any FT $(1+\eps)$-spanner that is resilient to $f$ vertex faults
is also resilient to $f$ edge faults (see \cite{LNS98,NS07}), while the lower bounds discussed below --- of $\Omega(f)$ on sparsity and degree and $\Omega(f^2)$ on lightness of $f$-FT spanners --- apply also to edge faults.}
(The basic (non-FT) setting corresponds to the case $f = 0$.) 
To perform efficiently in these applications, we would like the underlying spanner to be ``sparse''.
The {\em size} (number of edges) of the spanner is perhaps the most basic sparsity measure; the spanner {\em sparsity} is defined as the ratio of the spanner size to the minimum size $n-1$ of a connected spanning subgraph. 
The {\em weight} (sum of edge weights) of the spanner is a natural generalization of the size, and in many applications (such as for the metric TSP) we need to have small weight rather than small size; the spanner {\em lightness} is defined as the ratio of the spanner weight to the 
minimum spanning tree (MST) weight.
The spanner size corresponds to the {\em average} degree of a vertex,  
yet the stronger property of a small {\em (maximum) degree} (over all vertices) is important
in various applications in Computational Geometry,
as well as for reducing the space usage in compact routing schemes and in distributed systems.

A construction of Euclidean $(1+\eps$)-spanners with constant degree (and sparsity) and lightness can be built in time $O(n \log n)$ \cite{AS94,DN94,GLN02}.
In their pioneering work, Levcopoulos, Narasimhan and Smid \cite{LNS98} 
introduced the notion of FT spanners and generalized the basic construction of \cite{AS94,DN94,GLN02}
to obtain an $f$-FT $(1+\eps)$-spanner with degree (and sparsity) and lightness bounded by $2^{O(f)}$, within a running time of $O(n \log n) + n 2^{O(f)}$.
Clearly, the degree of any vertex in any $f$-FT spanner (for any stretch) must be at least $f+1$, and thus any $f$-FT spanner must have $\Omega(f n)$ edges. There are also simple point sets (even in 1 dimension), for which any $f$-FT spanner must have lightness $\Omega(f^2)$  \cite{CZ03}. 
Finally, the time needed to compute a $k$-FT spanner is $\Omega(n \log n + f n)$: The term $\Omega(n \log n)$
is the time lower bound for computing a basic (non-FT) spanner in the \emph{algebraic computation tree model} \cite{CDS01},
and the term $\Omega(f n)$ is the aforementioned lower bound on the size of any $f$-FT spanner.

There are also bunch of other constructions of Euclidean FT spanners (see Table \ref{tab1} for a summary of FT constructions),
and they can be grouped into two categories. 

In the first category, which contains almost all known constructions, the lightness parameter is either ignored or bounded from below by $\Omega(\log n)$; this line of work was culminated with the construction of Solomon \cite{Sol14} from STOC'14, which achieves the optimal degree $O(f)$ and running time $O(n \log n + n f)$, and it applies to the wider family of {\em doubling metrics} (see Section \ref{dob}). 
In addition to optimal degree and running time, the construction of \cite{Sol14} also achieves a near-optimal tradeoff 
of $O(f^2 + f \log n)$ versus $O(\log n)$
between the lightness and another property called the {\em (hop-)diameter};\footnote{A spanner for point set $P$ is said to have a {\em (hop-) diameter} of $k$ if it provides a $(1+\eps)$-spanner path with at most $k$ edges, for every $p,q \in P$.}
 importantly, any construction of diameter $O(\log n)$
must have lightness $\Omega(\log n)$, and more precisely this tradeoff between lightness and diameter is optimal up to a factor $\log f$ slack
on the diameter and a factor $\min\{\log f, \frac{\log n}{f}\} = O(\log \log n)$ slack on the lightness. (For a detailed discussion on the lower bound tradeoff between the lightness and diameter, we refer to \cite{Sol14}.)

The second category concerns ``light'' spanners, and there are only two such constructions to date. The first is the one by \cite{LNS98} mentioned above, and the second is a {\em greedy} construction due to Czumaj and Zhao from SoCG'03 \cite{CZ03},
and it achieves the optimal lightness guarantee of $O(f^2)$ together with the optimal degree (and sparsity) of $O(f)$. 
However, a na\"{\i}ve implementation of the greedy FT spanner construction requires time $\Omega(n^3) \cdot 2^{\Omega(f)}$
or more precisely $\Omega(n^2 \cdot \mathtt{Paths}(n,f+1,1+\eps))$,
where $\mathtt{Paths}(n,f+1,1+\eps)$ is the time needed to check whether an $n$-vertex
 Euclidean graph with $O(f n)$ edges contains $f+1$ vertex-disjoint $(1+\eps)$-spanner paths between an arbitrary pair of vertices.
Indeed, in the greedy FT construction, the ${n \choose 2}$ edges of the underlying Euclidean metric are traversed by nondecreasing weights, and each edge $(x,y)$ is added to the current spanner iff it does not contain $f+1$ vertex-disjoint $(1+\eps)$-spanner paths between $x$ and $y$. We note that a more sophisticated implementation of the basic (non-FT) greedy algorithm takes time $O(n^2 \log n)$ in Euclidean and doubling metrics \cite{BCFMS10}, but it is unclear if this implementation can be extended to the FT greedy algorithm; moreover, even if such an extension is possible
and even if we completely ignore the dependence on $f$, it would still lead to a super-quadratic in $n$ runtime.

\begin{table*}
\begin{center}
\resizebox {\textwidth }{!}{
\footnotesize
\begin{tabular}{|c|c|c|c|c|c|}
\hline  Reference & Sparsity & Degree &  Lightness & Runtime & Metric  \\
\hline \cite{LNS98}  & $2^{O(f)}$ & $2^{O(f)}$    & $2^{O(f)}$ & $O(n  \log n) + n 2^{O(f)}$ & Euclidean \\
\hline \cite{LNS98}   & $O(f^2)$ & unspecified   & unspecified & $O(n  \log n + f^2 n)$ & Euclidean \\
\hline \cite{LNS98}   & $O(f \log n)$ & unspecified   & unspecified & $O(f n \log n)$ & Euclidean \\
\hline \cite{Luk99}   & $O(f)$ & unspecified   & unspecified & $O(n \log^{d-1} n + f n  \log \log n)$ & Euclidean \\
\hline \cite{CZ03}   & $O(f)$ & $O(f)$   & $O(f^2)$ & $\Omega(n^3) \cdot 2^{O(f)})$ & Euclidean \\
\hline \cite{CZ03}   & $O(f)$ & $O(f)$   & $O(f^2 \log n)$ & $O(f n \log^{d} n + n f^2 \log f)$ & Euclidean \\
\hline \cite{CLN12a}   & $O(f)$ & unspecified  & unspecified & unspecified & doubling \\
\hline \cite{CLN12a}   & $O(f^2)$ & $O(f^2)$  & unspecified & unspecified & doubling \\
\hline \cite{CLNS13}   & $O(f^2)$ & $O(f^2)$ &  $O(f^2 \log n)$ & $O(n \log n + f^2 n)$ & doubling \\
\hline \cite{Sol14} &  $O(f)$ &  $O(f)$ &    $O(f^2 + f \log n)$ &    $O(n \log n + f n)$ & {\bf doubling} \\
\hline \hline {\bf New} &   {\boldmath $O(f)$} & {\boldmath $O(f)$} & {\boldmath $O(f^2)$} &  {\boldmath $O(n \log n + f n)$} &   {\bf doubling} \\
\hline
\end{tabular}
}
\end{center}
\caption[]{\label{tab1} \footnotesize A  comparison between previous  and our  
constructions of FT spanners with small size, degree, lightness and runtime, for low-dimensional Euclidean and doubling metrics.
The $O$-notation ignores dependencies on $\eps$ and  $d$.}
\end{table*}

Up to this date no construction of Euclidean $f$-FT spanners with runtime better than the $\Omega(n^3) \cdot 2^{\Omega(f)}$ bound of \cite{CZ03}, let alone the optimal $O(n \log n + fn)$ runtime bound,
could achieve lightness $\min\{o(\log n), 2^{o(f)}\}$, let alone the optimal lightness of $O(f^2)$.
In particular, the following question, stated in the book of Narasimhan and Smid \cite{NS07}, has remained open even for 2-dimensional point sets since the STOC'98 work of \cite{LNS98}.
\begin{question} [Open Problem 28 in \cite{NS07}] \label{Q3}
Is there an algorithm that constructs,  within $O(n \log n + fn)$ time, an $f$-FT $(1+\eps)$-spanner with lightness $O(f^2)$?
\end{question}

Further, is it possible to construct, ideally still in time  $O(n \log n + fn)$, a single construction
that combines the optimal lightness $O(f^2)$ with the optimal degree (and sparsity) $O(f)$?
\begin{question} \label{Qmain}
Is there an algorithm that constructs, within $O(n \log n + fn)$ time, an $f$-FT $(1+\eps)$-spanner with lightness $O(f^2)$ and degree $O(f)$?
\end{question}

\subsection{Doubling Metrics} \label{dob}
A metric is called \emph{doubling} if its {\em doubling dimension} is constant, where the latter is the smallest value $d$
such that every ball $B$ in the metric can be covered by at most
$2^{d}$ balls of half the radius of $B$. 
We note that the doubling dimension generalizes the standard Euclidean dimension, since the doubling dimension
of the Euclidean space $\mathbb R^d$ is $\Theta(d)$.
Spanners for doubling metrics have been intensively studied; see  \cite{GGN04,CGMZ05,CG06,HPM06,GR081,GR082,Smid09,CLN12a,ES15,CLNS13,Sol14,BLW19,LT22, KLMS22}, and the references therein.
Many of these works share a common theme, namely, to devise spanners for doubling metrics that are just as good as the analog Euclidean spanner constructions.

Some of the constructions mentioned in Section \ref{euc} apply to doubling metrics; see Table \ref{tab1}.
Much weaker variants
of Questions \ref{Q3} and \ref{Qmain} from Section \ref{euc} can be asked for the wider family of doubling metrics.
Indeed, in such metrics, it is not even clear whether {\em there exists} an $f$-FT spanner with lightness that depends solely on $f$ (even exponentially): All existing constructions have lightness $\Omega(\log n)$ since they are built on the {\em net-tree spanner}, which is induced by a hierarchical {\em net-tree} of lightness  $\Omega(\log n)$. The net-tree incurs a lightness of $\Omega(\log n)$ even for line metrics!

\begin{question} \label{Qultimate}
\begin{itemize}
\item
Does there exist,  for any doubling metric,  an $f$-FT $(1+\eps)$-spanner with lightness $O(f^2)$?
Does there exist such a spanner that also achieves degree $O(f)$?
\item Further, is there an algorithm that constructs, for any doubling metric,  within $O(n \log n + fn)$ time, 
an $f$-FT $(1+\eps)$-spanner with lightness $O(f^2)$ and degree  $O(f)$?
\end{itemize}
\end{question}

\subsection{Our Contribution}
The main result of this work is the following theorem. 

\begin{theorem} \label{tm1}
Let $(X,\delta)$ be an $n$-point doubling metric, with an arbitrary doubling dimension $d$.
For any $0 < \eps < \frac{1}{2}$ and any integer $1 \le f \le n-2$,
an $f$-FT $(1+\eps)$-spanner with lightness $\eps^{-O(d)}  \cdot f^2$ and degree $\eps^{-O(d)} \cdot f$ 
can be built within $\eps^{-O(d)}  (n \log n + fn)$ time.
\end{theorem}

The construction provided by Theorem \ref{tm1} is {\bf optimal with respect to all the involved parameters}
and it settles all the aforementioned questions, Questions \ref{Q3}-\ref{Qultimate}, in the affirmative.
We note that our construction improves the previous state-of-the-art constructions of FT spanners (with sub-cubic runtime) not only for doubling metrics, but also for Euclidean ones.

A central challenge that we faced on the way to proving Theorem \ref{tm1} is breaking the $\Omega(\log n)$ lightness barrier. 
To this end we introduce a new geometric object --- the {\em light net-forest}.
Like the net-tree, the light net-forest is induced by a hierarchy of nets.
However, to ensure small lightness, the light net-forest is inherently less ``well-connected'' than the net-tree, which, in turn, makes the task of achieving fault-tolerance significantly more challenging. 
We demonstrate the power of the light net-forest in achieving the optimal lightness, but our construction does not stop there.
Further, to achieve the optimal degree (and size) {\em together} with optimal lightness, and to do so within the optimal running time --- we overcome several highly nontrivial technical challenges.
In the following section we describe the technical and conceptual contributions of this work.

\subsection{Technical Overview and Conceptual Highlights} \label{techhighlight}
The previous constructions of Euclidean FT spanners \cite{LNS98,Luk99,CZ03,NS07} rely on geometric properties of low-dimensional Euclidean metrics, such as the \emph{gap property} \cite{AS94} and the \emph{leapfrog property} \cite{DHN93}.
In particular, achieving small lightness crucially relies on the leapfrog property, which is not known to extend to arbitrary doubling metrics.
On the other hand, the previous constructions of doubling FT spanners \cite{CLN12a,CLNS13,Sol14} use standard packing arguments of doubling metrics, but they all rely on the standard {\em net-tree spanner}  of \cite{GGN04,CGMZ05}, which is induced by a hierarchical {\em net-tree}
 $T = T(X)$ that corresponds to a \emph{hierarchical partition} of the metric $(X,\delta)$.
The lightness of the net-tree alone is $\Omega(\log n)$, even in 1-dimensional Euclidean spaces; as such, all the  constructions of \cite{CLN12a,CLNS13,Sol14} incur a lightness of $\Omega(\log n)$ even when ignoring the dependencies on $f$.

As in all previous constructions that apply to arbitrary doubling metrics, the starting point of our construction is the {net-tree spanner} $T = T(X)$. To break the lightness barrier of $\Omega(\log n)$, we will not be able to use the entire net-tree, and consequently we will not be able to use the entire net-tree spanner that is derived from it. 
We start with a brief overview of the net-tree spanner and the previous constructions of \cite{CLN12a,CLNS13,Sol14}.

Any tree node $x$ is associated with a single point $\rep(x)$ that belongs to the point set $L(x)$ of its descendant leaves. 
For any pair $x,y$ of level-$i$ tree nodes that are close together with respect to the \emph{distance scale} (or \emph{radius}) $2^i$ at that level, a \emph{cross edge} $(x,y)$ is added (edge $(x,y)$ translates to edge  $(\rep(x),\rep(y))$; for brevity we sometimes write $(x,y)$ as a shortcut for $(\rep(x),\rep(y))$); specifically, the weight of any level-$i$ cross edge $(x,y)$ is at most $\lambda 2^i$, where $\lambda = \Theta(\frac{1}{\eps})$ and $\rad(x) = \rad(y) = 2^i$. 
The net-tree spanner is the union of the tree edges (i.e., the edges of $T$) and the cross edges.
For every pair $u,v$ of points, a $(1+\eps)$-spanner path, denoted by $\Pi_T(u,v)$, goes up in the tree $T$ from a leaf $x$ corresponding to $u$ (i.e., $\rep(x) = u$) to some ancestor $x'$ of $u$, then takes a cross edge $(\rep(x'),\rep(y'))$ from $x'$ to an ancestor $y'$ of $y$ in the net-tree, and finally goes down in $T$ from $y'$ to $y$, where $\rep(y) = v$.
The reason  $\Pi_T(u,v)$ is a $(1+\eps)$-spanner path is due to the following key observation of the net-tree, which implies that the weight of the cross edge constitutes almost the entire weight of $\Pi_T(u,v)$:
\begin{observation} \label{ob:net}
Both $\delta(\rep(x),\rep(x'))$ and $\delta(\rep(y),\rep(y'))$ are at most $\eps \delta(u,v)$. 
Thus $(1-\eps) \delta(u,v) \le \delta(\rep(x'),\rep(y')) \le (1+\eps) \delta(u,v)$.
\end{observation}
To achieve fault-tolerance, the general idea in \cite{CLNS13} was to associate each tree node $x$ with
a {\em surrogate set} $S(x)$ of (up to) $f+1$ points from $L(x)$  rather than a single point; the FT spanner is obtained by replacing each edge $(x,y)$
of the basic net-tree spanner by a bipartite clique between the corresponding sets $S(x)$ and $S(y)$.
To achieve a degree of $O(f^2)$, \cite{CLNS13} used a ``rerouting'' technique from \cite{GR082} that assigns representative points for the tree nodes so as to minimize the maximum degree, but
achieving degree $o(f^2)$ using this approaches is doomed for  two reasons.
\\
{\bf  ``Global'' reason}: If each edge $(x,y)$ of the basic net-tree spanner is replaced by a bipartite clique between  $S(x)$ and $S(y)$,
then since $S(x)$ and $S(y)$ may contain $\Omega(f)$ points each, the size of this bipartite clique may be $\Omega(f^2)$.
As the basic net-tree spanner has $\Omega(n)$ edges, 
the FT spanner obtained in this way will contain $\Omega(f^2 n)$ edges (and will have degree $\Omega(f^2)$).
\\
{\bf ``Local'' reason}: As a node $x$ is associated with a set $S(x)$ of (up to) $f+1$ points from $L(x)$,
the same leaf point $p$ may belong to $\Omega(f)$ different sets $S(x)$ of internal nodes $x$.
For each  edge $(x,y)$ of the basic net-tree spanner that is incident on any of these $\Omega(f)$ nodes, 
$p \in S(x)$ is connected via edges to all $\Omega(f)$ points of $S(y)$, and so the degree of $p$ will be $\Omega(f^2)$.

\paragraph{The key idea of \cite{Sol14}.}
The idea of associating nodes of net-trees and other hierarchical tree structures, such as split trees and dumbbell trees,
with points from their descendant leaves has been widely used in the geometric spanner literature; see \cite{ADMSS95,GGN04,CGMZ05,CG06,NS07,GR082,CLN12a},
and the references therein.
Indeed, by Observation \ref{ob:net}, any point in $L(x)$ is close to the original net-point $\rep(x)$ of $x$, with respect to the distance scale $\rad(x)$ of $x$.
Instead of associating nodes $x$ with points chosen exclusively from $L(x)$,
the key idea of \cite{Sol14} is to consider a wider set $B(x)$ of all points
in the ball of radius $O(\rad(x))$ centered at $\rep(x)$.  
By associating nodes $x$ with points from $B(x)$, 
one obtains a \emph{hierarchical cover} of the metric (rather than a hierarchical partition). As the doubling dimension is constant,
this cover has a constant \emph{degree}, i.e., every point belongs to $O_{\eps,d}(1)$ sets $B(\cdot)$ at each level of the tree $T$.

Similarly to \cite{CLNS13}, \cite{Sol14} associates each tree node $x$ with a set $S(x)$ of (up to) $f+1$ points called \emph{surrogates}, but as mentioned the surrogates in \cite{Sol14} are chosen from the superset $B(x)$ of $L(x)$.
Moreover, \cite{Sol14} doesn't naively replace each edge
$(x,y)$ of the basic net-tree spanner by a bipartite clique between 
$S(x)$ and $S(y)$ as in \cite{CLNS13}, since (due to the ``Global'' reason above)
that would lead to $\Omega(f^2 n)$ edges. 
Instead, whenever the number of surrogates in $S(x)$ and $S(y)$ is $f+1$, a bipartite matching suffices for achieving fault-tolerance.
\cite{Sol14} assigns the surrogates bottom-up (first for level-0 nodes (leaves) in the net-tree $T$, then for level-1 nodes, etc.) via a complex procedure that guarantees that, for any level $i$ and any level-$i$ node $x$ in $T$, 
there are enough points of small degree in $B(x)$ to choose surrogates from, which ultimately leads to the desired degree bound of $O(f)$.
However, as mentioned, the lightness of the construction of \cite{Sol14}, as well as any other construction that applies to doubling metrics, is lower bounded by the lightness $\Omega(\log n)$ of the underlying net-tree; more precisely, the state-of-the-art lightness of any construction of FT spanners in doubling metrics prior to this work, due to \cite{Sol14}, is $\Theta(f^2 + f \log n)$.

\subsubsection{Our approach}
To breach the $\Omega(\log n)$ lightness barrier incurred by the net-tree, our first insight is that a ``light'' $(1 + \eps)$-spanner $G$ of $X$,
which is given as input, can be used for computing a light subtree of the net-tree --- which we name the {\em light net-forest} and abbreviate as LNF.
Equipped with the LNF, a natural approach would be to (i) apply the standard net-tree spanner construction on top of the LNF (instead of the net-tree) in order to get a light net-tree spanner, and (ii) transform the light net-tree spanner into a light FT spanner by replacing each edge $(x,y)$ of the spanner with a bipartite clique / matching between $S(x)$ and $S(y)$ similarly to \cite{Sol14}.
Alas, since the LNF is not a tree but rather a forest, it is inherently less ``well-connected'' than the net-tree, which renders step (ii) of achieving fault-tolerance highly challenging, as we next describe. 

In the standard net-tree spanner, as mentioned, for every pair $u,v$ of points, there is a path $\Pi_T(u,v)$ that uses a {\em single cross edge}, which we denote here by $(u',v')$, where $u'$ and $v'$ are ancestors of $u$ and $v$ in the net-tree $T$, respectively. 
A non-faulty $(1+\eps)$-spanner path in the FT spanner constructions of \cite{CLNS13,Sol14}, for a {\bf non-faulty pair $u,v$} of points, 
consists of going up in $T$ from $u$ and from $v$ to some non-faulty surrogates of $S(u')$ and $S(v')$, respectively, and it includes a single cross edge between those surrogates. There are non-faulty surrogates in $S(u')$ and $S(v')$ due to two reasons: (1) $u$ and $v$ are non-faulty descendant leaves of $u'$ and $v'$, respectively, and (2) there is only one cross edge in the path. Indeed, if $|S(u')| < f+1$ then,
since $u$ is a descendant leaf of $u'$, one can guarantee that $u \in S(u')$; otherwise $u$ may not belong to $S(u')$, but in that case too $S(u')$ contains a non-faulty surrogate; the same goes for $S(v')$. 
However, when using the LNS, we no longer have such a path $\Pi_T(u,v)$ for any pair $u,v$ of points.
Instead, we can afford to use such a path $\Pi_T(u,v)$ only for edges $(u,v)$ that belong to the light $(1+\eps)$-spanner $G$ (with constant lightness) that we receive as input; it can be shown that the union of all those paths over the edges of $G$ has constant lightness. 
Consider now a pair $u,v$ of points that are not incident in $G$, and let $\Pi_G(u,v) = (u_1 = u, u_2, \ldots, u_k = v)$ be a shortest path (which is a $(1+\eps)$-spanner path) between $u$ and
$v$ in $G$. Naturally, we would like to translate this path $\Pi_G = \Pi_G(u,v)$ into a union of non-faulty subpaths of the net-tree plus cross edges,
as that union should provide a non-faulty $(1+\eps)$-spanner path between $u$ and $v$. 
However, we cannot argue that ancestors $u'_i$ of intermediate nodes $u_i, u_i \ne u,v$ on the path contain non-faulty surrogates, since $u$ and $v$ are not necessarily descendant leaves of $u'_i$!
{\bf This is the crux in achieving a light FT-spanner from the LNS,} and it entails multiple challenges.

When translating $\Pi_G$ into a non-faulty path, the basic idea is to replace every cross edge by a bipartite matching.
First note that the cross edges corresponding to the edges of $\Pi_G$ may be located at different levels.
Indeed, for an edge $e_i = (u_i,u_{i+1})$ in $G$, the respective cross edge, denoted by $\mathtt{cross}(e_i) = (\hat{u}^i_i, \hat{u}^i_{i+1})$, lies at level roughly 
$\log \delta(u_i,u_{i+1})$, and edge weights in the path $\Pi_G$ may be very different one from another;
thus $\hat{u}^i_{i+1}$ and $\hat{u}^{i+1}_{i+1}$, which are the ``second'' endpoint of $\mathtt{cross}(e_i)$ and the ``first'' endpoint of 
$\mathtt{cross}(e_{i+1})$, respectively, may lie at very different levels of the net-tree.
Next, we stress that replacing a cross edge by a bipartite matching is possible only in the case that, for each cross edge $\mathtt{cross}(e_i) = (\hat{u}^i_i, \hat{u}^i_{i+1})$ along the path, there are $f+1$ 
``nearby'' points around the two endpoints $\hat{u}^i_i$ and $\hat{u}^i_{i+1}$ of the edge. 
By nearby we mean within distance that is an $O(\eps)$-fraction of the distance between $\hat{u}^i_i$ and $\hat{u}^i_{i+1}$,
and thus these nearby points can serve as part of the surrogate sets of $S(\hat{u}^i_i)$ and $S(\hat{u}^i_{i+1})$ of $\hat{u}^i_i$ and $\hat{u}^i_{i+1}$, respectively.
Let us consider this simpler case first. Then the distance between a pair of surrogates of $\hat{u}^i_i$ and $\hat{u}^i_{i+1}$  is the same, up to a factor of $1+O(\eps)$, as the distance between $\hat{u}^i_i$ and $\hat{u}^i_{i+1}$. 
It thus suffices to take a perfect matching between $S(\hat{u}^i_i)$ and $S(\hat{u}^i_{i+1})$ for every $i$, and then at least one of the matching edges must function for every $i$ (following at most $f$ vertex faults). 
One technicality is that when going from $\hat{u}^{i - 1}_{i-1}$ to $\hat{u}^{i-1}_{i}$ through a matching edge we
end up at a surrogate of $\hat{u}^{i - 1}_{i}$, say $s^{i - 1}_{i}$, and then when continuing from $\hat{u}^i_{i}$ to $\hat{u}^i_{i+1}$ through a matching edge we start at a surrogate of $\hat{u}^i_{i}$ (rather than $\hat{u}^{i - 1}_{i}$), say $s^i_{i}$, and these two surrogates 
$s^{i - 1}_{i}$ and $s^i_{i}$
might be different, so in general the union of the matched edges along the path does not form a valid path; however, as the distance between $s^{i - 1}_{i}$ and $s^i_{i}$ is negligible with respect to the weight of $e_i$,
our construction will provide a non-faulty path between  $s^{i - 1}_{i}$ and $s^i_{i}$ by induction.
The challenging case is when for some cross edges along the path, there are less than $f+1$ 
``nearby'' points around the two endpoints;
we refer to such edges as \textit{irreplaceable}.
Dealing with irreplaceable edges requires special care; we will get back to this issue towards the end of this section, when we discuss how to construct a non-faulty spanner path. For now we focus on the following fundamental issue. 

A major issue is that there could be nodes with less than $f+1$ nearby points.  For an edge $e_i$, there might be no functioning vertex near the endpoints of  $\mathtt{cross}(e_i)$, which means that there might be no functioning edge which is a good approximation of $e_i$ in $H - F$ (in terms of its weight).  We resolve this issue by adding more cross edges to guarantee fault tolerance.  Specifically, for every node $x$ for which we cannot find $f + 1$ surrogates, we add the bipartite clique  between $S(x)$ and $S(y)$ for all $y$ within distance $O(\rad(x) / \eps)$ from $x$.  There are two possible cases when we cannot find $f + 1$ surrogates for a node $x$. The first (and obvious) case is when $x$ is at a low level and there are not enough (less than $f + 1$) points near $\rep(x)$. The second (and more challenging) case is that the choices of surrogates must be subjected to other constraints such as bounded degree; if we choose the surrogates carelessly, there might be fewer usable vertices in the vicinity of $\rep(x)$ since all vertices close to $\rep(x)$ were overused as surrogates by nodes at levels lower than $x$.
In both cases, we add the bipartite clique between $S(x)$ to the surrogate set of every  nearby node.
However, each case contributes differently to the lightness.
While adding many edges at low levels does not affect the lightness significantly, adding those in higher levels could blow up the lightness by a factor of $\log{n}$; recall that we want to avoid using all the cross edges since they add a $\log{n}$ factor to the lightness. 
Thus, we must choose the surrogates carefully to ensure that the second case essentially does not happen. This is where our aforementioned LNF (light net forest) comes to the rescue.  Our key insight is that all nodes in the net-tree with less than $f +1$ nearby surrogates form an LNF, which consists of node-disjoint subtrees of $T$ that span all those nodes, such that the total radii associated with those nodes exceeds the weight of $\MST(X)$ by at most a constant factor, implying that the LNF is light. The structure of the LNF depends heavily on the input spanner $G$ and on our strategy to select surrogates.

The strategy of selecting surrogates directly affects all important parameters of our spanner: stretch, lightness and degree.
First, the surrogates of a node $x$ must be close to $\rep(x)$ to guarantee the stretch. 
Second, to guarantee the bounded degree property, we must avoid overusing any point as a  surrogate (i.e., using it as a surrogate for too many nodes in the tree). 
One idea is to use the leaves of a node as surrogates since the distance between (the representative of) a node $x$ to its leaves is small compared to any cross edge incident to $x$. However, using only leaves for surrogates might increase the degree of a single vertex to $\Omega(n)$; for example,  there might be a long branch of $T$ with only one leaf, say $z$, where we have to add edges to all the nodes in the branch to guarantee a good stretch, leaving $z$ with many incident edges. To overcome this issue, we choose the surrogates of $x$ in a ball centered at $\rep(x)$ with radius $c \cdot \rad(x)$, for an appropriate constant $c$. Once the degree (in the spanner) of a vertex reaches a certain threshold, it is forbidden for that vertex to serve as a surrogate ever again.

A new problem arises: choosing surrogates from a ball as suggested above does not guarantee the bounded lightness property.  Balls centered at net points at different levels might interact in a complex manner, and as a result, we might not be able to find enough (at least $f + 1$) surrogates for some nodes at high levels, since overused vertices are forbidden to be used again. If we follow the suggestion outlined in the previous paragraph, we will have to add bipartite cliques from the surrogates of these nodes to other nodes's surrogates. Unfortunately, this will break the property of the LNF and, in particular, it will be harder to control the lightness. To resolve this problem, we will guarantee the following property: if a node has at least $f + 1$ surrogates, any of its ancestors will also have $f + 1$ surrogates.  To this end, our key idea is to {\em prioritize} the choice of high-degree vertices in a large ball $\ball(x, 16\rad(x))$ over vertices of lower degree in $\ball(x, 4\rad(x))$.   This is somewhat counter-intuitive as one might expect to prioritize low degree vertices (to have a better chance of bounding the degree of the spanner) over high degree vertices. However, the intuition is that by prioritizing high degree vertices, we actually ``save'' the low degree ones to be used as ``fresh'' surrogates of other nodes at higher levels. 
Furthermore, when a high degree vertex becomes overused, many of its neighbors of low degree are closer to its ancestors, relative to the radii of the ancestors, leaving the ancestors more room to choose surrogates. 

We stress that even if we only try to control the lightness bound (regardless of the degree and even the size of the spanner), and even if we are aiming for a suboptimal dependence on $f$ in the lightness bound, it is still highly challenging to get an  FT-spanner with lightness $o(\log n)$. As discussed above, using cross edges restricted to the input light spanner is not enough to guarantee fault-tolerance. The challenge is to identify (or even just prove the existence of) a set of cross edges that has small lightness on the one hand and that can guarantee fault-tolerance on the other; achieving these two contradictory requirements simultaneously is highly non-trivial.

\begin{figure}[htb]
	\center{\includegraphics[width=0.8\textwidth]{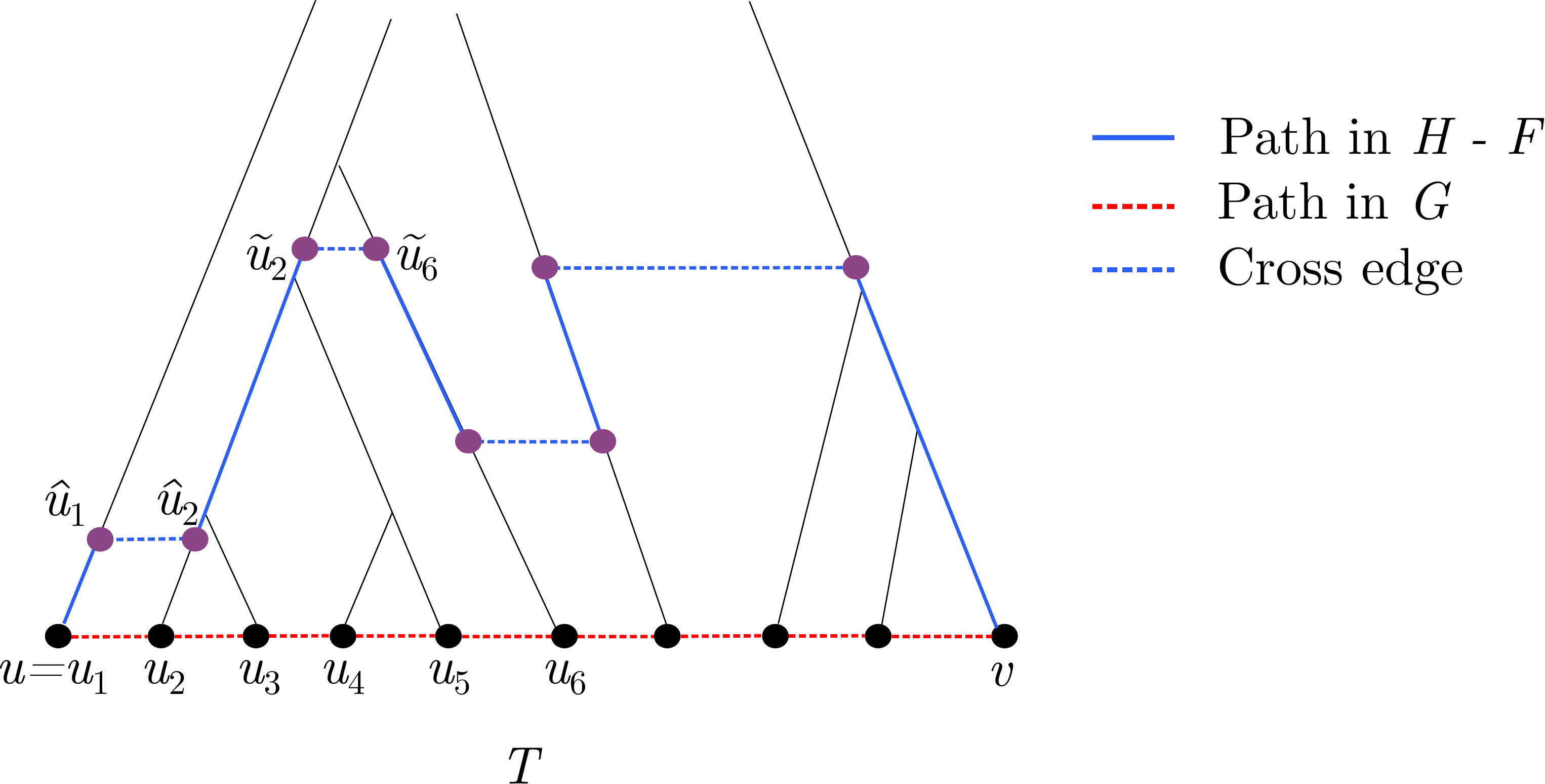}}
	\caption{A path from $u$ to $v$ in $H \setminus F$. Every purple node has $f + 1$ surrogates. We replace some edges and some subpaths of the shortest path from $u$ to $v$ in $G$ by the cross edges between nodes having $f + 1$ surrogates. It might be misleading that $\nd(\tilde{u}_2, \tilde{u}_6)$ is much smaller than $\nd(u_2, u_6)$. Indeed, $\nd(\tilde{u_2}, \tilde{u_6})$ is a good approximation of the length of $(u_2, u_3, u_4, u_5, u_6)$.}
	\label{fig:cross_edge_path}
\end{figure}
Finally, we return to the issue of finding a non-faulty spanner path between any two points. Recall that our idea is to carefully choose some (but not all) cross edges of $T$ and add the edges between the respective surrogates of each chosen cross edge to the VFT spanner. For any two points $u$ and $v$ in $X$, if the lowest good approximation (in terms of its  weight) cross edge $(\ddot{u}, \ddot{v})$ of $(u, v)$  is at a ``low" level, meaning that there are not enough (less than $f + 1$) surrogates for either $\ddot{u}$ or $\ddot{v}$, 
our construction guarantees that $(\ddot{u}, \ddot{v})$ is chosen, and there is an non-faulty edge between $S(\ddot{u})$ and $S(\ddot{v})$. If $\ddot{u}$ and $\ddot{v}$ are at a high enough level, then $(\ddot{u}, \ddot{v})$ might not be chosen. In that case, we follow the 
shortest path $P = \{u= u_1, u_2, u_3, \ldots u_l = v\}$ between $u$ and $v$ in the light spanner $G$ (if there are multiple shortest paths, we fix one arbitrarily).
If $(u_1, u_2)$ is {\em replaceable}, meaning that there is at least one non-faulty edge in the bipartite matching between $S(\hat{u}_1)$ and $S(\hat{u}_2)$, where $\hat{u}_1$ and $\hat{u}_2$ are ancestors of $u_1$ and $u_2$, respectively, at an appropriate level, we replace $(u_1,u_2)$ by a non-faulty edge in the matching. Otherwise, the edge $(u_1, u_2)$ is {\em irreplaceable}, and we find a good approximation cross edge of $(u_1, u_i)$ for some $i > 2$.
Our construction guarantees that we choose a good approximation cross edge $(\tilde{u}_1, \tilde{u}_i)$ of $\dist_P(u_1, u_i)$ for some $i > 2$; {\bf proving this is highly challenging and is the key in our argument for finding a non-faulty spanner path}.    
Then, we replace the prefix $(u_1, u_2, \ldots u_i)$ of $P$ by a non-faulty edge in the bipartite matching between $S(\tilde{u}_1)$ and $S(\tilde{u}_i)$ and continue the process recursively with the remaining suffix of $P$. Finally, the replaced edges can be ``glued together'' by induction via subpaths of total length negligible with respect to the weight of $P$; in this way, we have provided a spanner path from $u = u_1$ to $u_l = v$.
See \Cref{fig:cross_edge_path} for an illustration.

%% file: Preliminaries.tex
\section{Preliminaries}

Given a graph $G$, we denote by $V(G)$ and $E(G)$ the vertex and edge set of $G$, respectively.  Given a metric  $(X, \nd)$, a graph $G = (V, E, \weight)$ is a \emph{geometric graph} in $(X, \nd)$ if $V \subseteq X$ and $\weight(u,v) = \nd(u,v)$ for every edge $(u,v)\in E$.  In this paper, we only consider geometric graphs, so we will refer to points and vertices interchangeably. The weight of a graph $G$, denoted by $\weight(G)$, is the total weight of all edges in $G$. For any set of edges $E'$, $\weight(E') = \sum_{(u, v) \in E'}\nd(u, v)$. For any two vertices $u$ and $v$ in $G$, the distance between $u$ and $v$ in $G$ is denoted by $\dist_G(u, v)$. The minimum spanning tree of $X$, denoted by $\MST(X)$, is the minimum spanning tree of $(X, \binom{X}{2}, \weight)$ with $\binom{X}{2}$ is the set containing all pairs of points in $X$.

We say that $G$ is a \emph{$t$-spanner} of $(X,\nd)$ if $G$ is a geometric graph in $(X, \nd)$ with $X$ being its vertex set and for any two points $u, v \in X$, $\dist_G(u, v) \leq t\nd(u, v)$.  A path between $u$ and $v$ in a geometric graph $G$ in $(X, \nd)$, which might not be a spanner of $(X,\nd)$, is a $t$-spanner path if the total weight of edges of the path is at most $t\nd(u, v)$. A $t$-spanner $G$ of $X$ is an \emph{$f$-vertex-fault-tolerant} ($f$-VFT) if for any set of points $F \subseteq X$ such that $|F| \leq f$, the graph $G - F$ obtained by removing $F$ from $G$ is still a $t$-spanner of $X \setminus F$.  

Let $(X, \nd)$ be a metric of doubling dimension $d$. We denote by $\ball(p, r)$ the set of points in $X$ with distance at most $r$ from $p$. The \emph{spread} of $P$ is the ratio between the maximum distance and minimum distance between points in $P$. We say that $P$  is \emph{$r$-separated} if the distance between every two distinct points in $P$ is at least $r$. The following lemma is well known:
\begin{lemma}[Packing bound]
	\label{lm:packing-bound}
	Let $R \geq r > 0$ and $Y$ be an $r$-separated set contained in a ball of radius $R$. Then, $|Y|\leq \left(\frac{4R}{r}\right)^{d}$.
\end{lemma}

\paragraph{Net.} An $r$-net of a subset of points $P\in X$ is an $r$-separated subset $N$ of $P$ such that for each point $v$ in $P$, there exists $u \in N$ such that $\nd(u, v) \leq r$.

%% file: VFT-optimal.tex
\section{VFT Spanner Construction Algorithm}

In this section, we present a construction  of an $(f + 1)$-VFT $(1 + \eps)$-spanner with degree $O_{\eps, d}(f)$ and lightness $O_{\eps, d}(f^2)$ as described in \Cref{tm1}.  We focus on presenting the ideas of the construction; a fast implementation will be delayed to \Cref{sec:fast}.

\subsection{Net tree, surrogate sets, and bipartite connections}
\label{ssec:net-tree}
By scaling, we assume that the minimum and maximum distance in $X$ is $64$ and $\Delta$, respectively. We set the minimum distance to $64$ to handle some corner cases more gratefully. Throughout this paper, for every positive integer $a$, we use the notation $\log{a} = \log_5{a}$ to avoid the subscript.

\begin{definition}[Greedy Net Tree] \label{def:net-tree} Let $r_0 = 1, N_0 = X$, $\zeta = \lceil \log{\Delta}\rceil$, and $r_i = 5^ir_0$ for every  $i \in [1,\zeta]$. Let  $N_0\supsetneq N_1 \supsetneq \ldots \supsetneq N_{\zeta}$ be a hierachy of nets where  $N_i$ is an $r_i$-net of $N_{i - 1}$. The hierarchy of nets induces a hierarchical tree $T$ where:
	
	\begin{enumerate}[noitemsep]
		\item For every level $i \in [0,\zeta]$, there is one-to-one correspondence between nodes in at level $i$ in  $T$ and points in $N_i$.
		
		\item The parent in $T$ of a point in $N_{i-1}$ is the closest point in $N_i$ for $1\leq i \leq \zeta$, breaking tie arbitrarily. 
	\end{enumerate}
\end{definition}

We refer to points in the nets and nodes in the tree $T$ interchangeably. As a point $u$ can belong to multiple nets, to avoid confusion, we sometimes write $(u,i) \in N_i$ to indicate the copy of $u$ in $N_i$. For each node $(u, i)$, the value $r_i$ is called the radius of $(u, i)$. In our net-tree, we use the radius $r_i = 5^i$ instead of $2^i$ as other works, e.g., \cite{CGMZ16}, because of a specific property that we need: the sum of radii associated with any set of nodes such that no two of them are ancestors of each other is bounded by a constant time the weight of the minimum spanning tree; see \Cref{lm:sum-disjoint-nodes}.

For a node $x = (u, i)$, we sometimes write $x$ in place of $u$ in, e.g., the distance function. For example $\delta(x,y)$ for two nodes $x = (u, i)$ and $y = (v, j)$ refers to $\delta(u, v)$. 
A point $w$ is a leaf of $x$ if $(u, 0)$ is a leaf of the subtree with root $x$. We use the notation $\lvl(x)$ for the level of $x$. (If $x = (u, i)$ then $\lvl(x) = i$.)   

For each node $x = (u, i)$, let $L(x)$ or $L(u, i)$ be the set of leaves of $x$. The distance between any node to any of its descendants is at most a constant times the node's radius. 

\begin{claim}
	\label{clm:leaf-dist}
	Let $(u, i)$ be a node in $T$ and $(v, j)$ be a descendant of $(u, i)$ ($j < j$), then $\nd(u, v) \leq 5r_i/4$, implying that $\nd(u, v) \leq 2r_i$.
\end{claim}

The second bound $\nd(u, v) \leq 2r_i$ is usually used when we want to work with integers and the distance between a node to one of its leaves does not contribute significantly to the total distance. 

\begin{proof}[Proof of \Cref{clm:leaf-dist}]
	By \Cref{def:net-tree}, the distance between a node at level $i$ to its child is at most $r_i$. Hence, $\nd(u, v) \leq \sum_{k = j + 1}^ir_k$. Using geometric sum,
	\begin{equation}
		\nd(u, v) \leq \sum_{k = j + 1}^ir_k \leq \sum_{k = 0}^ir_k = r_i(1 + 1/5 + 1/5^2 + \ldots) \leq 5r_i/4,
	\end{equation}
	as claimed.
\end{proof}

A net tree is an important tool in almost all spanner constructions in doubling metric, e.g., see \cite{CGMZ16}. In the construction of $f$-VT $(1+\eps)$-spanner, Solomon~\cite{Sol14} introduced the notion of \emph{surrogate sets}, which was used on top of the net-tree spanners. 

\begin{definition}[Surrogate Sets] \label{def:sorrogate} Each node $x\in T$ at level $i\in [0,\zeta]$  is associated with a set of points $S(x)\in X\subseteq \ball(x, 16\cdot r_i)$ of size at most $f+1$ called a \emph{surrogate set}.
\end{definition}

An useful operation that we will use on top of the surrogate sets is forming a bipartite connection.  

\begin{definition}[Bipartite Connection]\label{def:bipartite-conn}Let $S(x)$ and $S(y)$ be the surrogate sets of two differrent nodes $x$ and $y$ in $T$, a  \emph{bipartite connection} between them is a set of edges, denoted by $M(S(x), S(y))$,  that is defined as folllows. If  $\min\{|S(x)|, |S(y)||\} < f + 1$ then $M(S(x), S(y))$ includes all edges between $S(x)$  and $S(y)$, i.e, $M(S(x), S(y))= S(x) \times S(y)$.  Otherwise, by \Cref{def:sorrogate}, $|S(x)| = S(y) = f+1$ and in this case, $M(S(x), S(y))$ is an (arbitrary) perfect matching between $S(x)$ and $S(y)$. 
\end{definition}

\subsection{The construction algorithm}

Our construction is fairly simple compared to existing fault-tolerant spanner constructions. Let  $\lambda = 5^{20}(1+1/\eps)$.   Let $x,y$ be two nodes of $T$ at the same level $i_{xy}$. We say that $(x,y)$ is a \emph{cross edge} if $\delta(x,y)\leq \lambda r_{i_{uv}}$. The notion of cross edges is central in all spanner constructions based on the net tree. If $(x,y)$ is a cross edge, we say that $x$ and $y$ are \emph{cross neighbors} of each other. Let $NC(x)$ be the set of cross neighbors of $x$, and $NC[x] = NC(x)\cup\{x\}$. 

Let $A$ be a set of nodes at some level $i$ of $T$, we denote by $\Cr(A)$ be the set of cross edges between nodes in $A$. For every node $x\in T$, let $T(x)$ be the subtree rooted at $x$ of $T$. Let $i_x$ be the level of $x$ and $j \in [0,\zeta]$, we define  $\Aug_j(x)$ as follows: if $j < i_{x}$, then $\Aug_j(x)$ is the union of $\Cr(NC[y])$ for all \emph{descendant nodes} $y$ at level $j$ of $x$; otherwise, $\Aug_j(x) = \Cr(NC[y])$ where $y$ is the ancestor at level $j$ of $x$.  We view $\Aug_j(x)$ as a set of \emph{augmented cross edges} at level $j$ of $x$. Let $\Aug(x, l, h) = \bigcup_{j = i_x + l}^{i_x + h}\Aug_j(x)$.

Let $(u,0)$ and $(v,0)$ be two nodes at level $0$ of $T$. We say that $(\hat{u},\hat{v})$ is an \emph{original cross edge} of $(u,0)$ and $(v,0)$ if $\hat{u}$ and $\hat{v}$ are two ancestors at \emph{lowest level}, say $i_{uv}$,  of $(u,0)$ and $(v,0)$ respectively such that  $\delta(\hat{u},\hat{v}) \leq  \lambda r_{uv}$. It could be that $\hat{u} = (u,0)$ and $\hat{v} = (v,0)$.

\begin{algorithm}[!htb]
	\caption{Spanner Construction}\label{alg:VFT-spanner}
	\KwIn{Doubling metric $(X,\nd)$ of dimension $d$}
	\KwOut{An $(f + 1)$-FT $(1 + 5\eps)$-spanner $H$ with degree $O_{\eps, d}(f)$ and lightness $O_{\eps, d}(f^2)$}
	Construct a  light $(1 + \eps)$-spanner $G$ of $X$ and a net-tree $T$  \label{line:net-tree}\;
	$E^* \gets \emptyset, H \gets (X, \emptyset), \lambda = 5^{20}(1+1/\eps)$\;
	\For{$e = (u, v) \in E(G)$}{
		\label{line:begin-aug-original}
		Let $(\hat{u},\hat{v})$ be the \emph{original cross edge }of $(u,0)$ and $(v,0)$\;
		$E^* \gets (\hat{u},\hat{v})$ \label{line:add-origin-CE}\;
		\For{$(x, y) \in \Aug(\hat{u}, 0, 5\log{\lambda}) \cup \Aug(\hat{v}, 0, 5\log{\lambda})$ \label{line:augm}}{
			Let $i_{xy}$ be the level of $x$ (and $y$) in $T$\;
			\If{$\nd(x, y) \geq 64 \cdot r_{i_{xy}}$ \label{line:check-length}}{
				$E^* \gets E^* \cup \{(x, y)\}$\label{line:Estar-update}\;
			}
		}
	}
	\label{line:end-aug-original}
	$\xi = \eps^{-O(d)}, c_1 = 50\log{\lambda}\cdot \xi, c_2 = 51\log{\lambda}\cdot \xi, c_3 = 55\log{\lambda}\cdot \xi$ \label{line:costants-c}\; 
	\For{$i \in [0\ldots \zeta - 1]$ \label{line:level-iter}}{
		\For{each node $x$ at level $i$ of $T$\label{line:begin-add-incomplete}}{
			Mark $x$ as \emph{small} if every leaf of $x$ has degree at most $c_1 \cdot f$ in $H$ \label{line:mark-x-small}\;
			$D_x \leftarrow$ set of all descendants of $x$ at levels from $i - 5\log{\lambda}$ to $i$\label{line:Dx-def}\;
			\If{there exists $w \in D_x$ s.t $w$ is small and $T(w)$ has at most $f$ leaves\label{line:check-Dx}}{
				$E^* \gets E^* \cup \{(y, z)~|~y, z \in NC[x] ~\text{and}~ 64r_i \leq \nd(y, z) \leq \lambda r_i\}$ \label{line:add-NCx}\;
			}
		}
		\label{line:end-add-incomplete}
		Mark all vertices with degree in $H$ larger than $c_3f$ as \emph{saturated}\label{line:mark-saturated}\;
		$E^*_i \leftarrow$ set of edges $(x,y)\in E^*$ s.t both $x$ and $y$ are at level $i$\label{line:Eistar}\;
		\For { each edge $(x,y)\in E^*_i$ \label{line:edge-process}}{
			$S(x) \leftarrow \text{SelectSurrogate}(x)$ \label{line:Sxdef}\;	
			$S(y) \leftarrow \text{SelectSurrogate}(y)$ \label{line:Sydef}\;
			$E(H) \leftarrow E(H)\cup M(S(x), S(y))$ \label{line:H-update}\;
		}
	}
	\Return $H$\;
\end{algorithm}

\Cref{alg:VFT-spanner} describes the construction: it can be divided into two phases: Phase 1 is from \crefrange{line:net-tree}{line:Estar-update} and Phase 2 from \crefrange{line:costants-c}{line:H-update}.   

\begin{algorithm}[!htb]
	\caption{SelectSurrogate}\label{alg:Surrogate}
	\KwIn{A node $x \in V(T)$}
	\KwOut{A surrogate set $S(x)$}
	\eIf{$x$ is small \label{line:check-small}}{	
		$S(x)\leftarrow$ arbitrary $f+1$ leaves of $T(x)$ if $T(x)$ has at least $f+1$ leaves; otherwise $S(x)\leftarrow $ all leaves of $T(x)$\label{line:add-leaf-S}\;}
	{
		$S' \leftarrow \{u \in X\cap \ball(x,16r_i): \deg_H(u) \geq c_2 \cdot f \mbox{ and $u$ is not saturated}\}$\label{line:def-Sprime}\;
		$S'' \leftarrow \{u \in X\cap \ball(x,4r_i): \deg_H(u) <  c_2 \cdot f \}$\label{line:def-Sdprime}\;
		\eIf{$|S'| \geq f+1$}{
			$S(x)\leftarrow$  arbitrary $f+1$ vertices in $S'$\;
		}{
			$S(x)\leftarrow $ $S'\cup \{\mbox{arbitrary $f+1 - |S'|$ vertices in $S''$}\}$\label{line:general-surrogate}\;    
		}
	}
	\Return $S(x)$\;
\end{algorithm}

In the first phase, we start with a $(1+\eps)$-spanner $G$ of $(X,\nd)$ with lightness $O(\eps^{-O(d)})$ that can be constructed in $O(n\log n)$ time~\cite{FS16}.  The edges of $G$ will serve as guidance for our construction, as we will later bound the lightness of the output VFT spanner $H$ by charging to edges of $G$. The goal of this phase is to construct a set of cross edges $E^*$. Unlike other (both fault-tolerant and non-fault-tolerant) spanner constructions~\cite{Sol14,CGMZ16},  where one would add a cross edge for \emph{every pair of points in $X$}, we only add cross edges corresponding to edges of $G$. Specifically for each edge $(u,v)\in E(G)$, we first add to $E^*$ the original cross edge, say $(\hat{u},\hat{v})$ of $(u,0)$ and $(v,0)$ (\cref{line:add-origin-CE}).  It is not hard to see that $\delta(\hat{u},\hat{v})$ is approximately $\delta(u,v)$. Next, we add to $E^*$ the augmented cross edges from the ancestors that are within $O(\log\lambda)$ levels from $\hat{u}$ and $\hat{v}$. Note that: (i) we only add edges that are long enough (\cref{line:check-length}) and (ii)  by definition of $\Aug(\cdot,\cdot,\cdot)$,  $E^*$ includes not only cross edges incident to the ancestors of $\hat{u}$ and $\hat{v}$, but also those that are between the cross neighbors of the ancestors.  As agumented cross edges are not much longer than $(\hat{u},\hat{v})$, we can later show that $\weight(E^*) = \eps^{-O(d)}\weight(G)$, which is a part of the proof of \Cref{lm:weight_E_O}.

In the second phase, we use edges in $E^*$ found in the first phase to add edges to the spanner $H$ to guarantee the fault-tolerant property. Specifically, we visit the levels of $T$ from lower to higher, and for each edge $(x,y)$ in $E^*$ at level $i$ (set $E^*_i$ in \cref{line:Eistar}), we would add a bipartite connection (\Cref{def:bipartite-conn}) between two surrogate sets $S(x)$ and $S(y)$ as in \cref{line:H-update}. (Suppose for now that $S(x)$ and $S(y)$ are chosen arbitrarily following \Cref{def:sorrogate}). However, this is not enough. In particular, the nodes $x$ that are marked as small in \cref{line:mark-x-small} are problematic. That is if there exists a small descendant within $O(\log \lambda)$ levels of $x$ that is small and has at most $f$ leaves in its subtree (\cref{line:Dx-def} and \cref{line:check-Dx}), then we may not able to find $f + 1$ vertices for $S(x)$. (We later prove that the set of small nodes, each of which has at most $f$ leaves, can be partitioned into LNF.) To fix this issue, we add (long enough) cross edges between nodes in $NC[x]$ in \cref{line:add-NCx} to $E^*$. Using the argument outlined in \Cref{techhighlight}, which is the key technical contribution of our work, we are able to show that adding the bipartite connection between two surrogate sets $S(x)$ and $S(y)$ for every edge $(x,y)\in E^*$ at level $i$ suffices to guarantee $f$-VFT, and furthermore, the spanner will have lightness $O(f^2)$. However, the degree could be $\Omega(n)$ if  $S(x)$ and $S(y)$ are chosen arbitrarily.

\Cref{alg:Surrogate} shows how to choose $S(x)$ carefully to reduce the degree all the way down to $O_{\eps,d}(f)$ while keeping the fault tolerant property in check. If the degree in $H$ of a node is at least $c_3\cdot f$, the algorithm will mark it as \emph{saturated} in \cref{line:mark-saturated}, and it will not be used in any surrogate set of the nodes considered in future iterations. (Only edges incident to points in surrogate sets are added to $H$.)  If $x$ is small (\cref{line:check-small}), then one can show that every leaf of $T(x)$ is not saturated by definition and hence can safely be added to $S$ (\cref{line:add-leaf-S}). Otherwise, we consider two sets $S'\subseteq \ball(x,16r_i)$ (\cref{line:def-Sprime}) and $S'' \subseteq \ball(x,4r_i)$ (\cref{line:def-Sdprime}) containing unsaturated vertices . We prefer adding vertices of $S'$ to $S(x)$ over vertices of $S''$ as those in $S'$ are closer to being saturated. While it is not hard to see that our final spanner has degree $O_{\eps,d}(f)$ due to the choice of the surrogate sets,  showing that it remains $f$-VFT is extremely challenging. Indeed, this is another major technical contribution of our paper. 

The main result in our paper is that  \Cref{alg:VFT-spanner} returns a light and bounded degree VFT spanner in optimal time.

\begin{theorem}
	\label{thm:VFT-correctness}
	The output graph $H$ of \Cref{alg:VFT-spanner} is a $f$-VFT $(1 + \eps)$-spanner of $X$ with maximum degree $2c_3f$ and lightness $O_{\eps, d}(f^2)$. Furthermore, \Cref{alg:VFT-spanner} can be implemented to run in $O(n(\log n + f))$ times. 
\end{theorem}

We show that $H$ has maximum degree $2c_3f$ in \Cref{sec:degree}, $O(f^2)$ lightness in \Cref{sec:light} and prove that $H$ is $f$-vertex fault tolerant in \Cref{sec:connectivity}. All together prove \Cref{thm:VFT-correctness}. We construct a fast implementation of our algorithm in \Cref{sec:fast}.

%% file: Degree-Analysis.tex
\section{Degree Analysis}
\label{sec:degree}
Observe from \Cref{alg:VFT-spanner}'s for loop (\cref{line:level-iter}) that before iteration $i$, each non saturated point has a degree less than $c_3f$. After being saturated, the degree of a point does not increase. To show that $H$ has a bounded maximum degree, it remains to prove that at the last iteration before any point $u$ become saturated, the degree of $u$ increases by at most $O_{\eps, d}(f)$. Indeed, we prove a stronger result: the degree of each point increases by at most $O_{\eps, d}(f)$ after each iteration of the for loop in \cref{line:level-iter}.

Since after adding a complete bipartite connection to $H$, the degree of each point increases by at most $f$, we need to show that any point $u$ is in a constant number of surrogate sets (with multiplicity) at level $i$. Note that $u$ can belong to the surrogate set of a node $x$ multiple times. An edge in $H$ is a level-$i$ edge for a non-negative integer $i \leq \zeta$ if it is in $M(S(x), S(y))$ for some $x, y \in V(T)$. If an edge is added to $H$ at multiple levels, we choose the lowest one. 

\begin{lemma}
	\label{lm:bounded-surrogate}
	Every point $v \in X$ belongs to $O_{\eps, d}(1)$ surrogate sets (with multiplicity) at level $i$ for every $i \in [0\ldots \zeta]$. Furthermore, $v$ is incident to at most $\xi f$ level-$i$ edges in $H$ with $\xi = O(1/\eps)^d$ chosen in \cref{line:costants-c}.
\end{lemma}

To prove \Cref{lm:bounded-surrogate}, we first show that any node in $V(T)$ is only incident to $O_{\eps, d}(1)$ edges in $E^*$.

\begin{observation}
	\label{obs:bound-node-deg}
	Every node $(u, i)$ is incident to $O(1/\eps)^d$ cross edges in $E^*$.
\end{observation}

\begin{proof}
	Let $N(u, i) = \{v_1 , v_2, \ldots \}$ be the set of cross neighbors of $(u, i)$. By \Cref{alg:VFT-spanner}, for every $v_j \in N(u, i)$, $\nd(v_j, u) \leq \lambda r_i$, implying that $v_j \in \ball(u, \lambda r_i)$. Additionally, $N(u, i)$ is an $r_i$-separated set since it is a subset of an $r_i$-net. Hence, by packing bound (\Cref{lm:packing-bound}), $|N(u, i)| \leq \left(\frac{4\lambda r_i}{r_i}\right)^d = O(1/\eps)^d$. 
\end{proof} 

We now prove \Cref{lm:bounded-surrogate}.

\begin{proof}[Proof of \Cref{lm:bounded-surrogate}]
	For any node $(u, i)$ such that $v$ is in the surrogate set of $(u, i)$, $\nd(u, v) \leq 16r_i$, implying that $u \in \ball(v, 16r_i)$. By \Cref{obs:bound-node-deg}, for each $u$, $v$ is chosen to be a surrogate of $(u, i)$ for $O(1/\eps)^d$ times. Let $S_v = \{u_1, u_2, \ldots \}$ be the set of points in $X$ satisfying for each $u_k\in S_v$, $v$ is a surrogate of $(u_k, i)$. By \Cref{alg:Surrogate}, $\nd(v, u_k) \leq 16r_i$ for every $u_k \in S_v$. Since $S_v$ is a subset of an $r_i$-net, $|S_v| \leq \left(\frac{4 \times 16r_i}{r_i}\right)^d = 2^{6d}$ by \Cref{lm:packing-bound}. 
	
	Because each node $(u_k, i)$ ($u_k \in S_v$) is incident to $O(1/\eps)^d$ cross edges and adding the bipartite connection of each cross edge increases the degree of a point in $S(u_k, i)$ by at most $f$ in \cref{line:H-update}, there are $2^{6d}\cdot O(1 /\eps)^d \cdot f$ level-$i$ edges incident to $v$.
\end{proof}

We are now ready to bound the degree of $H$.

\begin{lemma}
	\label{lm:bound-deg-H}
	$H$ has a maximimum degree bounded by $2c_3f$.
\end{lemma}

\begin{proof}
	Observe that when a point is saturated, it will never be used later in the algorithm. Thus, it is sufficient to bound the degree of any point after the last iteration when it is not marked as saturated in \cref{line:mark-saturated} of \Cref{alg:VFT-spanner}. For each point $u$, let $i_u$ be the highest level such that $u$ is not saturated. By \Cref{lm:bounded-surrogate}, the degree of $u$ in $H$ increases by at most $\xi f$ after iteration $i_u$. Since $u$ is not marked as saturated in \cref{line:mark-saturated}, the degree of $u$ before we add the bipartite connection of cross edges at level $i_u$ is at most $c_3 f$. Hence, the degree of $u$ after we add level-$i_u$ edges to $H$ is at most $(c_3 + \xi)f \leq 2c_3f$ by the choice of $c_3$.
\end{proof}

%% file: Lightness.tex
\section{Lightness}
\label{sec:light}
In this section, we analyze the lightness of $H$. Recall that for each edge $(u, v)$ in $G$, we add the original cross edge $(\hat{u}, \hat{v})$ of $(u, v)$ and the cross edges in $\Aug(\hat{u}, 0, 5\log{\lambda}) \cup \Aug(\hat{v}, 0, 5\log{\lambda})$ to $E^*$; let $E^*_O$ be the set contains all these cross edges for every $(u, v) \in E(G)$. Let $E_O$ be the set of edges added to $H$ from every bipartite connection of all cross edges in $E^*_O$. Formally,

\begin{equation}
	\label{def:E-O}
	E_O = \{M(S(x), S(y)) : (x, y) \in E^*_O\} \qquad .
\end{equation}

In \Cref{lm:weight_E_O} , we show that $\weight(E_O) = O(f^2)\weight(\MST(X))$. In \Cref{lm:weight_E-E_O}, we show that $\weight(E(H) \setminus E_O) = O(f^2)\weight(\MST(X))$. Given the two lemmas, we now prove: 

\begin{lemma}
	\label{lm:light}
	$H$ has lightness $\eps^{-O(d)}f^2$.
\end{lemma}

\begin{proof}
	Observe that $\weight(E(H)) \leq \weight(E_O) + \weight(E(H) \setminus E_O)$. Since  $\weight(E_O) = \eps^{-O(d)}\weight(\MST(X))$ by \Cref{lm:weight_E_O} and $\weight(E(H) \setminus E_O) = \eps^{-O(d)}\weight(\MST(X))$ by \Cref{lm:weight_E-E_O}, $\weight(E(H)) \leq \eps^{-O(d)}\weight(\MST(X))$, which implies the lemma.
\end{proof}

First, we list some properties that we need for the proofs of \Cref{lm:weight_E_O} and \Cref{lm:weight_E-E_O}.

\begin{property} 
	\label{prop:lightness-properties}
	We have the following properties: 
	\begin{enumerate}
		\item \label{it:connection-small} Let $(x, y)$ be a cross edge in $E^*$. For every $(u, v) \in M(S(x), S(y))$, $\nd(u, v) \leq 2\nd(x, y)$.
		\item \label{it:bounded-up-aug} Let $x$ be an arbitrary level-$i$ node in $T$. For every non-negative integer $\gamma$, $\weight(\Aug(x, 0, \gamma)) = \eps^{-O(d)} 5^\gamma r_i$. Recall that $\Aug(x, 0, \gamma) = \bigcup_{j = i}^{i + \gamma}\Aug_j(x)$ with $\Aug_j(x)$ ($j \geq 0$) is the set of cross edges between cross neighbors of the ancestor of $x$ at level $j$.
		\item \label{it:original-weight} For every original cross edge $(\hat{u}, \hat{v})$ at level $i$, $\nd(\hat{u}, \hat{v}) \geq \lambda r_i / 6$.
		\item \label{it:approx-weight} Let $e = (u, v)$ be a pair of points in $X$ and $(\hat{u}, \hat{v})$ be the original cross edge of $e$. Then, $\nd(\hat{u} , \hat{v}) \leq (1 + \eps)\nd(u, v)$.
	\end{enumerate}
\end{property}

\begin{proof}
	\textit{\Cref{it:connection-small}: } Let $i$ be the level of $x$ and $y$. Note that $\nd(x, u), \nd(y, v) \leq 16r_i$ by \Cref{def:sorrogate}. By triangle inequality, $\nd(u, v) \leq \nd(x, y) + \nd(x, u) + \nd(y, v) \leq \nd(x, y) + 32r_i \leq 2\nd(x, y)$ since $\nd(x, y) \geq 64r_i$ by the choice of cross edges in $E^*$.
	
	\textit{\Cref{it:bounded-up-aug}: } Let $k$ be a level such that $k \geq i$. The weight of a cross edge at level $k$ is at most $\lambda r_k = \lambda r_i \cdot 5^{k - i}$ by the definition of cross edges. Let $x_k$ be the ancestor of $x$ at level $k$. Recall that $\Aug_k(x)$ is the set of cross edges with both ends in $NC[x_k]$. Since $NC[x_k]$ is a subset of a $r_k$-net with diameter $2\lambda r_k$, $|NC[x_k]| \leq (4\lambda)^d = \eps^{-O(d)}$ by packing bound (\Cref{lm:packing-bound}). Then, $|\Aug_k(x)| \leq |NC[x_k]| ^2 = \eps^{-O(d)}$. Hence,
	
	\begin{equation*}
		\weight(\Aug(x, 0, \gamma)) = \sum_{k = i}^{i + \gamma} \weight(\Aug_k(x)) = \eps^{-O(d)} \lambda r_i \cdot \sum_{k = i}^{i + \gamma}5^{k - i} \leq 5^\gamma \eps^{-O(d)} r_i \qquad,
	\end{equation*}
	as claimed.
	
	\textit{\Cref{it:original-weight}: } Let $(u, v)$ be the pair of points in $X$ of which $(\hat{u}, \hat{v})$ is the original cross edge. Let $(u', i - 1)$ and $(v', i - 1)$ be the ancestors at level $i - 1$ of $(u, 0)$ and $(v, 0)$, respectively. Since $\hat{u}$ and $\hat{v}$ are parents of $(u', i - 1)$ and $(v', i - 1)$, $\nd(\hat{u}, u'), \nd(\hat{v}, v') \leq r_i$. By the minimality of $i$, $\nd(u', v') > \lambda r_{i - 1}$. Using the triangle inequality, we have: 
	\begin{equation*}
		\nd(\hat{u}, \hat{v}) \geq \nd(u', v') - (\nd(\hat{u}, u') + \nd(\hat{v}, v')) \geq \lambda r_{i - 1} - 2r_i \geq {(\lambda /  5 - 2)}r_i > \lambda r_i/6 \qquad,
	\end{equation*}
	as $\lambda = 5^{20}(1 + \eps^{-1})$.
	
	\textit{\Cref{it:approx-weight}: } Let $i$ be the level of $(\hat{u}, \hat{v})$. By triangle inequality, $\nd(\hat{u}, \hat{v}) \leq \nd(u, v) + \nd(\hat{u}, u) + \nd(\hat{v}, v)$. Recall from \Cref{clm:leaf-dist} that $\nd(\hat{u}, u), \nd(\hat{v}, v) \leq 5r_i/4$. Therefore,
	\begin{equation*}
		\frac{\nd(u, v)}{\nd(\hat{u}, \hat{v})} \geq \frac{\nd(\hat{u}, \hat{v}) - 5r_i/4 - 5r_i/4}{\nd(\hat{u}, \hat{v})} \geq 1 - \frac{5r_i}{2\nd(\hat{u}, \hat{v})} \geq 1 - \frac{15}{\lambda} \qquad,
	\end{equation*}
	since $\nd(\hat{u}, \hat{v}) \geq \lambda r_i / 6$ by \Cref{it:approx-weight}. Hence, $\nd(\hat{u}, \hat{v}) \leq \frac{\nd(u, v)}{1 - 15/\lambda} \leq (1 + \eps)\nd(u, v)$.
\end{proof}

\subsection{Weight of edges due to light spanner}
\begin{lemma}
	\label{lm:weight_E_O}
	$\weight(E_O) = \eps^{-O(d)}f^2\weight(\MST(X))$.
\end{lemma}

\begin{proof}
	For each cross edge $(x, y) \in E^*_O$, the bipartite connection $M(S(x), S(y))$ contains at most $(f + 1)^2$ edges, each of those has weight at most $2\nd(x, y)$ by \Cref{it:connection-small} of \Cref{prop:lightness-properties}. Hence, by the definition of $E_O$ in \Cref{def:E-O}, $\weight(E_O) \leq 2(f + 1)^2\weight(E^*_O)$. It remains to bound $\weight(E^*_O)$.
	
	Let $O$ be the set of original cross edges of edges in $E(G)$. By \Cref{it:approx-weight} of \Cref{prop:lightness-properties}, $\weight(O) \leq (1 + \eps)\weight(G)$. Observe that:
	\begin{equation*}
		E^*_O \subseteq \bigcup_{(\hat{u}, \hat{v}) \in O}\left[\{(\hat{u}, \hat{v})\} \cup \Aug(\hat{u}, 0, 5\log{\lambda}) \cup \Aug(\hat{v}, 0, 5\log{\lambda})\right]
	\end{equation*}
	Recall that for each node $x$, $\lvl(x)$ is the level of $x$. We have:
	\begin{equation*}
		\begin{split}
			\weight(E^*_O) &\leq \sum_{(\hat{u}, \hat{v}) \in O}(\nd(\hat{u}, \hat{v}) + \weight(\Aug(\hat{u}, 0, 5\log{\lambda}) + \weight(\Aug(\hat{v}, 0, 5\log{\lambda}))))\\
			&\leq \sum_{(\hat{u}, \hat{v}) \in O}(\nd(\hat{u}, \hat{v}) + 2 \cdot 5^{5\log\lambda}\eps^{-O(d)}r_{\lvl(\hat{u})}) \qquad\text{(by \Cref{it:bounded-up-aug} of \Cref{prop:lightness-properties})}\\
			&\leq  \sum_{(\hat{u}, \hat{v}) \in O}\left(\nd(\hat{u}, \hat{v}) + 2 \cdot \lambda^{5}\eps^{-O(d)}\frac{\nd(\hat{u}, \hat{v})}{\lambda / 4}\right) \qquad\text{(by \Cref{it:original-weight} of \Cref{prop:lightness-properties})}\\
			&\leq \sum_{(\hat{u}, \hat{v}) \in O}\eps^{-O(d)}\nd(\hat{u}, \hat{v}) = \eps^{-O(d)}\weight(O)\\ 
			&\leq \eps^{-O(d)}\weight(G) \leq \eps^{-O(d)}\weight(\MST(X)) \qquad\text{(since $G$ is a light spanner of $X$)},
		\end{split}
	\end{equation*}
	as desired. 
\end{proof}

\subsection{Weight of remaining edges}

Throughout this section, we show that: 

\begin{lemma}
	\label{lm:weight_E-E_O}
	$\weight(E(H) \setminus E_O) = \eps^{-O(d)}f^2\weight(\MST(X))$.
\end{lemma} 

We partition $E(H) \setminus E_O$ into two sets $E_{inc}$ and $E_{com}$, whose formal definition will be given later. We then bound the total weight of each set in \Cref{lm:Einc-weight} and \Cref{lm:Ecom-weight}. Recall that each edge in $E(H) \setminus E_O$ is added to $H$ by some bipartite connection of cross edges in $E^* \setminus E^*_O$. Those cross edges are added to $E^*$ by \cref{line:check-Dx}--\ref{line:add-NCx}. 

To formally define $E_{inc}$ and $E_{com}$, we need more notation. A node is \emph{large} if it is not small, i.e., there exists a leaf in its subtree with degree at least $c_1f$.  A node is \emph{incomplete} if it is small and has at most $f$ leaves, otherwise it is \emph{complete}. That is a node is complete if it is either large or has at least $f + 1$ leaves. By \cref{line:check-Dx}--\ref{line:add-NCx}, if a node $x$ is incomplete or has an incomplete descendant within $5\log{\lambda}$ levels, we add all long enough cross edges from $\Cr(NC[x])$ to $E^*$. Recall that $\Cr(NC[x])$ is the set of cross edges with both ends in $NC[x]$. Hence, for each cross edge $(y, z)$ in $E^* \setminus E^*_O$, both $y$ and $z$ are cross neighbors of some node $x$ such that $x$ is either incomplete or is an ancestor within $5\log\lambda$ level of an incomplete node. Let $E^*_{com}$ be the set of cross edges in $E^* \setminus E^*_O$ having both complete end nodes. More formally, $E^*_{com} = \{(x, y) ~|~ (x, y) \in E^* \setminus E^*_O ~\text{and $x, y$ are complete}\}$. Let $E^*_{inc} = E^* \setminus (E^*_O \cup E^*_{com})$, meaning that $E^*_{inc}$ contains every cross edge in $E^* \setminus E^*_O$ that has at least one incomplete end node. We denote by $E_{com}$ and $E_{inc}$ the sets of edges added to $H$ by the bipartite connection of edges in $E^*_{com}$ and $E^*_{inc}$. Formally,

\begin{equation}
	\label{def:E-com}
	E_{com} = \{M(S(x), S(y)) : (x, y) \in E^*_{com}\} \quad \text{and}\quad E_{inc} = \{M(S(x), S(y)) : (x, y) \in E^*_{inc}\}.
\end{equation}

The key to our proof is the following lemma:

\begin{lemma}
	\label{lm:complete-node-complete}
	If a node $x$ in $T$ is complete, $S(x)$ always has size $f + 1$. 
\end{lemma}

While the statement of \Cref{lm:complete-node-complete} is simple and easy to understand, the proof of \Cref{lm:complete-node-complete} is intricate. Here, we sketch the ideas of our proof. The full proof is deferred to \Cref{appen:p1}. If $x$ is small, we know that the subtree at root $x$ has at least $f + 1$ leaves with a low degree. Therefore, there are always $f + 1$ points with a low degree to be added to $S(x)$. If $x$ is large, then one of its leaves, say $v$, is incident to many edges. We keep track of the degree change of points close to $v$. By \Cref{lm:bounded-surrogate}, the degree of $v$ in level $i - 2\log{\lambda}$ is still significantly large (at least $48\log{\lambda} \cdot \xi$), then for any edge $(v, w)$ incident to a point close to $v$, all the low degree points close to $w$ will be in $\ball(x, 4r_i)$. We prove after the iterations from $i - 2\log{\lambda}$ to $i - 1$, there are still $4f + 4$ points with degree less than $c_2f$ in $\ball(x, 4r_i)$. For this to hold, we have to choose the surogates carefully, which explains the choice of points with high and low degrees in \Cref{alg:Surrogate}. Hence, there are always enough candidates for $S(x)$, which will prove \Cref{lm:complete-node-complete}.

\Cref{lm:complete-node-complete} contains a key insight for our algorithm about which edges could be chosen in addition to the original cross edges of (edges in) $G$. We call those subtrees a light net-forest (LNF). 

\begin{definition}
	\label{def:lnf}
	A light net-forest of $T$ with respect to $G$, denoted by $\mathrm{LNF}_G(T)$, is a subgraph of $T$ whose vertex set containing all incomplete nodes in $T$ and edge set is the set of edges in $T$ between any two incomplete nodes. If the net-tree $T$ and the light spanner $G$ are clear in the context, we use the notation $\mathrm{LNF}$. 
\end{definition}

By the definition of $E^*_{inc}$, every cross edge in $E^*_{inc}$ is incident to at least one node in LNF. Then, we will bound $\weight(E^*_{inc})$ by $\weight(\mathrm{LNF})$. To do that, we first claim that LNF is a set of (rooted) subtrees of $T$.

\begin{observation}
	\label{obs:all-leaves}
	Every leaf node of $T$ is in $\mathrm{LNF}$.
\end{observation}

\begin{proof}
	By definition of a small node, every leaf node $(u, 0)$ is small since the degree of every $u$ before the first iteration is $0 < c_1f$. Since $(u, 0)$ has only one leaf in its subtree (which is itself), $(u, 0)$ is incomplete. 
\end{proof}

The following claim implies that the ancestor of every almost complete node is complete.

\begin{claim}
	\label{clm:parent-complete}
	The parent of a complete node is complete. 
\end{claim} 

\begin{proof}
	For each point $v \in X$ and each level $i \in [0, \zeta]$, let $\deg_{<i}(v)$ be the degree of $v$ before $i^{th}$ iteration. Let $x$ be a complete node and $p$ be the parent of $x$. If $x$ is large, then by definition of a large node, there exists a leaf $v$ of $T(x)$ such that $\deg_{< i}(v) \geq c_1 f$. Therefore, since $\deg_{< i}(v)$ is non-decreasing as level $i$ gets larger, $p$ also has a leaf $v$ of degree at least $c_1f$, implying that $p$ is large and thus complete. 
	
	If $x$ is small, then $T(x)$ has at least $f + 1$ leaves. If $p$ is small, $T(p)$ also has at least $f + 1$ leaves since $p$ is the parent of $x$, implying that $p$ is complete by definition. If $p$ is large then $p$ is also complete. 
\end{proof}

From \Cref{clm:parent-complete}, we obtain that the LNF contains a set of node-disjoint subtrees of $T$. Hence, we will bound the weight of cross edges incident to nodes in LNF by the total radius of the roots, which are called almost complete nodes. Formally, an \emph{almost complete} node is an incomplete node whose parent is complete. From \Cref{clm:parent-complete}, any two almost complete nodes do not have the ancestor-descendant relationship. Thus, there is no almost complete node that is an ancestor/descendant of another almost complete node, meaning that the subtrees in LNF are node-disjoint.

By a relatively simple argument, one can show that the total radius of the root nodes of the LNF  is a good approximation of the weight of $\MST(X)$. Note that \Cref{lm:sum-disjoint-nodes} is not true for arbitrary construction of the net-tree $T$. We need the property that our net-tree is constructed by a greedy method, meaning the parent of a net-point at level $i$ is its closest net-point at level $i + 1$. 

\begin{lemma}
	\label{lm:sum-disjoint-nodes}
	Let $A$ be a subset of $V(T)$ such that there is no pair of nodes in $A$ having the ancestor-descendant relationship. Then, $\sum_{x \in A} r_{\lvl(x)} =  O(\weight(\MST(X)))$.
\end{lemma}

\begin{proof}
	The crucial property to prove \Cref{lm:sum-disjoint-nodes} is the greedy property of the net-tree. Recall that in our net-tree construction (\Cref{def:net-tree}), the parent of a node $(u, i)$ is the one corresponding to the closest point to $u$ in $N_{i + 1}$. Let $\mathcal{B} = \{\ball(x, r_{\lvl(x)} / 5): x \in A\}$.
	
	We claim that every point $v \in X$ is in at most one ball in $\mathcal{B}$. Since no node in $A$ is an ancestor of another, it is sufficient to show that if a ball $\ball(x, r_{\lvl(x)} / 5)$ in $\mathcal{B}$ contains $v$ for some $x$, $v$ has to be a leaf of $x$. Then, all nodes $x$ such that $\ball(x, r_{\lvl(x)} / 5)$ containing $v$ must lie on the path from $(v, 0)$ to the root of $T$, and only one node in that path can belong to $A$. 
	
	Assume that there exists a node $x = (u, i)$ such that $\ball(u, r_{i} / 5)$ contains $v$ and $v$ is not a leaf of (the subtree rooted) $x$. Let $y = (w', i - 1)$ and $z = (w, i)$ be the ancestor of $(v, 0)$ at level $i - 1$ and $i$, respectively. By our choice of children for each node in the net, $\nd(u, w') \geq \nd(w, w')$, implying that:
	\begin{equation}
		\label{eq:half-dist-1}
		\nd(u, w') \geq (\nd(u, w') + \nd(w', w))/2 \geq \nd(u, w)/2 \geq r_i/2 \qquad.
	\end{equation}
	
	On the otherhand, by \Cref{clm:leaf-dist}, $\nd(w', v) \leq 5r_{i - 1}/4$. Using the triangle inequality, we obtain:
	\begin{equation*}
		\label{eq:half-dist-2}
		\nd(u, w') \leq \nd(u, v) + \nd(v, w') \leq r_i / 5 + 5r_{i - 1}/4 = 9r_i / 20 < r_i/2 \qquad,
	\end{equation*}
	contradicting to \Cref{eq:half-dist-1}. 
	
	Therefore, each point in $X$ belongs to at most one ball in $\mathcal{B}$. Thus, for any two nodes $x, y \in A$, $\nd(x, y) \geq \max\{r_{\lvl(x)}, r_{\lvl(y)}\} / 5$.
	
	Let $X_A$ be the set of points in $A$ (we translate each node in $A$ by its representative in $X$). By a folkore result, $\weight(\MST(X_A)) \leq 2\cdot \weight(\MST(X))$. For each node $x \in A$, let $e(x)$ be an arbitrary edge incident to $x$ in $\MST(X_A)$. Hence, $\weight(e(x)) \geq r_{\lvl(x)}/5$, implying that: 
	\begin{equation}
		\label{eq:sumradi}
		\sum_{x \in A}\weight(e(x)) \geq \sum_{x \in A}r_{\lvl(x)}/5\qquad.
	\end{equation}
	On the other hand, since each edge in $\MST(X_A)$ is incident to at most $2$ nodes in $A$, we have:
	\begin{equation}
		\label{eq:sumedges}
		\sum_{x \in A}\weight(e(x)) \leq 2\sum_{e \in E(X_A)}\weight(e) \leq 2\weight(\MST(X_A)) \qquad.
	\end{equation}
	By \Cref{eq:sumradi} and \Cref{eq:sumedges}, $\sum_{x \in A}r_{\lvl(x)}/5 \leq 2\weight(\MST(X_A)) \leq 4\weight(\MST(X))$. Therefore, $\sum_{x \in A}r_{\lvl(x)} = O(\weight(\MST(X)))$.
\end{proof} 
We are now ready to bound $\weight(E_{inc})$ (in \Cref{lm:Einc-weight}) and $\weight(E_{com})$ (in \Cref{lm:Ecom-weight}).

\begin{lemma}
	\label{lm:Einc-weight}
	$\weight(E_{inc}) = \eps^{-O(d)}f^2\weight(\MST(X))$.
\end{lemma}

\begin{proof}
	Let $A$ be the set of roots of subtrees in $\mathrm{LNF}$ (or the set of almost complete nodes of $T$). There is no node in $A$ being the ancestor of another. Hence, by \Cref{lm:sum-disjoint-nodes},
	\begin{equation}
		\label{eq:bdd-rad-almost1}
		\sum_{x \in A}r_{\lvl(x)} = O(\weight(\MST(X))).
	\end{equation}
	
	By \Cref{def:lnf} and \Cref{clm:parent-complete}, each incomplete node has exactly one almost complete ancestor in $A$. For each point $u \in X$, let $i_u$ be the level of the root of the subtree in LNF containing $(u, 0)$. Let $E_{inc}(u)$ be the set of all edges with level from $0$ to $i_u$ incident to $u$ in $E_{inc}$, then $E_{inc} = \bigcup_{u \in X}E_{inc}(u)$. Note that there might be other edges in $E_{inc}$ that are incident to $u$; however, each of those edges must be incident to some vertex $v$ with $i_v$ larger than the level of that edge.  
	
	By \Cref{lm:bounded-surrogate}, $u$ is incident to at most $\xi f$ level-$i$ edges for every $0 \leq i \leq \zeta$. Hence, the total weight of level-$i$ edges incident to $u$ is at most $\xi f \cdot \lambda r_i$. This gives:
	\begin{equation*}
		\weight(E_{inc}(u)) \leq \sum_{0 \leq i \leq i_u} \xi f \cdot \lambda r_i \leq 2\lambda\xi f\cdot r_{i_u}
	\end{equation*}
	by geometric sum.
	
	As $E_{inc} = \bigcup_{u \in X}E_{inc}(u)$, $\weight(E_{inc}) \leq \sum_{u \in X}2\lambda\xi f\cdot r_{i_u}$. By \Cref{obs:all-leaves}, we partition the set $X$ into the representatives of leaves within the same subtree in LNF. Thus,
	\begin{equation}
		\label{eq:einc1}
		\begin{split}
			\weight(E_{inc}) &\leq 2\lambda\xi f \cdot \sum_{u \in X} r_{i_u} \leq \eps^{-O(d)}f \cdot\sum_{x \in A}\sum_{u \in L(x)}\underbrace{r_{i_u}}_{ = r_{\lvl(x)}} \leq  \eps^{-O(d)}f\cdot\sum_{x \in A} |L(x)| \cdot r_{\lvl(x)}
		\end{split}
	\end{equation}
	Since each node in LNF is incomplete, each of its subtrees has at most $f$ leaves by the definition of incomplete. Hence, from \Cref{eq:einc1}, we have: 
	\begin{equation*}
		\begin{split}
			\weight(E_{inc}) &\leq  \eps^{-O(d)}f\cdot\sum_{x \in A} |L(x)| \cdot r_{\lvl(x)} \leq \eps^{-O(d)}f^2 \cdot \sum_{x \in A} r_{\lvl(x)}\\
			& \leq \eps^{-O(d)}f^2 \cdot \weight(\MST(X)) \qquad \text{(by \Cref{eq:bdd-rad-almost1})},
		\end{split}
	\end{equation*}
	as claimed.
\end{proof}

\begin{lemma}
	\label{lm:Ecom-weight}
	$\weight(E_{com}) = \eps^{-O(d)}f^2\weight(\MST(X))$.
\end{lemma}

\begin{proof}
	We bound the weight of $E_{com}$ by the weight of $E^*_{com}$. Since the bipartite connection of each cross edge $(y, z)$ in $E^*_{com}$ has exactly $f + 1$ edges. Each edge in $M(S(y), S(z))$ has both ends in $\ball(y, 16r_{\lvl(y)})$ and $\ball(z, 16r_{\lvl(z)})$ ($\lvl(y) = \lvl(z)$), and hence has weight at most $(\lambda + 32) r_{\lvl(y)}$ by the triangle inequality. We obtain that: 
	\begin{equation}
		\label{eq:Ecom}
		\begin{split}
			\weight(E_{com}) &\leq \sum_{(y, z) \in E^*_{com}}(f + 1)(\lambda + 32) r_{\lvl(y)} \\
			&\leq \sum_{(y, z) \in E^*_{com}}(f + 1)(\lambda + 32) \nd(y, z) / 64 \qquad \text{(since $\nd(y, z) \geq 64r_{\lvl(y)}$)}\\
			&\leq \sum_{(y, z) \in E^*_{com}}(f + 1)\lambda\nd(y, z) = (f + 1)\lambda\weight(E^*_{com})
		\end{split}		
	\end{equation}
	By \cref{line:check-Dx}--\ref{line:add-NCx}, every cross edge$(y, z)$ in $E^* \setminus E^*_O$ is in $\Cr(NC[x])$ for some node $x$ such that $x$ is either in the LNF or $x$ is an ancestor within $5\log{\lambda}$ levels of the root of some subtree in LNF. Let $A = \{a_1, a_2, \ldots a_k\}$ be the set contains all roots of subtrees in $\mathrm{LNF}$. By \Cref{lm:sum-disjoint-nodes},
	\begin{equation}
		\label{eq:bdd-rad-almost2}
		\sum_{i = 1}^kr_{\lvl(a_i)} = O(\weight(\MST(X))).
	\end{equation}
	
	We have:  
	\begin{equation*}
		E^* \setminus E^*_O \subseteq \bigcup_{x \in A}\Aug(x, 0, 5\log{\lambda}) \bigcup \bigcup_{y \in \mathrm{LNF}}\Cr(NC[y]). 
	\end{equation*}
	Since $E^*_{com} = E^*\setminus (E^*_O \cup E^*_{inc}) \subseteq E^* \setminus E^*_O$, we have the following bound on $\weight(E^*_{com})$:
	
	\begin{equation}
		\label{eq:E*com}
		\weight(E^*_{com}) \leq \weight(E^* \setminus E^*_O) \leq \sum_{x \in A}\weight(\Aug(x, 0, 5\log{\lambda})) + \sum_{y \in LNF}\weight(\Cr(NC[y])).
	\end{equation}
	
	The reason behind the exclusion of $E^*_{inc}$ is to make sure that each bipartite connection of cross edges in $E^*_{com}$ contains exactly $f + 1$ edges in $H$. 
	
	By \Cref{it:bounded-up-aug} of \Cref{prop:lightness-properties}, we get: 
	\begin{equation}
		\label{eq:up}
		\weight(\Aug(x, 0, 5\log{\lambda}) \setminus E^*_{inc}) \leq \weight(\Aug(x, 0, 5\log{\lambda})) = \eps^{-O(d)} r_{\lvl(x)} \qquad.
	\end{equation} 
	
	We then bound $\sum_{y \in \mathrm{LNF}}\weight(\Cr(NC[y]))$. By the packing bound (\Cref{lm:packing-bound}), $NC[y]$ contains at most $\eps^{-O(d)}$ nodes. Thus, $\Cr(NC[y])$ contains at most $\eps^{-O(d)}$ cross edges, each of them has weight $\lambda r_{\lvl(y)}$. Thus,
	
	\begin{equation}
		\label{eq:sum-cross}
		\sum_{y \in \mathrm{LNF}}\weight(\Cr(NC[y])) \leq \sum_{y \in \mathrm{LNF}}\lambda r_{\lvl(y)}.
	\end{equation}
	
	We partition the nodes in LNF into the set of nodes in subtrees of $T$. Let $\{T_1, T_2, \ldots T_k\}$ is the set of subtrees in $T$ with root $\{a_1, a_2, \ldots a_k\}$ respectively ($T_1, T_2, \ldots T_k$ are node-disjoint). Then $\mathrm{LNF} = T_1 \cup T_2 \cup \ldots T_k$. For each $i \in [1, k]$, let $T_{i, j}$ be the set of nodes of $T_i$ at level $j$. Let $A = \{a_1, a_2, \ldots a_k\}$ be roots of $T_1, T_2, \ldots T_k$. 
	
	Since each subtree has at most $f$ leaves, $|T_{i, j}| \leq f$ for every $i \in [1, k]$ and $j \leq \lvl(a_i)$. We have: 
	
	\begin{equation}
		\label{eq:down}
		\begin{split}
			\sum_{y \in \mathrm{LNF}}\lambda r_{\lvl(y)} &= \lambda\sum_{i = 1}^k \sum_{y \in T_i} r_{\lvl(y)} = \lambda\sum_{i = 1}^k \sum_{j = 0}^{\lvl(a_i)}\sum_{y \in T_{i, j}}r_{\lvl(y)} = \lambda\sum_{i = 1}^k \sum_{j = 0}^{\lvl(a_i)}\sum_{y \in T_{i, j}}r_{j} \\
			&= \lambda\sum_{i = 1}^k \sum_{j = 0}^{\lvl(a_i)}|T_{i, j}| r_{j} \leq \lambda\sum_{i = 1}^k \sum_{j = 0}^{\lvl(a_i)}f r_{j} \qquad \text{(since $|T_{i, j}| \leq f$)}\\
			&\leq \lambda f \sum_{i = 1}^k 2r_{\lvl(a_i)} \qquad \text{(by geometric sum)}\\
			&\leq \lambda f \cdot \weight(\MST(X)) \qquad \text{(by \Cref{eq:bdd-rad-almost2})}.
		\end{split}
	\end{equation}
	
	By \Cref{eq:Ecom} and \Cref{eq:E*com}, we have:
	\begin{equation*}
		\begin{split}
			\weight(E_{com}) &\leq \lambda(f + 1) \cdot \weight(E^*_{com}) \leq \lambda(f + 1) \left[\sum_{x \in A}\weight(\Aug(x, 0, 5\log{\lambda})) + \sum_{y \in LNF}\weight(\Cr(NC[y]))\right]\\
			& \leq \lambda(f + 1)\left[\sum_{x \in A}\eps^{-O(d)} r_{\lvl(x)} +  \sum_{y \in LNF}\weight(\Cr(NC[y]))\right] \qquad\text{(by \Cref{eq:up})}\\
			&\leq  \lambda(f + 1)\left[\sum_{x \in A}\eps^{-O(d)} r_{\lvl(x)} + \sum_{y \in \mathrm{LNF}}\lambda r_{\lvl(y)}\right] \qquad\text{(by \Cref{eq:sum-cross})}\\
			&\leq  \lambda(f + 1)\left[\eps^{-O(d)} \weight(\MST(X)) + \lambda f \cdot \weight(\MST(X))\right] \qquad\text{(by Eq. \ref{eq:bdd-rad-almost2} and Eq. \ref{eq:down})}\\
			&= \eps^{-O(d)}f^2 \cdot \weight(\MST(X)),
		\end{split}
	\end{equation*}
	as desired.
\end{proof}

We now ready to prove \Cref{lm:weight_E-E_O}.

\begin{proof}[Proof of \Cref{lm:weight_E-E_O}]
	Observe that $E(H) \setminus E_O = E_{com} \cup E_{inc}$, we obtain $\weight(E(H) \setminus E_O) \leq \weight(E_{com}) + \weight(E_{inc}) = \eps^{-O(d)}f^2\cdot\weight(\MST(X))$ by \Cref{lm:Einc-weight} and \Cref{lm:Ecom-weight}. 
\end{proof}

\subsubsection{Proof of \Cref{lm:complete-node-complete}}

\label{appen:p1}
In this section, we provide a detailed proof of \Cref{lm:complete-node-complete}. First, we introduce some notation. A point $v$ is clean if it has a degree of at most $c_2f$ and is semi-saturated if it is not marked as saturated and has a degree larger than $c_2f$. Point $v$ is saturated if it is marked as saturated in \cref{line:mark-saturated} of \Cref{alg:VFT-spanner}. Note that ``clean" and ``semi-saturated" are \emph{time-sensitive} properties, meaning that they change over time. Specifically, a vertex may change from clean to semi-saturated after the execution of \cref{line:H-update} of \Cref{alg:VFT-spanner}. In this section, we say a point is ``clean" or ``semi-saturated" with respect to some specific moment while running \Cref{alg:VFT-spanner}; this usually happens when we select surrogates in \Cref{alg:Surrogate} or before/after the execution of \cref{line:H-update} in some iteration.

For every node $(u, i)$, recall that we choose the surrogate set of $(u, i)$ by: 
\begin{itemize}
	\item If $(u, i)$ is small and incomplete, $S(u, i) = L(u, i)$ where $L(u, i)$ is the set of leaves of the subtree of $T$ with root $(u, i)$.
	\item If $(u, i)$ is small and complete, $S(u, i)$ contains $f + 1$ arbitrary non-saturated points in $L(u, i)$.
	\item If $(u, i)$ is large, we find all semi-saturated points in $\ball(u, 16r_i)$ and select arbitrary $f + 1$ of them to $S(u, i)$. If there are not enough $f + 1$ such points, we add clean points in $\ball(u, 4r_i)$ to $S(u, i)$ until $|S(u, i)|$ reaches $f + 1$.
\end{itemize}

To prove \Cref{lm:complete-node-complete}, we show that there are always more than $f + 1$ points in $\ball(u, 4r_i)$ if $(u, i)$ is large. Hence, the surrogate set of $(u, i)$ always contains $f + 1$ points. 

For each point $v$, let $\deg_{< i}(v)$ and $\deg_{\leq i}$ be the degree of $v$ before and after the $i^{th}$ iteration, respectively. A point $v$ is \textit{$i$-clean} if $\deg_{< i}(v) \leq c_2f$ and is \textit{$i$-saturated} if $\deg_{< i}(v) > c_3f$. Recall that a small node is complete if it has at least $f + 1$ leaves and a large node is always complete. 

\begin{lemma}
	\label{lm:complete-ball}
	For every large node $(u, i)$, the number of $i$-clean points in $\ball(u, 4r_i)$ is at least $(4f + 4)$.
\end{lemma}

We will show that \Cref{lm:complete-node-complete} follows from \Cref{lm:complete-ball}. 

\begin{proof}[Proof of \Cref{lm:complete-node-complete}]
	Let $x = (u, i)$ be a large node. We consider two cases:
	
	Case $1$. If $x$ is small, then $x$ has at least $f + 1$ leaves by the definition of a complete node. We claim that there are $f + 1$ non-saturated points in $L(u, i)$. For every leaf $v$ of $x$, i.e., $v \in L(u, i)$, the degree of $v$ before iteration $i$ is at most $c_1f$ since $x$ is small. Hence, $v$ is not marked as saturated before iteration $i$ since $c_1f \leq c_3f$. (Note that $v$ is not marked as saturated during the execution of level $i$.) Therefore, $v$ is either clean or semi-saturated when we update $S(u, i)$. (In \Cref{alg:VFT-spanner}, $S(u, i)$ might be updated multiple times.) Since all leaves of $x$ are either clean or semi-saturated, there must be at least $f + 1$ non-saturated points in $L(u, i)$ as $x$ has at least $f + 1$ leaves. By \cref{line:add-leaf-S} in \Cref{alg:Surrogate}, $|S(u, i)| = f + 1$.
	
	Case $2$. If $x$ is large, by \Cref{lm:complete-ball}, the number of $i$-clean points in $\ball(u, 4r_i)$ is at least $4f + 4$. Let $C$ be the set of clean points in $\ball(u, 4r_i)$. During the execution of level $i$, the points in $C$ are not marked as saturated. Therefore, they either remain clean or become semi-saturated while updating $S(u, i)$. From \cref{line:def-Sprime}--\ref{line:general-surrogate}, $S(u, i)$ is formed by choosing semi-saturated points in $\ball(u, 16r_i)$ and clean points in $\ball(u, 4r_i)$. As $C \subseteq \ball(u, 4r_i) \subseteq \ball(u, 16r_i)$, all points in $C$ are eligible to be selected to $S(u, i)$. Thus, there are enough $f + 1$ "candidates" for surrogates in $S(u, i)$, implying $|S(u, i)| = f + 1$.   
\end{proof}

We now focus on proving \Cref{lm:complete-ball}. First, we list some properties of the net-tree $T$. Recall that for a given node $(u, i)$, $D(u, i)$ is the set of descendants of $(u, i)$.

\begin{property} \label{prop:pool-prop} We have the following properties:
	\begin{enumerate}
		\item \label{it:descendant-pool} Let $(u, i)$ be an arbitrary complete node in $T$. For any $(v, j) \in D(u, i)$ and any $\alpha \geq 2$, $\ball(v, \alpha r_j) \subseteq \ball(u, \alpha r_i)$.
		\item \label{it:gain-vertices} Let $(u, i)$ be an arbitrary node in $T$, $v$ be a point in $\ball(u, 3r_i)$ and $(v, w)$ be a level-$k$ edge with $k \leq i - 2\log{\lambda}$. Then, $w \in \ball(u, 4r_i)$. 
		\item \label{it:level-i-edge} For every level-$i$ edge $(u, v)$, $\nd(u, v) \geq 32r_i$.
	\end{enumerate}
\end{property}

\begin{proof}
	\textit{\Cref{it:descendant-pool}: } Since $i > j$, we have $r_j \leq r_i / 5$. By \Cref{clm:leaf-dist}, $\nd(u, v) \leq 5r_i / 4$. Thus, for each point $w \in \ball(v, \alpha r_j)$, we get:
	\begin{equation}
		\nd(u, w) \leq \nd(u, v) + \nd(v, w) \leq 5r_i / 4 + \alpha r_j \leq (5/4 + \alpha/5) r_i \leq \alpha r_i \qquad,
	\end{equation}
	since $\alpha \geq 2$.
	
	\textit{\Cref{it:gain-vertices}: } By triangle inequality, $\nd(v, w) \leq (\lambda + 32)r_{k} \leq \frac{\lambda + 32}{\lambda^2}r_i \leq r_i/2$. Then, by the triangle inequality, $\nd(u, w) \leq \nd(u, v) + \nd(v, w) \leq 3r_i + r_i/2 \leq 4r_i$.
	
	\textit{\Cref{it:level-i-edge}: } Let $(\tilde{u}, \tilde{v})$ be the level-$i$ cross edge such that $(u, v) \in M(S(\tilde{u}), S(\tilde{v}))$. By construction, $\nd(\tilde{u}, \tilde{v}) \geq 64r_i$. Since $u \in \ball(\tilde{u}, 16r_i)$ and $v \in \ball(\tilde{v}, 16r_i)$, we have $\nd(u, v) \geq \nd(\tilde{u}, \tilde{v}) - \nd(u, \tilde{u}) - \nd(v, \tilde{v}) \geq 32r_i$.
\end{proof}

We show that for any set $A$ of small diameter, any large node using a clean point in $A$ as a surrogate must also use all semi-saturated points in $A$. This property is due to the fact that we prioritize the use of semi-saturated points over clean ones. 

\begin{claim}
	\label{clm:semi-first}
	Let $i$ and $i'$ be two levels such that $i' \geq i$, $x$ be a node at level $i'$ and $S(x)$ be the result of $\mathrm{SelectSurrogate}(x)$ (\Cref{alg:Surrogate}). For any set $A$ of diameter at most $8r_i$, if $S(x)$ contains any clean point in $A$, then $S(x)$ also contains all semi-saturated points in $A$.
\end{claim}

\begin{proof}
	Since $|S(x)| \leq f + 1$, $|A \cap S(x)| \leq f + 1$. Let $x = (u, i')$ and $v$ be a clean point in $A \cap S(x)$. Recall that in the construction of $S(x)$, first, we find all semi-saturated points in $\ball(u, 16r_{i'})$. If there are not enough $f + 1$ points in $\ball(u, 16r_{i'})$, we pick some clean points in $\ball(u, 4r_{i'})$. Thus, $v \in \ball(u, 4r_{i'})$. Since $\diam(A) \leq 8r_i$, for every $w \in A$,
	\begin{equation}
		\nd(w, u) \leq \nd(w, v) + \nd(u, v) \leq \diam(A) + 4r_{i'} \leq 8r_i + 4r_{i'} \leq 12r_{i'}\qquad,
	\end{equation}
	implying that $A \subseteq \ball(u, 16r_{i'})$. Let $A_s$ be the set of semi-saturated points in $A$. Observe that $A_s \subset S(x)$ since otherwise, there is no clean point in $S(x)$.
\end{proof}

Note that in the proof of \Cref{clm:semi-first}, we only need that for every node $x$ with $i = \lvl(x)$, all the semi-saturated vertices in $\ball(x, 12r_{i})$ must be in $S(x)$ before any clean vertex is added to $S(x)$. Hence, we prioritize selecting semi-saturated vertices in $\ball(x, 16r_{i})$ over the clean vertices in $\ball(x, 4r_{i})$ in \cref{line:def-Sprime}. Throughout \Cref{alg:VFT-spanner}'s analysis, the proof of \Cref{clm:semi-first} is the first (and also the only) proof that requires some geometric property other than the radius of the set where we choose the semi-saturated vertices from in \Cref{alg:Surrogate} (which is $\ball(x, 16r_i)$). Indeed, for all the other proofs in \Cref{sec:degree}, \Cref{sec:light} and \Cref{sec:connectivity}, we only need: $(i)$ all semi-saturated vertices in $S(x)$ have distance at most $16r_i$ from $x$ and $(ii)$ the set where we choose the semi-saturated vertices in \cref{line:def-Sdprime} contains the set where we choose clean vertices in \cref{line:def-Sprime}.

\begin{remark}
	\label{rm:ext-pool}
	\Cref{clm:semi-first} still holds if we replace the $\ball(x, 16r_i)$ ($i$ is the level of $x$) in \cref{line:def-Sprime} of \Cref{alg:Surrogate} by any subset of $\ball(x, 16r_i)$ containing $\ball(x, 12r_i)$. Furthermore, the correctness of \Cref{alg:VFT-spanner} still holds if we choose the semi-saturated vertices in $S'$ (in \cref{line:def-Sprime} of \Cref{alg:Surrogate}) from any subset of $\ball(x, 16r_i)$ containing $\ball(x, 12r_i)$. 
\end{remark}

By \Cref{clm:semi-first}, if a surrogate set $S(x)$ contains a clean point in a low-diameter set $A$, then the total number of semi-saturated points in $A$ is less than $f + 1$. More importantly, assuming that we are considering a cross edge $(x, y)$ in \cref{line:edge-process}, then after adding $M(S(x), S(y))$ to $H$, there are still $f + 1$ semi-saturated points in $A$, since we only change the degree of at most $f + 1$ points in $A$ and prioritize using semi-saturated points over clean ones. Then, we have the following direct corollary of \Cref{clm:semi-first}:

\begin{corollary}
	\label{cor:f+1-non-clean}
	Let $i$ and $i'$ be two levels such that $i' \geq i$, $A$ be a set of points with a diameter at most $8r_i$ and $(x, y)$ be a level-$i'$ cross edge in $E^*$. If $S(x)$ contains a clean point in $A$ before $M(S(x), S(y))$ is added to $H$, then the total number of semi-saturated points in $A$ after the adding of $M(S(x), S(y))$ is at most $f + 1$.
\end{corollary}

For each point $u \in X$, let $\deg_i(u)$ be the number of level-$i$ edges incident to $u$ in $H$. We have $\deg_{\leq i}(u) = \sum_{k = 0}^i\deg_{k}(u)$ and $\deg_{< i}(u) = \sum_{k = 0}^{i - 1}\deg_k(u)$. For two integers $i_1$ and $i_2$ that $0 \leq i_1 \leq i_2 \leq \zeta$, let $\deg_{[i_1, i_2]}(u)$ be the total number of edges at a level within the range $[i_1, i_2]$. Formally, $\deg_{[i_1, i_2]}(u) = \sum_{k = i_1}^{i_2}\deg_k(u)$.

We now prove that if the maximum degree increase of a set is less than the gap between saturated and clean, there are at most $f + 1$ points in that set becoming semi-saturated. For every $A \subseteq X$, let $\md_i(A) = \max_{w \in A}\deg_{i}(w)$ and $\md_{[i, j]}(A) = \max_{w \in A}\deg_{[i, j]}(w)$

\begin{lemma}
	\label{lm:clean-ball}
	Let $i$ be a level and $A$ be a set of $i$-clean points with a diameter at most $8r_i$. For every $k$ such that  $\md_{[i, i + k]}(A) \leq (c_3 - c_2)f$, the number of $(i + k)$-clean points in A is at least $|A| - (f + 1)$. 
\end{lemma}

\begin{proof}
	Let $A = \{a_1, a_2, \ldots a_{|A|}\}$. If $|A| \leq f + 1$, \Cref{lm:clean-ball} trivially holds. Assume that $|A| > f + 1$. We prove that at most $f + 1$ points in $A$ are semi-saturated before the execution of level $i + k$ in \cref{line:level-iter} of \Cref{alg:VFT-spanner}. For any $j \in [1\ldots |A|]$, since $a_j$ is $i$-clean, $\deg_{< i}(a_j) \leq c_2 f$. Then, $\deg_{\leq i + k}(a_j) \leq c_2 f + \max_{w \in A}\deg_{[i, i + k]}(w) \leq c_3 f$ by the assumption of the lemma. Thus, no point in $A$ is saturated before level $i + k$. 
	
	We prove by contradiction that $A$ contains at least $|A| - (f + 1)$ clean points before iteration $i + k$. Let $(x, y)$ be the first cross edge such that after adding $M(S(x), S(y))$ to $E(H)$, $A$ contains less than $|A| - (f + 1)$ clean points. Let $i_{xy}$ be the level of $(x, y)$; we have $i \leq i_{xy} \leq i + k$. Thus, no point in $A$ is saturated during the execution of level $i_{xy}$, implying that $A$ contains only clean and semi-saturated points before and after adding $M(S(x), S(y))$. Since some points in $A$ become semi-saturated after adding $M(S(x), S(y))$, either $S(x)$ or $S(y)$ contains some clean points in $A$. Without loss of generality, assume that $S(x)$ does. We claim that $S(y)$ does not contain any point in $A$. Let $u$ be a point in $S(x) \cap A$. By the construction of $S(x)$ in \Cref{alg:Surrogate}, $\nd(x, u) \leq 16r_{i_{xy}}$. For every point $v \in A$, we have:
	\begin{equation*}
		\nd(v, y) \geq \nd(x, y) - \nd(x, u) - \nd(u, v) \geq 64r_{i_{xy}} - 16r_{i_{xy}} - \diam(A) \geq 40r_{i_{xy}} \qquad,   
	\end{equation*}
	since $\diam(A) \leq 8r_i \leq 8r_{i_{xy}}$. Hence, $S(y)$ does not contain any point in $A$ because by \Cref{alg:Surrogate}, points in $S(y)$ are selected from $\ball(y, 16r_{i_{xy}})$. By \Cref{cor:f+1-non-clean}, there are at most $f + 1$ semi-saturated points in $A$ after adding $M(S(x), S(y))$, a contradiction.
\end{proof}

By \Cref{lm:bounded-surrogate}, each point's degree can only increase by $\xi f$ after each iteration, which means that after a small number of iterations, a set of clean points has at most $f + 1$ points turning into non-clean.

\begin{corollary}
	\label{cor:small-change-constant-lv}
	Let $x = (u, i)$ be a node in $T$ and $x' = (u', i')$ be a descendant of $x$ with $i' \geq i - 4\log{\lambda}$. If $\ball(u', 4r_{i'})$ contains a set $A$ of $i'$-clean points, then the number of $i$-clean points in $\ball(u, 4r_i)$ is at least $|A| - (f + 1)$. 
\end{corollary}

\begin{proof}
	By \Cref{lm:bounded-surrogate}, $\md_{[i', i]}(w) \leq (i - i') \cdot \xi f\leq  4\log{\lambda} \cdot \xi f < (c_3 - c_2)f$. By \Cref{lm:clean-ball}, the number of $i$-clean points in $A$ is at least $|A| - (f + 1)$. From \Cref{it:descendant-pool} of \Cref{prop:pool-prop}, $A \subseteq \ball(u', 4r_{i'}) \subseteq \ball(u, 4r_i)$. Therefore, the number of $i$-clean points in $\ball(u, 4r_i)$ is at least $|A| - (f + 1)$.
\end{proof}

The next lemma shows that for each node $(u, i)$, $\ball(u, 4r_i)$ contains some clean points of the balls in lower levels close to some descendant $(u', i')$ of $(u, i)$. To find such balls, we monitor the changes in the degree of points near $u'$ and prove that for every degree gained, there are a corresponding number of clean points ``contributed" to $\ball(u, 4r_i)$. If there is an edge from a point near $u'$ to a leaf of a small node, the leaves of that node are clean and we can reuse them after a constant number of levels. If the edge is between a point near $u'$ to a leaf of a large node, we assume that there are some clean points close to the representative of each large node. This assumption will be our induction hypothesis in the proof of \Cref{lm:complete-ball}.

\begin{lemma}
	\label{lm:grow-clean}
	Let $(u, i)$ be a node at level $i > \lceil2\log{\lambda}\rceil + 1$ and $(u', i')$ be a descendant of $(u, i)$ at level $i' = i - (\lceil2\log{\lambda}\rceil + 1)$. Assume that for every large node $(u'', i')$ at level $i'$, the number of $i'$-clean points in $\ball(u'', 4r_{i'})$ is at least $4f + 4$. Then, there exists a subset $A$ of $\ball(u, 4r_i)$ containing only $i$-clean points such that $|A| \geq \lceil \frac{\md_{i'}\ball(u', 4r_{i'})}{\xi} \rceil$ and $A \cap \ball(u', 4r_{i'}) = \emptyset$.
\end{lemma}

\begin{proof}
	Let $z$ be a point in $\ball(u', 4r_{i'})$ such that $\deg_{i'}(z) = \md_{i'}(\ball(u', r_{i'}))$. Let $W$ be the set of points connected to $z$ by level-$i'$ edges. Thus, $|W| = \deg_{i'}(z)$. We consider two cases:
	
	Case $1$: If $\deg_{<i'}(w) \leq c_1f$ for every $w \in W$,  then by \Cref{lm:bounded-surrogate}, $\deg_{< i'}(w) \leq c_1f + (\lceil2\log{\lambda}\rceil + 1) \cdot \xi f \leq c_2f$. Hence, for all $w \in W$, $w$ is $(i' + \lceil2\log{\lambda}\rceil + 1)$-clean (or $i$-clean). By setting $A = W$, we claim that $A$ satisfies all properties in \Cref{lm:grow-clean}. First, $W$ is a set of $i$-clean points contained in $\ball(u_{i}, 4r_{i})$ and $|W| = \deg_{i'}(z) = \md_{i'}(\ball(u, r_{i'})) \geq \lceil\frac{\md_{i'}\ball(u', 4r_{i'})}{\xi}\rceil$. Furthermore, for every $w \in W$, $\nd(w, z) \geq 32r_{i'}$ by \Cref{it:level-i-edge} in \Cref{prop:pool-prop}. By triangle inequality, $\nd(w, u') \geq \nd(w, z) - \nd(z, u') \geq 32r_{i'} - 4r_{i'} = 28r_{i'}$, implying that $w \not \in \ball(u', 4r_{i'})$. Thus, $W \cap \ball(u', 4r_{i'}) = \emptyset$.
	
	Case $2$: There exists a point $w_0 \in W$ such that $\deg_{< i'}(w_0) > c_1f$. Hence, the ancestor of $(w_0, 0)$ at level $i'$, denoted by $(w', i')$, is a large node. By the lemma's assumption, the number of $i'$-clean points in $\ball(w', 4r_{i'})$ is at least $4f + 4$. Let $C$ be the set of $i'$-clean points in $\ball(w', 4r_{i'})$. By \Cref{lm:bounded-surrogate}, for every level $k$, $\md_{k}(C) \leq \xi f$; hence, $\sum_{k = i'}^{i}\md_k(C) \leq (\lceil2\log{\lambda}\rceil + 1) \cdot \xi f \leq (c_3 - c_2)f$ by the choice of $c_2$ and $c_3$ in \Cref{alg:VFT-spanner}. Thus, by \Cref{lm:clean-ball}, the number of $i$-clean points in $C$ is at least $|C| - (f + 1) \geq 3(f + 1) \geq \lceil \frac{\md_{i'}\ball(u', 4r_{i'})}{\xi} \rceil$ since $\md_{i'}\ball(u', 4r_{i'}) \leq \xi f$ by \Cref{lm:bounded-surrogate}. 
	
	Let $A$ be the set of $i$-clean points in $C$ satisfies conditions in \Cref{lm:grow-clean}. We show that $A$ satisfies the conditions in \Cref{lm:grow-clean}. Since $A \subseteq C \subseteq \ball(u', 4r_{i'})$, $\diam(A) \leq \diam(C) \leq 8r_{i'}$. For every point $v \in C$, by the triangle inequality, 
	\begin{equation}
		\begin{split}
			\nd(u, v) &\leq \nd(u, u') + \nd(u', v) \leq 2r_i + \nd(u', v) \qquad \text{(by \Cref{clm:leaf-dist})}\\
			&\leq 2r_i + \underbrace{\nd(u', z)}_{\leq 4r_{i'}} + \nd(z, w_0) + \underbrace{\nd(w_0, w')}_{\leq 2r_{i'}} + \underbrace{\nd(w', v)}_{4r_{i'}} \leq 2r_i + 10r_{i'} + \nd(z, w_0)\\
			&\leq 2r_i + 10r_{i'} + (\lambda + 32)\underbrace{r_{i'}}_{\leq r_i / \lambda^2} \qquad \text{(since $(z, w_0)$ is a level-$i$ edge)}\\
			&\leq 2r_i + \frac{\lambda + 42}{\lambda^2}r_i \leq 4r_i.
		\end{split}
	\end{equation}
	Therefore, $C \subseteq \ball(u, 4r_i)$, implying that $A \subseteq \ball(u, 4r_i)$. Furthermore, for each $v \in C$, we have:
	\begin{equation}
		\begin{split}
			\nd(v, u') &\geq \nd(w_0, z) - \nd(v, w_0) - \nd(u', z) \qquad \text{(by the triangle inequality)}\\
			&\geq 32r_{i'} - \nd(v, w_0) - \underbrace{\nd(u', z)}_{\leq 4r_{i'}} \qquad \text{(since $\nd(z, w_0) \geq 32r_{i'}$)}\\
			&\geq 32r_{i'} - \underbrace{\nd(v, w')}_{\leq 4r_{i'}} - \underbrace{\nd(w', w_0)}_{\leq 2r_{i'}} - 4r_{i'} \geq 32r_{i'} - 4r_{i'} - 2r_{i'} - 4r_{i'} \geq 22r_{i'} > 4r_{i'}.
		\end{split}
	\end{equation}
	This implies that $v \not \in \ball(u', 4r_{i'})$ for every $v \in C$. Since $A \subseteq C$, it follows that $A \cap \ball(u', 4r_{i'}) = \emptyset$.
\end{proof}

We now return to the proof of \Cref{lm:complete-ball}. The idea is to track the degree change of points in some path $P$ of $T$ from $(u, i)$ to one of its leaves. If there is a point whose degree changed, then there is at least one cross edge between (a node in) $P$ to another node, say $(v, i')$. The clean points in $\ball(v, 4r_{i'})$ become ``closer" to the $P$ at higher levels relative to the level radius and eventually can be used as surrogates by nodes in $P$.

\begin{proof}[Proof of \Cref{lm:complete-ball}]
	
	We induct on the level $i$. When $i = 0$, \Cref{lm:complete-ball} holds trivially since there is no large node at level $0$. Assume that \Cref{lm:complete-ball} holds for all large nodes at levels lower than $i$. Since $x$ is large, there exists a leaf $v$ of $x$ whose $\deg_{<i}(v) \geq c_1 f$. Let $i' = i - 2\log{\lambda}$ and $x'$ be the ancestor of $(v, 0)$ at level $i'$. By \Cref{lm:bounded-surrogate}, $\deg_{<i'}(v) \geq c_1 f - 2\log{\lambda}\cdot \xi f$. We show that the number of $i'$-clean points in $\ball(u', 4r_{i'})$ is at least $6f + 6$. 
	
	For $k \in [0 \ldots i']$, let $(v_k, k)$ be the ancestor of $(v, 0)$ at level $k$. Let $l = \lceil 2\log{\lambda} \rceil$ and $i''$ be the highest level such that $\sum_{k = i''}^{i'}\md_{k}(\ball(v_k, 4r_k)) \geq (l + 1) \cdot 2\xi \cdot (6f + 6)$. The level $i''$ exists since $v \in \ball(v_k, 4r_k)$ for every $k$ by \Cref{clm:leaf-dist} and 
	\begin{equation*}
		\deg_{<i'}(v) \geq c_1 f - 2\log{\lambda}\cdot \xi f \geq 48\log{\lambda}\cdot \xi \cdot f \geq (l + 1) \cdot 2\xi \cdot (6f + 6) \qquad \text{(by the choice of $c_1$).}
	\end{equation*}
	We partition the set $I = \{i'', i'' + 1, \ldots i'\}$ into congruent classes $I_0, I_1, I_2, \ldots I_l$ of modulo $l + 1$. Formally, $I_\alpha = \{k \in I : k - \alpha \equiv 0 \mod l + 1\}$ for each $\alpha \in [0\ldots l]$. By the pigeonhole principle, there exists an integer $\beta \in [0 \ldots l]$ such that:
	\begin{equation}
		\sum_{k \in I_\beta} \md_k(\ball(v_k, 4r_k)) \geq 2\xi \cdot (6f + 6)
	\end{equation}
	For simplicity, assume that $i' \equiv i'' \equiv \beta \mod l + 1$. Let $I_{\beta} = \{i'' = i_0, i_1, i_2, \ldots i_s = i'\}$ with $i_j = i'' + j \cdot (l + 1)$ and $s = \frac{i' - i''}{l + 1}$. 
	
	By \Cref{lm:grow-clean}, for each $k \in [0, 1, \ldots {s - 1}]$, there exists a set $A_{i_k}$ of $i_{k + 1}$-clean points in $\ball(v_{i_{k + 1}}, r_{i_{k + 1}})$ such that $|A_{i_k}| \geq \lceil\frac{\md_{i_k}(\ball(v_{i_k}, 4r_{i_k}))}{\xi}\rceil$. For each $h \leq k$,  $A_{i_h}$ is a subset of $\ball(v_{i_{h + 1}}, 4r_{i_{h + 1}})$ and hence is a subset of $\ball(v_{i_{k + 1}}, 4r_{i_{k + 1}})$. Since $A_{i_{k + 1}} \cap \ball(v_{i_{k + 1}}, 4r_{i_{k + 1}}) = \emptyset$, $A_{i_h} \cap A_{i_{k + 1}} = \emptyset$ for every $h \leq k$, implying that $A_{i_1}, A_{i_2}, \ldots A_{i_{s - 1}}$ are pairwise-disjoint. See  \Cref{fig:growing-clean-ball} for an illustration of the inclusion relation between $A_{i_1}, A_{i_2}$ and $A_{i_3}$ with $\ball(v_{i_2}, 4r_{i_2}), \ball(v_{i_3}, 4r_{i_3})$ and $\ball(v_{i_4}, 4r_{i_4})$. One can see from \Cref{fig:growing-clean-ball} that, if we keep going up the path, we gain more clean points.
	\begin{center}
		\begin{figure}[H]
			\includegraphics[width=0.7\textwidth]{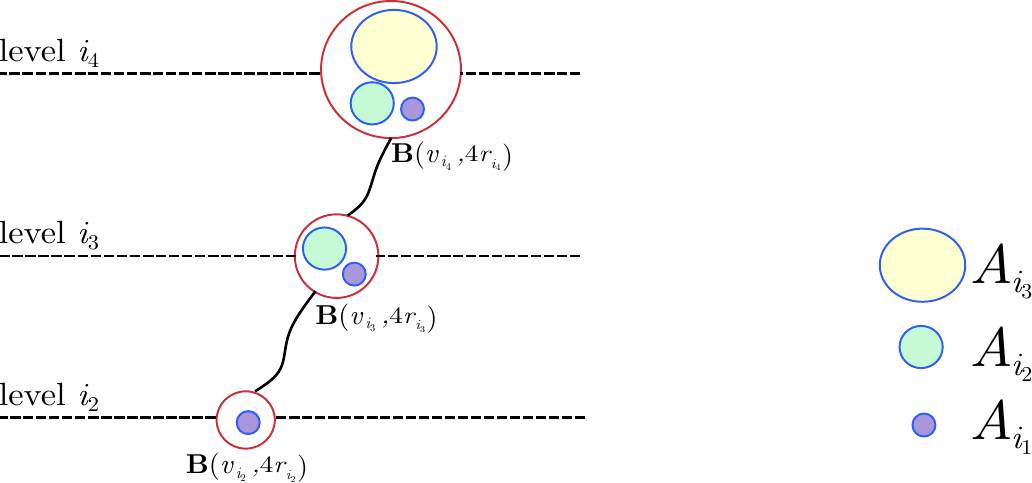}
			\caption{An example of $A_{i_1}, A_{i_2}$ and $A_{i_3}$. Here, $A_{i_1}, A_{i_2}$ and $A_{i_3}$ are disjoint. For each $k' \in \{1, 2, 3\}$, $A_{i_{k'}} \subseteq \ball(v_{i_{k}}, 4r_{i_{k}})$ for all $k > k'$.}
			\label{fig:growing-clean-ball}
		\end{figure}
	\end{center}
	We then bound the total number of points in $A_{i_k}$ for all $k \in [1, s-1]$: 
	\begin{equation}
		\begin{split}
			|\bigcup_{k = 1}^{s - 1}A_{i_k}| &= \sum_{t = 1}^{s - 1}|A_{i_k}| \geq \sum_{k = 1}^{s - 1}\lceil \frac{\md_{i_k}\ball(v_{i_k}, 4r_{i_k})}{\xi} \rceil \geq \sum_{k = 1}^{s - 1}\frac{\md_{i_k}\ball(v_{i_k}, 4r_{i_k})}{\xi}\\
			& = \left(\sum_{k = 1}^{s}\frac{\md_{i_k}\ball(v_{i_k}, 4r_{i_k})}{\xi}\right) - \frac{\md_{i_s}\ball(v_{i_s}, 4r_{i_s})}{\xi} \geq \frac{2\xi(6f + 6) - \xi f}{\xi} \geq 6f + 6.
		\end{split}
	\end{equation}
	
 In the next part of the proof, we show that only $f + 1$ of those clean points gained from lower levels turned into semi-saturated.  

Let $A = \bigcup_{k = 1}^{s - 1}A_{i_k}$. We prove by contradiction that at most $f + 1$ points in $A$ are not $i_s$-clean, which will give us the lemma. We have two observations. First, if a point is $k$-clean for some $k$, it is also $k'$-clean for every $k' \leq k$. Similarly, if a point is $k'$-saturated, it is also $k$-saturated for every $k \geq k'$. Hence, every point in $A$ is $i_0$-clean. Second, every point in $A$ is not saturated before iteration $i_s$.

\begin{claim}
	\label{clm:A-no-satuarted}
	There is no point in $A$ becoming saturated before iteration $i_s$.
\end{claim}

\begin{proof}
	Let $w$ be an arbitrary point in $A$ and $k \in [0, s - 1]$ be the index that $w \in A_{i_k}$. Since $w$ is $i_{k + 1}$-clean, $\deg_{<i_{k + 1}}(w) \leq c_2f$. On the other hand, $w \in \ball(v_{i_{k + 1}}, 4r_{i_{k + 1}})$. Recall that for every $i_0 \leq h' \leq h \leq i_s$, $(v_h, h)$ is the ancestor of $(v_{h'}, h')$ and hence $\ball(v_{h'}, 4r_{h'}) \subseteq \ball(v_h, 4r_{h})$. Thus, for every $h \in [i_{k + 1}, i_s]$, $w$ is also in $\ball(v_{h}, 4r_{h})$, implying that $\deg_h(w) \leq \md_h(\ball(v_{h}, 4r_{h}))$. Therefore,
	\begin{equation}
		\label{eq:w-deg}
		\begin{split}
			\deg_{< i_s}(w) &= \deg_{< i_{k + 1}}(w) + \sum_{h = i_{k + 1}}^{i_s - 1}\deg_h(w) \leq c_2f + \sum_{h = i_{k + 1}}^{i_s - 1}\md_h(\ball(v_h, 4r_h))\\
			&\leq c_2f + \sum_{h = i''}^{i'}\md_h(\ball(v_h, 4r_h)) \qquad \text{(since $i'' = i_0 \leq i_{k + 1} \leq i_s = i'$)}
		\end{split}
	\end{equation}
	Recall that $i''$ is the highest index such that $\sum_{k = i''}^{i'}\md_{k}(\ball(v_k, 4r_k)) \geq (l + 1) \cdot 2\xi \cdot (6f + 6)$. By the maximality of $i''$, $\sum_{k = i'' + 1}^{i'}\md_{k}(\ball(v_k, 4r_k)) < (l + 1) \cdot 2\xi \cdot (6f + 6)$, implying that:
	\begin{equation}
		\label{eq:i''-bound}
		\begin{split}
			\sum_{k = i''}^{i'}\md_{k}(\ball(v_k, 4r_k)) &\leq \sum_{k = i'' + 1}^{i'}\md_{k}(\ball(v_k, 4r_k)) + \md_{i''}(\ball(v_{i''}, 4r_{i''}))\\
			&\leq (l + 1) \cdot 2\xi \cdot (6f + 6) + \xi f \leq (c_3 - c_2)f.
		\end{split}
	\end{equation}
	The last equation holds by the choice of $c_2$ and $c_3$. By \Cref{eq:w-deg} and \Cref{eq:i''-bound}, we get $\deg_{<i_s}(w) \leq c_3f$ and hence, $w$ is non-saturated before iteration $i_s$.
\end{proof}

Assume that $A$ contains more than $f + 1$ non-clean points before iteration $i_s$. Let $(x, y)$ be the first cross edge in $E^*$ such that after $M(S(x), S(y))$ is added to $H$, $A$ has more than $f + 1$ non-clean points. Let $j$ be the level of $(x, y)$ and $t$ be index such that $i_t \leq j < i_{t + 1}$ $(0 \leq t < s)$. Let $\tau{(x, y)}$ be the time when $M(S(x), S(y))$ is added to $H$. Since all points in $A$ are not $i_s$-saturated by \Cref{clm:A-no-satuarted}, $A$ contains only clean and semi-saturated points before $\tau{(x, y)}$. Furthermore, the number of semi-saturated points in $A$ is at most $f + 1$ before $\tau{(x, y)}$. Let $A_{< i_t} = \bigcup_{k < t}A_{i_k}$ and $A_{\geq i_t} = A \setminus A_{< i_t} =  \bigcup_{k \geq t}A_{i_k}$. Since all points in $A_{\geq i_t}$ are $i_{t + 1}$-clean, they are also clean after $\tau{(x, y)}$. Then, all semi-saturated points in $A$ before and after $\tau{(x, y)}$ are in $A_{< i_t}$. Recall that $A_{< i_t} \subseteq \ball(u_{i_t}, 4r_{i_t})$; hence, $\diam(A_{< i_t}) \leq 8r_{i_t} \leq 8r_j$ as $i_t \leq j$. Since the number of clean points in $A_{< i_t}$ before $\tau(x, y)$ is smaller than that after $\tau(x, y)$, there must be a clean point in $A_{< i_t}$ becoming semi-saturated. Thus, $S(x)$ or $S(y)$ must contain a clean point in $A_{< i_t}$. However, by \Cref{cor:f+1-non-clean}, there are at most $f + 1$ semi-saturated points in $A_{< i_t}$ after $\tau(x, y)$, contradicted to the assumption that $A$ contains more than $f + 1$ non-clean points after $\tau(x, y)$. 

Therefore, the number of $i_s$-clean points in $A$ is at least $(6f + 6) - (f + 1) = 5f + 5$. Since $A \subseteq \ball(u_{i_s}, 4r_{i_s})  = \ball(u_{i'}, 4r_{i'})$ (recall $i_s = i'$), $\ball(u_{i'}, 4r_{i'})$ contains at least $5f + 5$ clean points at level $i'$. Using \Cref{cor:small-change-constant-lv}, we obtain that the number of $i$-clean points in $\ball(u, 4r_{i})$ is at least $(5f + 5) - (f + 1) = 4f + 4$.
\end{proof}

%% file: Connectivity.tex
\section{Fault-tolerance}
\label{sec:connectivity}
In this section, we prove the fault tolerance property, meaning that after removing any $f$ points from $H$, the remaining graph is still a spanner. Throughout this section, we assume that $\eps \leq \frac{1}{20}$.

\begin{lemma}
	\label{lm:vertex-disjoint-path}
 	\Cref{alg:VFT-spanner} produces a $f$-VFT $(1 + 5\eps)$-spanner $H$ of $X$, i.e., for any set $F$ of at most $f$ points in $X$, $H[X - F]$ is a $(1 + 5\eps)$-spanner of $(X - F, \nd)$. 
\end{lemma}

Equivalently, we need to show that for every pair of points $(u, v)$ and every set $F \subseteq X \setminus \{u, v\}$ of at most $f$ points, there is a path from $u$ to $v$ in $H - F$ whose length is at most $(1 + \eps)\nd(u, v)$. We prove by induction on the length $\nd(u, v)$. To find a short path from $u$ to $v$ in $H - F$, we follow the shortest path $P$ from $u$ to $v$ in $G$. Intuitively, we find a cross edge, denoted by $(x, y)$, in $E^*$ from an ancestor $x$ of (the leaf corresponding to) $u$ to the ancestor $y$ of some point $u'$ between $u$ and $v$ in $P$. Hence, we create a path in $H - F$ from $u$ to a surrogate of $x$ to a surrogate of $y$, and recursively do the same for the path from $u'$ to $v$. For this method to work, we need $(x, y)$ to have two properties:

\begin{itemize}
	\item $\nd(x, y)$ is a good approximation of $(u, u')$. This is the case when $(x, y)$ is a low ancestor of the original cross edge of $(u, u')$ (the formal definition of an ancestor of a cross edge will be given later). 
	\item Each surrogate set of $x$ and $y$ has $f + 1$ points; otherwise, if one surrogate set, say $S(x)$, has less than $f + 1$ points, then there is no non-faulty edges in the bipartite connection $M(S(x), S(y))$ in case $F = S(x)$ (a faulty edge is an edge with at least one end in $F$). By \Cref{lm:complete-node-complete}, a complete node always has $f + 1$ points in its surrogate. Hence, we find $x$ among the complete ancestors of $(u, 0)$. 
\end{itemize}

For each incomplete node $x$ in $T$, the \textit{lowest complete ancestor} (LCA) of $x$, denoted by $\LCA(x)$, is the parent of the almost complete ancestor of $x$. Recall that an almost complete node is an incomplete node whose parent is complete. For each point $u \in X$, let $\LCA(u) = \LCA(u, 0)$. For each cross edge $(x, y)$, a cross edge $(\tilde{x}, \tilde{y})$ is an ancestor (parent) of $(x, y)$ if $\tilde{x}$ and $\tilde{y}$ are ancestors (parents) of $x$ and $y$, respectively. For each pair $(u, v) \in X^2$, assume that $(\hat{u}, \hat{v})$ is the original cross edge of $(u, v)$. The \emph{$\kappa$-cross edge} of $(u, v)$ is the ancestor of $(\hat{u}, \hat{v})$ at level $\lvl(\hat{u}) + \kappa$. If $\kappa \leq 10$, we call $(x, y)$ a \emph{good cross edge} of $(u, v)$. 

Recall that an original cross edge of $(u, v)$  is the lowest-level cross edge $(x, y)$ such that $x$ and $y$ are ancestor of $(u, 0)$ and $(v, 0)$. For each level-$i$ node $x$ and a level $j$, consider two cases:
\begin{itemize}
	\item If $j \geq i$, $\Aug_j(x)$ is $\Cr(NC[x'])$, with $x'$ is the ancestor of $x$ at level $j$.
	\item If $j < i$, $\Aug_j(x)$ is the union among all descendants $x''$ of $x$ at level $j$ of $\Cr(NC[x''])$.
\end{itemize}
Recall that $\Cr(A)$ is the set of cross edges between nodes in $A$ for every $X \subseteq V(T)$.
\begin{property} We have the following properties: 
	\label{prop:connectivity}
	\begin{enumerate}
		\item \label{it:good-cross-edge-long} Let $(u, v)$ be any pair of points in $X$, $i$ be the level of the original cross edge of $(u, v)$  and $\kappa$ be any integer in $[0, \zeta - i]$. For every $\kappa$-cross edge $(\tilde{u}, \tilde{v})$ at level $j = i + \kappa$, $\nd(\tilde{u}, \tilde{v}) \geq r_{j} \cdot \left(\frac{\lambda}{6\cdot5^{\kappa}} - 4\right)$. Furthermore, if $\kappa \leq 10$, $\nd(\tilde{u}, \tilde{v}) \geq 64r_{j}$.

		\item \label{it:good-approx-weight} Let $(u, v)$ be any pair of points in $X$, $(\tilde{u}, \tilde{v})$ be any good cross edge of $(u, v)$ and $j$ be the level of $(\tilde{u}, \tilde{v})$. For every two points $u' \in \ball(\tilde{u}, 16r_j)$ and $v' \in \ball(\tilde{v}, 16r_j)$, $1 - 5^{-3}\eps \leq \frac{\nd(u', v')}{\nd(u, v)} \leq 1 + 5^{-3}\eps$.
		
		\item \label{it:node-to-augmented-dist} Let $x$ be any node in $T$, $i$ be a level in $T$ such that $\lvl(x) \leq i \leq \zeta$ and $(y, z)$ be any level-$i$ cross edge in $\Aug_i(x)$. Then, $\nd(x, y) \leq (\lambda + 2)r_i$.
		
		\item \label{it:level-weight-relation}
		For any original cross edge $(\hat{u}, \hat{v})$, $\log{\frac{\nd(\hat{u}, \hat{v})}{\lambda}} \leq \lvl(\hat{u}) < \log{\frac{\nd(\hat{u}, \hat{v})}{\lambda}} + 2$.
		
		\item \label{it:cross-edge-ancestor}
		Let $(x, y)$ be a cross edge at level $i$ ($\nd(x, y) \leq \lambda r_i$). For any ancestor $(x', y')$ of $(x, y)$ at level $i'$, $(x', y')$ is also a cross edge, i.e., $\nd(x', y') \leq \lambda r_{i'}$, which implies that $x' \in NC[y']$ and $y' \in NC[x']$.
	\end{enumerate}
\end{property}

\begin{proof}
	\emph{\Cref{it:good-cross-edge-long}:} Let $(\hat{u}, \hat{v})$ be the original cross edge of $(u, v)$. By \Cref{it:original-weight} of \Cref{prop:lightness-properties}, $\nd(\hat{u}, \hat{v}) \geq \lambda r_i / 6$. By \Cref{clm:leaf-dist}, $\nd(\hat{u}, \tilde{u}), \nd(\hat{v}, \tilde{v}) \leq 2r_{j}$. By the triangle inequality,
	
	\begin{equation}
		\nd(\tilde{u}, \tilde{v}) \geq \nd(\hat{u}, \hat{v}) - \nd(\tilde{u}, \hat{u}) - \nd(\tilde{v}, \hat{v}) \geq \lambda r_i / 6 - 4r_{j} = r_j\left(\frac{\lambda}{6 \cdot 5^{\kappa}} - 4\right) \qquad,
	\end{equation}
	as claimed. When $\kappa \leq 10$, $\nd(\tilde{u}, \tilde{v}) \geq r_j \cdot \left(\frac{\lambda}{6 \cdot 5^{10}} - 4\right) = r_j \cdot (5^{8}(1 + \eps^{-1}) - 4) \geq 64r_j$ as $\lambda = 5^{20}(1 + \eps^{-1})$.
	
	\emph{\Cref{it:good-approx-weight}:}	Let $(\hat{u}, \hat{v})$ be the original cross edge of $(u, v)$ and $i = \lvl(\hat{u})$. By \Cref{clm:leaf-dist}, $\nd(\tilde{u}, \hat{u}), \nd(\tilde{v}, \hat{v}) \leq 2r_{j}$.
	Using the triangle inequality, we have: 
	\begin{equation}
		\begin{split}
			\nd(u' , v') &\leq \nd(\hat{u}, \hat{v}) + \nd(u', \hat{u}) + \nd(v', \hat{v}) \\
			&\leq \nd(\hat{u}, \hat{v}) + \underbrace{\nd(u', \tilde{u})}_{\leq 16r_{j}} + \underbrace{\nd(v', \tilde{v})}_{\leq 16r_{j}} + ~\nd(\tilde{u}, \hat{u}) + \nd(\tilde{v}, \hat{v})\\
			&\leq \nd(\hat{u}, \hat{v}) + 16r_{j} + 16r_{j} + 2r_{j} + 2r_{j} \leq \nd(\hat{u}, \hat{v}) + 36r_{j} \leq \nd(\hat{u}, \hat{v}) + 5^{13}r_{i}.\\ 
		\end{split}
	\end{equation}
	Similarly, $\nd(u' , v') \geq \nd(\hat{u}, \hat{v}) - 5^{13}r_{i}$. 
	\begin{equation}
		\begin{split}
			\frac{\nd(u', v')}{\nd(u, v)} &\leq \frac{\nd(\hat{u}, \hat{v}) + 5^{13} r_i}{\nd(\hat{u}, \hat{v}) - 4r_i} = 1 + \frac{5^{13}r_i + 4r_i}{\nd(\hat{u}, \hat{v}) - 4r_i} = 1 + \frac{5^{13} + 4}{\nd(\hat{u}, \hat{v})/r_i - 4} \\
			&\leq 1 +\frac{5^{13} + 4}{\lambda/6 - 4} \qquad \text{(by \Cref{it:original-weight} of \Cref{prop:lightness-properties})} \\
			&= 1 + \frac{5^{13} + 4}{5^{19}\left(1 + 1/\eps\right) - 6} \leq 1 + 5^{-3}\eps \qquad ,
		\end{split}
	\end{equation}
	Using similar argument, $\frac{\nd(u', v')}{\nd(u, v)} \geq 1 - 5^{-3}\eps.$
	
	\emph{\Cref{it:node-to-augmented-dist}:} Let $x_i$ be the ancestor of $x$ at level $i$. Since $(y, z) \in \Aug_i(x)$, $y$ and $z$ are both in $NC[x_i]$. Hence, $\nd(x_i, y) \leq \lambda r_i$. By the triangle inequality, we have $\nd(x, y) \leq \nd(x, x_i) + \nd(x_i, y) \leq 2r_i + \lambda r_i \leq (\lambda + 2)r_i$ since $\nd(x, x_i) \leq 2r_i$ by \Cref{clm:leaf-dist}.
	
	\emph{\Cref{it:level-weight-relation}: } Let $i = \lvl(\hat{u})$. By the definition of cross edges, $\nd(\hat{u}, \hat{v}) \leq \lambda r_i$. By \Cref{it:original-weight} of \Cref{prop:lightness-properties}, $\nd(\hat{u}, \hat{v}) \geq \lambda r_i / 6$. Taking the logarithm of both sides, we obtain the desired inequality.
	
	\emph{\Cref{it:cross-edge-ancestor}: } We only need to show that \Cref{it:cross-edge-ancestor} is true when $(x', y')$ is the parent of $(x, y)$ ($i' = i + 1$), the result for any ancestor of $(x, y)$ follows by induction. By \Cref{clm:leaf-dist}, $\nd(x', x), \nd(y', y) \leq 2r_{i'}$. Using the triangle inequality,
	\begin{equation}
		\begin{split}
			\nd(x', y') &\leq \nd(x', x) + \nd(x, y) + \nd(y, y') \leq 2r_{i'} + \lambda r_i + 2r_{i'} \\
			&\leq \left(\lambda / 5 + 4\right)r_{i'} \leq \lambda r_{i'} \qquad \text{(since $r_{i'} = 5r_i$ and $\lambda > 5$)},
		\end{split}
	\end{equation}
	as claimed.
\end{proof}

By \Cref{it:good-cross-edge-long}, a good cross edge $(x, y)$ at level $i$ is always longer than $64r_i$. Hence, when \Cref{alg:VFT-spanner} discovers a good cross edge $(x, y)$ either in \cref{line:augm} or \cref{line:add-NCx}, $(x, y)$ is always added to $E^*$.

Given a path $P = (u_1, u_2, \ldots u_l)$ in $G$, a cross edge $(x, y)$ is a \emph{$P$-detour} if:

\begin{itemize}
	\item $(x, y)$ is a good cross edge of $(u_1, u_l)$ or 
	\item $(x, y)$ is a good cross edge of $(u_1, u_k)$ for some integer $k \in (1, l)$ and both $x$ and $y$ are complete.
\end{itemize}

We consider some properties of a $P$-detour:
\begin{property} \label{prop:detour}Let $u$ and $v$ be two points in $X$ and $P$ be the shortest path from $u$ to $v$ in $G$. Let $(x, y)$ be a $P$-detour in $E^*$ and $i = \lvl(x)$ ($= \lvl(y)$). We have the following properties:
	\begin{enumerate}
		\item \label{it:non-fault} There exists a non-faulty edge $(w, z)$ in $H$ such that $w \in \ball(x, 16r_i)$ and $z \in \ball(y, 16r_i)$.
		\item \label{it:radius-small}$r_i \leq 5^{-8}\eps\cdot\min\{\nd(u, v), \nd(x, y)\}$.
	\end{enumerate}
\end{property}

\begin{proof}
	\Cref{it:non-fault}: Consider the time we add $M(S(x), S(y))$ to $H$ in \cref{line:H-update}. We show that there is at least one remaining edge in $M(S(x), S(y))$ after $f$ points (not including $u$ or $v$) is deleted.
	
	If $x$ and $y$ are complete, $S(x)$ and $S(y)$ contains $f + 1$ points each, and by \Cref{def:bipartite-conn}, $M(S(x), S(y))$ is a matching of size $f + 1$. Hence, when we delete $f$ points, there is still at least one remaining edge in $M(S(x) , S(y))$. By \Cref{alg:Surrogate}, $S(x) \subseteq \ball(x, 16r_i)$ and $S(y) \subseteq \ball(y, 16r_i)$, which completes our proof. 
	
	If either $S(x)$ or $S(y)$ is incomplete, then $x$ and $y$ are ancestors of $(u, 0)$ and $(v, 0)$ by the definition of a $P$-detour. By \Cref{def:bipartite-conn}, $M(S(x), S(y))$ is the complete bipartite graph between $S(x)$ and $S(y)$. Hence, we only need to show that $S(x) \setminus F$ and $S(y) \setminus F$ are non-empty. This is true if $|S(x)| = f + 1$ or $|S(y)| = f + 1$. If $|S(x)| \leq f$, since the surrogate set of an incomplete node must contain all of its leaves, $S(x)$ must contain $u$. Then, $S(x) \setminus F$ is non-empty since $u \not \in F$. Similarly, $S(y) \setminus F$ is non-empty.
	
	\Cref{it:radius-small}: Since $(x, y)$ is a good cross edge of $(u, u')$ for some $u' \in P$, $(x, y)$ is a $\kappa$-cross edge of $(u, u')$ with $\kappa \leq 10$. From \Cref{it:good-cross-edge-long}, we have $\nd(x, y) \geq r_i\left(\frac{\lambda}{6\cdot5^{\kappa}}- 4\right)$. Hence,
	\begin{equation}
		\label{eq:ri-fi-bound}
		\begin{split}
			r_i &\leq \nd(x, y) \cdot \left(\frac{\lambda}{6 \cdot 5^{\kappa}} - 4\right) ^ {-1} \leq \frac{\nd(u, u')}{1 + 5^{-3}\eps} \cdot  \left(\frac{\lambda}{6\cdot 5^{\kappa}} - 4\right) ^ {-1} \qquad \text{(by It. \ref{it:good-approx-weight} of Prop. \ref{prop:connectivity})}\\
			&\leq \frac{\dist_G(u, u')}{1 + 5^{-3}\eps} \cdot  \left(\frac{\lambda}{6 \cdot 5^{\kappa}} - 4\right) ^ {-1} \leq  \frac{\dist_G(u, v)}{1 + 5^{-3}\eps} \cdot  \left(\frac{\lambda}{6 \cdot 5^{\kappa}} - 4\right) ^ {-1} \qquad\text{(since $P$ contains $u'$)} \\
			&\leq \frac{(1 + \eps)\nd(u, v)}{1 + 5^{-3}\eps}\left(\frac{\lambda}{6\cdot 5^{\kappa}} - 4\right) ^ {-1} \leq \frac{(1 + \eps)\nd(u, v)}{1 + 5^{-3}\eps}\left(\frac{\lambda}{6 \cdot 5^{10}} - 4\right) ^ {-1} \qquad \text{(since $\kappa \leq 10$)}\\
			&= \frac{1 + \eps}{1 + 5^{-3}\eps} \cdot \frac{6\cdot5^{10}}{5^{20}(1 + \eps^{-1}) - 24\cdot5^{10}}\nd(u, v) \leq 5^{-8}\eps \nd(u, v)
		\end{split}
	\end{equation}
	The first equation of \Cref{eq:ri-fi-bound} also gives us:
	\begin{equation}
		\begin{split}
			r_i &\leq \nd(x, y) \left(\frac{\lambda}{6\cdot 5^{\kappa}} - 4\right) ^ {-1} \leq \nd(x, y) \cdot \left(\frac{\lambda}{6\cdot5^{10}} - 4\right) ^ {-1} \qquad \text{(since $\kappa \leq 10$)}\\
			&= \frac{6 \cdot 5^{10}}{5^{20}(1 + \eps^{-1}) - 24 \cdot 5^{10}}\nd(x, y) \leq 5^{-8}\eps\nd(x, y),
		\end{split}
	\end{equation}
	as claimed.
\end{proof}

We have the following lemma:

\begin{lemma}
	\label{lm:detour-exist}
	For every pair of points $(u, v)$ in $X$, let $P$ be any shortest path from $u$ to $v$ in $G$. Then, there exists a $P$-detour in $E^*$. 
\end{lemma}

We now show that \Cref{lm:vertex-disjoint-path} follows from \Cref{lm:detour-exist}.

\begin{proof}[Proof of \Cref{lm:vertex-disjoint-path}]
	Let $u$ and $v$ be two points of $X \setminus F$ and $P =  (u = u_1, u_2, \ldots u_l = v)$ be a shortest path from $u$ to $v$ in $G$. We induct on the length $\nd(u, v)$. If $\nd(u, v) \leq \lambda$ (the minimum distance between two points is $64$), we claim that the cross edge between $(u, 0)$ and $(v, 0)$ is in $E^*$. Observe that both $(u, 0)$ and $(v, 0)$ are incomplete since each of them has one leaf with degree $0$ before iteration $0$. Hence, by \cref{line:begin-add-incomplete}--\ref{line:end-add-incomplete}, the cross edge between $(u, 0)$ and $(v, 0)$ is in $E^*$ since $\nd(u, v) \geq 64 \geq 64r_0$. By \cref{line:H-update}, $M(S(u, 0), S(v, 0))$ is added to $H$. By \cref{line:add-leaf-S} of \Cref{alg:Surrogate} $S(u, 0) = \{u\}$ and $S(v, 0) = \{v\}$. Thus, $(u, v) \in E(H)$.
	
	Assume that for any two points in $X \setminus F$ with distance less than $\nd(u, v)$, there is an $(1 + 5\eps)$-spanner path between them in $H - F$. By \Cref{lm:detour-exist}, there exist a $P$-detour $(x_1, y_1)$. Let $i_1 = \lvl(x_1)$. We consider two cases:
	
	Case $1$: $(x_1, y_1)$ is a good cross edge of $(u, v)$. We claim that there is a path from $u$ to $v$ with total weight less than $(1 + 5\eps)\nd(u, v)$. Let  $u', v'$ be two points in $\ball(x_1, 16r_{i_1})$ and $\ball(y_1, 16r_{i_1})$ such that $(u', v') \in E(H - F)$. $u'$ and $v'$ exist by \Cref{it:non-fault} of \Cref{prop:detour}.

	Since $(u, 0)$ is a leaf of $x_1$, $\nd(u, x_1) \leq 2r_{i_1}$ by \Cref{clm:leaf-dist}. By the triangle inequality,
	\begin{equation}
		\label{eq:u'u}
		\nd(u', u) \leq \nd(u, x_1) + \nd(x_1, u') \leq 2r_{i_1} + 16r_{i_1} = 18r_{i_1}
	\end{equation}
	By \Cref{it:radius-small}, $r_{i_1} \leq 5^{-8}\eps\nd(u, v)$. Plugging in \Cref{eq:u'u}, we obtain $\nd(u', u) \leq 18 \cdot 5^{-8}\eps\nd(u, v) < \nd(u, v)$. By our induction hypothesis, $\dist_{H - F}(u', u) \leq (1 + 5\eps)\nd(u, u')$. Using the same argument, we get $\nd(v, v') \leq 18r_{i_1}$ and $\dist_{H - F}(v, v') \leq (1 + 5\eps)\nd(v, v')$. By the triangle inequality,
	\begin{equation}
		\label{eq:dist-good-cross}
		\begin{split}
			\dist_{H - F}(u, v) & \leq \dist_{H - F}(u, u') + \nd(u', v') + \dist_{H - F}(v', v)\\
				& \leq (1 + 5\eps)\underbrace{\nd(u, u')}_{\leq 18r_{i_1}} +  \nd(u', v') + (1 + 5\eps)\underbrace{\nd(v', v)}_{\leq 18 r_{i_1}}\\
				&\leq (1 + 5\eps)36r_{i_1} + \nd(u', v') \leq (1 + 5\eps)36r_{i_1} + (1 + 5^{-3}\eps)\nd(u, v) \qquad \text{(by It. \ref{it:good-approx-weight} of Prop. \ref{prop:connectivity})}\\
				&\leq (1 + 5\eps)36 \cdot 5^{-8} \eps \nd(u, v) + (1 + 5^{-3}\eps)\nd(u, v) \qquad \text{(by It. \ref{it:radius-small} of Prop. \ref{prop:detour})}\\
				&\leq (1 + 2\eps)\nd(u, v) \leq (1 + 5\eps)\nd(u, v) \qquad \text{(since $\eps \leq 1/20$)},
		\end{split}
	\end{equation}
	as claimed.
	
	Case $2$: $(x_1, y_1)$ is a good cross edge of $(u_1, u_{a_1})$ for some $a_1 \in (1, l)$. For each $j \in [1, l]$, let $P_j$ be the subpath $(u_j, u_{j + 1}, \ldots u_l)$. By \Cref{lm:detour-exist}, there exists a $P_{a_1}$-detour, denoted by $(x_2, y_2)$, such that $(x_2, y_2)$ is a good cross edge of $(u_{a_1}, u_{a_2})$ for some $a_2 > a_1$. Recursively, for each $k > 0$, we find the $P_{a_k}$-detour, denoted by $(x_{k + 1}, y_{k + 1})$, such that $(x_{k + 1}, y_{k + 1})$ is a good cross edge of $(u_{a_k}, u_{a_{k + 1}})$ until $a_{k + 1} = l$. Assume that we have $s + 1$ of such detours with $s > 0$. Let $D = \{(x_1, y_1), (x_2, y_2), \ldots (x_s, y_s), (x_{s+1}, y_{s+1})\}$. 
	\begin{observation}{\label{obs:detour-decom}}We have the following properties: 
		\begin{enumerate}
			\item For every $k \in [1, s + 1]$, $(x_k, y_k)$ is a good cross edge of $(u_{a_k}, u_{a_{k + 1}})$.
			\item For every $k \in [1, s]$, $x_k$ and $y_k$ are complete. 
			\item \label{it:anc-same-branch} For every $k \in [1, s]$, $x_{k + 1}$ and $y_k$ are ancestors of $(u_{a_k}, 0)$. Hence, either $x_{k+1}$ or $y_k$ is the ancestor of the other.
			\item $a_1 < a_2 < a_3 < \ldots < a_{s + 1} = l$.
		\end{enumerate}
	\end{observation}

	For every $k \in [1, s + 1]$, let $i_k = \lvl(x_k)$. For $k \in [1, s]$, let $w_{k}$ and $z_{k}$ be two points in $\ball(x_k, 16r_{i_k})$ and $\ball(y_k, 16r_{i_k})$, respectively, such that $(w_{k}, z_{k}) \in E(H)$. $w_k$ and $z_k$ exist by \Cref{it:non-fault} of \Cref{prop:detour}. See \Cref{fig:lm13-1} for an illustration.
	\begin{center}
		\begin{figure}[htp!]
			\includegraphics[width=1\textwidth]{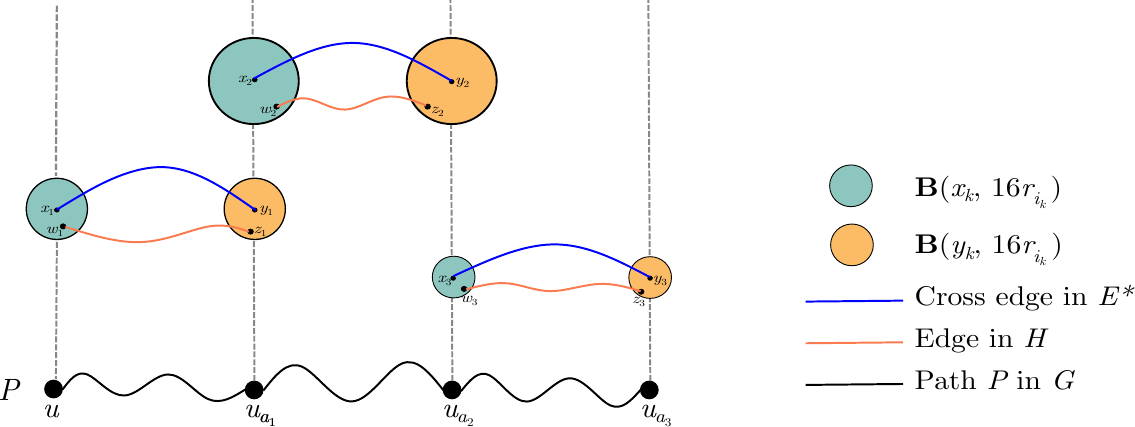}
			\caption{An illustration of $(x_k, y_k), (w_k, z_k)$ for $i \in [1, 3]$. Here, it might be misleading that $u_{a_k}$ is far from $x_{k + 1}$ and $y_k$. In fact, $u_{a_k}$ is in both $\ball(x_{k + 1}, 16r_{k + 1})$ (the blue ball) and $\ball(y_k, 16r_k)$ (the yellow ball). Furthermore, the higher ball in $\{\ball(x_{k + 1}, 16r_{k + 1}), \ball(y_k, 16r_k)\}$ contains the other.}
			\label{fig:lm13-1}
		\end{figure}
	\end{center}
	
	 Let $a_0 = 1$. Using the triangle inequality, we have:
	\begin{equation}
		\label{eq:d(uzs)}
		\begin{split}
			\dist_{H - F}(u, z_{s}) &\leq \dist_{H - F}(u, w_{1}) + \sum_{k = 1}^s\dist_{H - F}(w_{k}, z_{k}) + \sum_{k = 1}^{s - 1}\dist_{H - F}(z_{k}, w_{k + 1})\\
			&= \dist_{H - F}(u, w_{1}) + \sum_{k = 1}^s\nd(w_{k}, z_{k}) + \sum_{k = 1}^{s - 1}\dist_{H - F}(z_{k}, w_{k + 1})		
		\end{split}
	\end{equation} 
	By \Cref{it:good-approx-weight} of \Cref{prop:connectivity}, $\nd(w_k, z_k) \leq (1 + 5^{-3}\eps)\nd(u_{a_{k - 1}}, u_{a_k})$ for every $k \in [1, s]$. Hence,
	\begin{equation}
		\label{eq:sum-zw}
		\begin{split}
			\sum_{k = 1}^s\nd(w_k, z_k)& \leq (1 + 5^{-3})\sum_{k = 1}^s\nd(u_{a_{k - 1}}, u_{a_k})\leq (1 + 5^{-3})\sum_{k = 1}^s\dist_G(u_{a_{k - 1}}, u_{a_k})\\
			&\leq (1 + 5^{-3}\eps)\dist_G(u, u_{a_s})\qquad.
		\end{split}
	\end{equation}
	The last equation holds since $u_{a_0} = u_1 = u$ and the shortest path from $u$ to $u_{a_s}$ in $G$, which is $P$, contains $u_{a_1}, u_{a_2} \ldots u_{a_{s - 1}}$. 
	
	We claim that $\nd(u, w_1)$ and $\nd(z_k, w_{k + 1})$ for $k \in [1, s - 1]$ are less than $\nd(u, v)$. Hence, $\dist_{H - F}(u, w_1)$ and $\dist_{H - F}(z_k, w_{k + 1})$ are approximate their distances in $(X, \nd)$. Since $\nd(u, x_1) \leq 2r_{i_1}$ by \Cref{clm:leaf-dist} and $w_1 \in \ball(x, 16r_{i_1})$, using the triangle inequality, we have:
	
	 \begin{equation}
	 	\label{eq:d(uw1)}
	 	\nd(u, w_1) \leq \nd(u, x_1) + \nd(x_1, w_1) \leq 18r_{i_1} < 32r_{i_1}.
	 \end{equation}
	 For every $k \in [1, s - 1]$, $z_k \in \ball(y_k, 16r_{i_k})$ and $w_{k + 1} \in \ball(x_{k + 1}, 16r_{i_{k + 1}})$.  By \Cref{it:anc-same-branch} of \Cref{obs:detour-decom}, either $y_k$ or $x_{k + 1}$ is the ancestor of the other. If $y_k$ is an ancestor of $x_{k + 1}$, then $\ball(x_{k + 1}, 16r_{i_{k + 1}}) \subseteq \ball(y_k, 16r_{i_k})$ by \Cref{it:descendant-pool} of \Cref{prop:pool-prop}, implying that $w_{k + 1} \in \ball(y_k, 16r_{i_k})$. Hence, $\nd(z_k, w_{k + 1}) \leq 16r_{i_k} + 16r_{i_k} = 32r_{i_k}$. Otherwise, $x_{k + 1}$ is an ancestor of $y_k$. Using similar argument, we get $\nd(z_k, w_{k + 1}) \leq 32r_{i_{k + 1}}$. Therefore, we obtain:
	 \begin{equation}
	 	\label{eq:d(zkwk+1)}
	 	\nd(z_k, w_{k + 1}) \leq 32\max\{r_{i_k}, r_{i_{k  +1}}\}.
	 \end{equation}	
 	
 	We then prove that $\nd(u, v) > 32 r_k$ for every $k \in [1, s]$. Since $G$ is a $(1 + \eps)$-spanner of $X$, for every $k \in [1, s + 1]$, we have: 
	 \begin{equation}
	 	\label{eq:d(uv)>32rk}
	 	\begin{split}
	 		\nd(u, v) &\geq \frac{\dist_G(u, v)}{1 + \eps} \geq \frac{\nd(u_{a_{k}}, u_{a_{k + 1}})}{1 + \eps} \qquad\text{(since $P$ contains $(u_{a_k}, u_{a_{k + 1}})$)}\\
	 		&\geq \frac{\nd(x_k, y_k)}{(1 + 5^{-3}\eps)(1 + \eps)}\qquad\text{(by It. \ref{it:good-approx-weight} of Prop. \ref{prop:connectivity})}\\ 
	 		&\geq \frac{5^8\eps^{-1}r_{i_k}}{(1 + \eps)^2}. \qquad\text{(by It. \ref{it:radius-small} of Prop. \ref{prop:detour})},
	 	\end{split}		
	\end{equation} 
	
	implying that $\nd(u, v) > 32r_{i_k}$. By \Cref{eq:d(uw1)}, \Cref{eq:d(zkwk+1)}, \Cref{eq:d(uv)>32rk}, we get $\nd(u, w_1)$ and $\nd(z_k, w_{k + 1})$ are less than $\nd(u, v)$. Hence, by the induction hypothesis, $\dist_{H - F}(u, w_1) \leq (1 + 5\eps)\nd(u, w_1)$ and $\dist_{H - F}(z_k, w_{k + 1}) \leq (1 + 5\eps)\nd(z_k, w_{k + 1})$. Then, \Cref{eq:d(uw1)} implies that:
	\begin{equation}
		\label{eq:duw1-2}
		\begin{split}
			\dist_{H - F}(u, w_1) \leq (1 + 5\eps)\nd(u, w_1) \leq (1 + 5\eps)32r_{i_1} \leq (1 + 5\eps)32\sum_{k = 1}^{s}r_{i_k}.
		\end{split}
	\end{equation}
	Similarly, \Cref{eq:d(zkwk+1)} implies that:
	\begin{equation}
		\label{eq:sum-zw2}
		\sum_{k = 1}^{s - 1}\dist_{H - F}(z_{k}, w_{k + 1})	\leq (1 + 5\eps)\sum_{k = 1}^{s - 1}\nd(z_k, w_{k + 1}) \leq (1 + 5\eps)32\sum_{k = 1}^{s - 1}\max\{r_{i_k}, r_{i_{k + 1}}\} \leq (1 + 5\eps)64\sum_{i = 1}^kr_{i_k}.
	\end{equation}
	We then bound $\sum_{k = 1}^sr_{i_k}$. For every index $k \in [1, s]$, by \Cref{it:radius-small} of \Cref{prop:detour}, $r_{i_k} \leq 5^{-8}\eps\nd(x_k, y_k)$. Since $(x_k, y_k)$ is a good cross edge of $(u_{a_{k - 1}}, u_{k})$, $r_{i_k} \leq 5^{-8}\eps\nd(x_k, y_k)\leq 5^{-8}\eps(1 + \eps)\nd(u_{a_{k - 1}}, u_{k})$ by \Cref{it:good-cross-edge-long} of \Cref{prop:connectivity}. Hence,
	\begin{equation}
		\label{eq:sum-ri}
		\begin{split}
			\sum_{k = 1}^{s}r_{i_k} &\leq 5^{-8}\eps(1 + \eps)\sum_{k = 1}^s\nd(u_{a_{k - 1}}, u_{a_k}) \leq  5^{-8}\eps(1 + \eps)\dist_G(u, u_{a_s}) \qquad,
		\end{split}
	\end{equation}
	since the shortest path between $u$ to $u_{a_s}$ in $G$ contains $u_{a_1}, u_{a_2}, \ldots u_{a_{s - 1}}$. Plugging \Cref{eq:sum-zw}, \Cref{eq:duw1-2} and \Cref{eq:sum-zw2} in \Cref{eq:d(uzs)}, we get:
	\begin{equation}
		\label{eq:duzs2}
		\begin{split}
			\dist_{H - F}(u, z_s) &\leq (1 + 5^{-3}\eps)\dist_G(u, u_{a_s}) + (1 + 5\eps) \cdot 96\sum_{k = 1}^{s}r_{i_k}\\
			&\leq (1 + 5^{-3}\eps)\dist_G(u, u_{a_s}) + \underbrace{(1 + 5\eps)96\cdot5^{-8}}_{\leq 1/2 ~\text{ as $\eps \leq 1/20$}}\eps(1 + \eps)\dist_G(u, u_{a_s}) ~\text{(by Eq. \ref{eq:sum-ri})}\\
			&\leq(1 + 3\eps / 4)\dist_G(u, u_{a_s}) \leq (1 + 3\eps/4)(1 + \eps)\nd(u, u_{a_s})~\text{(since $G$ is a $(1 + \eps)$-spanner of $X$)}\\
			&\leq (1 + 2\eps)\nd(u, u_{a_s})
		\end{split}
	\end{equation}
	
	The proof of \Cref{eq:duzs2} contains all the key ideas of this lemma. The remaining part of the proof focuses on how to deal with a tricky corner case when $x_{s + 1}$ is incomplete.  
	
	We consider $(x_{s + 1}, y_{s + 1})$. If $x_{s + 1}$ is complete, let $w_{s + 1}$ and $z_{s + 1}$ be two points in $\ball(x_{s + 1}, 4r_{i_{s + 1}})$ and $\ball(y_{s + 1}, r_{i_{s + 1}})$ such that $(w_{s + 1}, z_{s + 1}) \in E(H - F)$. $w_{s + 1}$ and $z_{s + 1}$ exist by \Cref{it:non-fault} of \Cref{prop:detour}. Using the same argument in \Cref{eq:duzs2}, we have $\dist_{H - F}(u, z_{s + 1}) \leq (1 + 2\eps)\nd(u, v)$. By \Cref{clm:leaf-dist}, $\nd(v, y_{s + 1}) \leq 2r_{i_{s + 1}}$.  Using the triangle inequality, we get:
	\begin{equation}
		\label{eq:vzs+1}
		\begin{split}
			\nd(v, z_{s + 1}) &\leq \nd(v, y_{s + 1}) + \nd(y_{s + 1}, z_{s + 1}) \leq 2r_{i_{s + 1}} + 16r_{i_{s + 1}} \\
			&= 18r_{i_{s + 1}} \leq 18(1 + \eps)^25^{-8}\eps\nd(u, v) \qquad \text{(by Eq. \ref{eq:d(uv)>32rk})}\\
			&\leq \eps\nd(u, v) / 2 \qquad,
		\end{split}
	\end{equation}
	implying that $\nd(v, z_{s + 1}) < \nd(u, v)$. By our induction hypothesis, $\dist_{H - F}(v, z_{s + 1}) \leq (1 + 5\eps)\nd(v, z_{s + 1}) \leq (1 + 5\eps)\eps\nd(u, v)/2$. Thus, 
	\begin{equation}
		\begin{split}
			\dist_{H - F}(u, v) &\leq \dist_{H - F}(u, z_{s + 1}) + \dist_{H - F}(z_{s + 1}, v)\\
			&\leq (1 + 2\eps)\nd(u, v) + (1 + 5\eps)\eps\nd(u, v)/2 \leq (1 + 3\eps)\nd(u, v) \qquad,
		\end{split}
	\end{equation}
	as $\eps \leq 1/20$.
	
	The last case is when $x_{s + 1}$ is incomplete. Let $t = \LCA(u_{a_s})$ and $i = \lvl(z)$ and $\tilde{v}$ be the ancestor of $(v, 0)$ at level $i$. Since $(t, \tilde{v})$ is an ancestor of a cross edge, $\nd(t, \tilde{v}) \leq \lambda r_i$ by \Cref{it:cross-edge-ancestor} of \Cref{prop:connectivity}. Consider two cases:
	
	\begin{center}
		\begin{figure}[htp!]
			\includegraphics[width=1\textwidth]{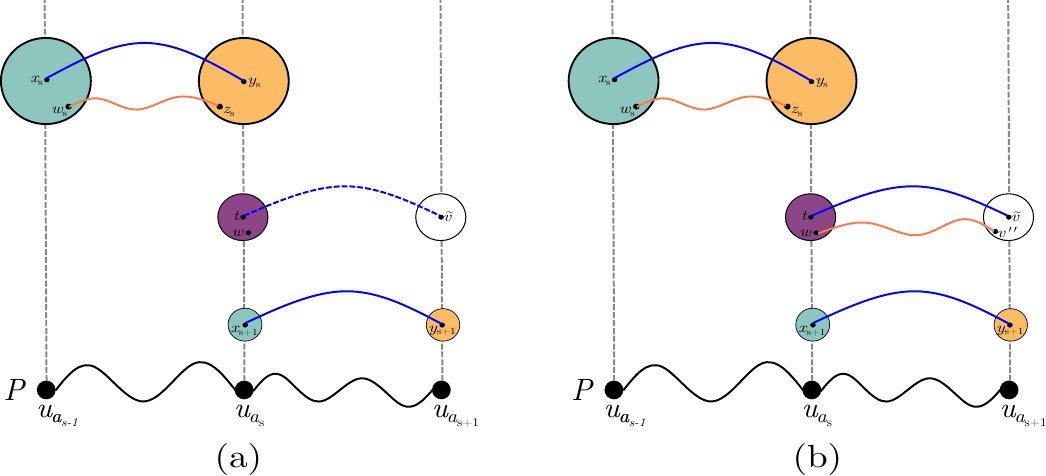}
			\caption{An illustration of two cases. The blue, orange edges and the cyan and yellow balls have the same meaning as those in \Cref{fig:lm13-1}. The purple node is the $\LCA$ of $u_{a_s}$ and the white node $\tilde{v}$ is the ancestor of $v = u_{a_{s + 1}}$ at the same level. The cross edge between $t$ and $\tilde{v}$ is not in $E^*$ for the first case $(a)$. In figure $(b)$, $(t, \tilde{v}) \in E^*$ and hence there is an edge of $H - F$ in the bipartite connection between $S(t)$ and $S(\tilde{v})$.}
			\label{fig:lm13-2}
		\end{figure}
	\end{center}
	
	Case $1$: $\nd(t, \tilde{v}) < 64r_i$. Let $w$ be a point in $\ball(t, 16r_i) \setminus F$. $w$ exists since by \Cref{lm:complete-ball}, $\ball(t, 4r_i)$ contains at least $4f + 4$ points. See \Cref{fig:lm13-2} (a) for an illustration. Since $i \leq i_s$, we get $y_s$ is an ancestor of $z$, implying that $\ball(t, 16r_i) \subseteq \ball(y_s, 16r_{i_s})$ by \Cref{it:descendant-pool} of \Cref{prop:pool-prop}. Hence, $w \in \ball(y_s, 16r_{i_s})$. Both $z_s$ and $w$ are in $\ball(y_{s}, 16r_{i_s})$, $\nd(z_s, w) \leq 32r_{i_s} < \nd(u, v)$ by \Cref{eq:d(uv)>32rk}, which implies:
	\begin{equation}
		\label{eq:dzsw}
		\dist_{H - F}(z_s, w) \leq (1 + 5\eps)\nd(z_s, w) \leq (1 + 5\eps)32r_{i_s} 
	\end{equation}
 	by the induction hypothesis. On the other hand, 
	\begin{equation}
		\begin{split}
			\nd(w, v) &\leq \nd(w, t) + \nd(t, \tilde{v}) + \nd(\tilde{v}, v) \\
			&\leq \nd(w, t) + 64r_i + \nd(\tilde{v}, v) \qquad\text{(by our assumption)} \\
			&\leq 16r_i + 64r_i + \nd(\tilde{v}, v) \leq 80r_i + \nd(\tilde{v}, v) \qquad \text{(since $w \in \ball(t, 16r_i)$)}\\
			&\leq 96r_i \leq 96r_{i_s} \qquad\text{(since $\nd(\tilde{v}, v) \leq 2r_i$ by \Cref{clm:leaf-dist})}.
		\end{split}
	\end{equation}
	Since $r_{i_s} \leq 2^{-8}\eps(1 + \eps)^2 \nd(u, v)$ (\Cref{eq:d(uv)>32rk}), we get $\nd(w, v) \leq 96 \cdot2^{-8}\eps (1 + \eps)^2 \nd(u, v) < \nd(u, v)$ as $\eps \leq 1/20$. By the induction hypothesis, 
	\begin{equation}
		\label{eq:dwv}
		\dist_{H - F}(w, v) \leq (1 + 5\eps)\nd(w, v) \leq (1 + 5\eps)96r_{i_s}
	\end{equation}
	Hence,
	\begin{equation}
		\begin{split}
			\dist_{H - F}(u, v) &\leq \dist_{H - F}(u, z_s) + \dist_{H - F}(z_s, w) + \dist_{H - F}(w, v)\\
			&\leq (1 + 2\eps)\underbrace{\nd(u, u_{a_s})}_{\leq \weight(P) \leq (1 + \eps)\nd(u, v)} + \dist_{H - F}(z_s, w) + \dist_{H - F}(w, v) \qquad \text{(by Eq. \ref{eq:duzs2})}\\
			&\leq  (1 + 2\eps)(1 + \eps)\nd(u, v) + (1 + 5\eps)32r_{i_s} + (1 + 5\eps)96r_{i_s} \qquad\text{(by Eq. \ref{eq:dzsw} and Eq. \ref{eq:dwv})}\\
			&\leq (1 + 4\eps)\nd(u, v) + 128(1 + 5\eps)r_{i_s}\\
			& \leq (1 + 4\eps)\nd(u, v) + 128(1 + 5\eps) \cdot 5^{-8}\eps(1 + \eps)^2\nd(u, v) \qquad\text{(by Eq. \ref{eq:d(uv)>32rk})}\\
			&\leq (1 + 4\eps)\nd(u,v ) + \eps \nd(u, v) = (1 + 5\eps)\nd(u, v) \qquad\text{(since $\eps \leq 1/20$)}, 
		\end{split}
	\end{equation}
	as claimed.
	
	Case $2$: $\nd(t, \tilde{v}) \geq 64r_i$. Then, $(t, \tilde{v}) \in E^*$ since $\nd(t, \tilde{v}) \leq \lambda r_{i_s}$ and $\Aug_i(t) \subseteq E^*$ as $z$ has at least one incomplete child (\cref{line:add-NCx}). Let $w$ and $v''$ be two points in $S(t)$ and $S(\tilde{v})$ such that $(w, v'') \in E(H - F)$. $w$ and $v''$ exist since $S(t)$ contains $f + 1$ points and $S(\tilde{v})$ contains either $f + 1$ points or $v$. See \Cref{fig:lm13-2} (b) for an illustration. 
	
	Since both $w$ and $u_{a_s}$ are in $\ball(t, 16r_i)$ ($\nd(u_{a_s}, t) \leq 2r_i$ by \Cref{clm:leaf-dist}), we have $\nd(w, u_{a_s}) \leq 32r_i \leq 32r_{i_{s}}$. Similarly, $\nd(z_s, w), \nd(v'', v) \leq 32r_{i_s}$. Since $32r_{i_s} < \nd(u, v)$ by \Cref{eq:d(uv)>32rk}, using the induction hypothesis, we obtain:
	 \begin{equation*}
	 	\dist_{H - F}(z_s, w) \leq (1 + 5\eps)\nd(z_s, w) \leq (1 + 5\eps)32r_{i_s} \qquad.
	 \end{equation*}
	 Similarly, $\dist_{H - F}(v'', v) \leq (1 + 5\eps)32r_{i_s}$. Using the triangle inequality, we have:
	
	\begin{equation}
		\label{eq:last-cross}
		\begin{split}
			\dist_{H - F}(z_s, v) &\leq \dist_{H - F}(z_s, w) + \dist_{H - F}(w, v'') + \dist_{H - F}(v'', v)\\
			&\leq (1+ 5\eps)32r_{i_s} + \nd(w, v'') + (1 + 5\eps)32r_{i_s} \qquad\text{(since $(w, v'')$ is in $H - F$)} \\
			&\leq (1 + 5\eps)64r_{i_s} + \nd(w, u_{a_s}) + \nd(u_{a_s}, v) + \nd(v, v'')\\
			&\leq (1 + 5\eps)64r_{i_s} + 32r_{i_s} + \nd(u_{a_s}, v) + 32r_{i_s} \qquad\text{(since $w, u_{a_s} \in \ball(y_s, 16r_{i_2})$ and $v, v'' \in \ball(\tilde{v}, 16r_i)$)}\\
			&\leq (1 + 5\eps)128r_{i_s} + \nd(u_{a_s}, v).
		\end{split}
	\end{equation} 
	Consider a path from $u$ to $v$ passing through $z_s$, we have the following bound based on the triangle inequality: 
	\begin{equation}
		\label{eq:incomplete-last-cross}
		\begin{split}
			\dist_{H - F}(u, v) &\leq \dist_{H - F}(u, z_s) + \dist_{H - F}(z_s, v)\\
			&\leq (1 + 2\eps) \nd(u, u_{a_s}) + (1 + 5\eps)128r_{i_s} + \nd(u_{a_s}, v)\qquad\text{(by Eq. \ref{eq:d(uzs)} and Eq. \ref{eq:last-cross})}\\
			&\leq (1 + 2\eps) (\nd(u, u_{a_s}) + \nd(u_{a_s}, v)) + (1 + 5\eps)128r_{i_s}\\
			&\leq (1 + 2\eps)(1 + \eps)\nd(u, v) + (1 + 5\eps)128 \cdot 5^{-8}\eps(1 + \eps)^2\nd(u, v)\qquad\text{(by Eq. \ref{eq:d(uv)>32rk})}\\
			&\leq (1 + 5\eps)\nd(u, v) \qquad \text{(since $\eps \leq 1/20$)},
		\end{split}
	\end{equation}
	as desired.
\end{proof}

In the rest of this section, we focus on proving \Cref{lm:detour-exist}. We first define some notation. Given a path $P = (u_1, u_2, \ldots u_l)$ in $G$. A cross edge $(x, y)$ is a \textit{$P$-jump} if:

\begin{enumerate}
	\item $x$ is a complete ancestor of $(u_1, 0)$ and $y$ is an ancestor of $(u_j, 0)$ for some $j \in (1, l]$.
	\item $(x, y)$ is a $\kappa$-cross edge of $(u_1, u_j)$ with $0 \leq \kappa \leq 2$.
\end{enumerate}

A $P$-jump is different from a $P$-detour. In general, a $P$-detour has two complete end nodes while the definition of $P$-jump only guarantees one end to be complete. A $P$-jump $(x, y)$ with incomplete $y$ is a $P$-detour if and only if $y$ is an ancestor of $(u_l, 0)$. 

For every set of cross edges $A$, let $\Pa(A)$ be the set of all cross edges $(x', y')$ such that there exists a child $(x, y)$ of $(x', y')$ in $A$. We call $\Pa(A)$ the parent set of $A$. Given a positive integer $i$, let $\Pa^1(A) = \Pa(A)$ and $\Pa^i(A) = \Pa(\Pa^{i - 1}(A))$. For each node $x$ and a nonnegative integer $i$, we also use the notation $\Pa^i(x)$ for the ancestor of $x$ at level $\lvl(x) + i$. We prove the following lemma:

\begin{lemma}
	\label{lm:detour-edge}
	Let  $u$ and $v$ be two points in $X$, $P = (u = u_1, u_2, \ldots u_l = v)$ be the shortest path between $u$ and $v$ in $G$, $i_1$ be the level of $\LCA(u_1)$ and $v_1$ be the ancestor $v$ at level $i_1$. If $\nd(\LCA(u_1), v_1) > \lambda r_{i_1}$, then there exists a $P$-jump $(x_0, y_0)$ in $E^*$. Furthermore, $\Pa^2(\Aug(x_0, 0, \log{\lambda}))$ and $\Pa^2(\Aug(y_0, 0, \log{\lambda}))$ are subsets of $E^*$.
\end{lemma} 

To prove \Cref{lm:detour-edge}, we first show that if a cross edge $(y, z)$ is an augmented cross edge of a node $x$, then some particular augmmented cross edges of $y$ and $z$ are also  augmented cross edges of $x$. 

\begin{claim}
	\label{clm:augmented}
	Let $x$ be a node in $T$, $\gamma_1$ and $\gamma_2$ be two non-negative integers such that $\gamma_1 < \gamma_2 - 2$ and $(y, z)$ is a cross edge in $\Aug(x, 0, \gamma_1)$. Then, $\Pa^2(\Aug(y, 0, \gamma_2 - \gamma_1 - 2))$ and $\Pa^2(\Aug(z, 0, \gamma_2 - \gamma_1 - 2))$ are subsets of $\Aug(x, 0, \gamma_2)$.
\end{claim}

\begin{proof}
	By symmetric property, we only need to show $\Pa^2(\Aug(y, 0,  \gamma_2 - \gamma_1 - 2)) \subseteq \Aug(x, 0, \gamma_2)$. Let $i_x$ and $i_y$ be the levels of $x$ and $y$, $i_x \leq i_y \leq i_x + \gamma_1$. For each cross edge $(z, w) \in \Aug(y, 0, \gamma_2 - \gamma_1 - 2)$, let $i_{zw}$ be the level of $z$ and $w$ ($i_y \leq i_{zw} \leq i_y + \gamma_2 - \gamma_1 - 2$), $z' = \Pa^2(z)$ and $w' = \Pa^2(w)$. We prove that $(z', w') \in \Aug(x, 0, \gamma_2)$.  
	
	Since $(z, w)$ is a cross edge, its ancestor $(z', w')$ is also a cross edge by \Cref{it:cross-edge-ancestor} of \Cref{prop:connectivity}. Observe that $\lvl(z') = i_{zw} + 2 \leq i_y + \gamma_2 - \gamma_1 \leq i_x + \gamma_2$. Then, to finish the proof of $(z', w') \in \Aug(x, 0, \gamma_2)$, we only need to show that both $z'$ and $w'$ are in the cross neighborhood of the ancestor of $x$ at level $\lvl(z')$.
	
	Let $x'$ and $y'$ be the ancestors of $x$ and $y$ at level $\lvl(z') = i_{zw} + 2$. By \Cref{it:node-to-augmented-dist} of \Cref{prop:connectivity}, 
	\begin{equation}
		\label{eq:dxyz}
		\nd(x, y) \leq (\lambda + 2) r_{i_y}\qquad\text{and} \qquad\nd(y, z) \leq (\lambda + 2) r_{i_{zw}}
	\end{equation}
	We show that $\nd(x', z') \leq \lambda r_{i_{zw} + 2}$ and hence $z' \in NC[x']$. By the triangle inequality, we have:
	\begin{equation}
		\label{eq:approx-x'-z'}
		\begin{split}
			\nd(x', z') &\leq \underbrace{\nd(x', x)}_{\leq 2r_{\lvl(x')}} + \nd(x, z) + \underbrace{\nd(z, z')}_{\leq 2r_{\lvl(z')}} \leq 4r_{i_{zw} + 2} + \nd(x, z) \qquad\text{(since $\lvl(x') = \lvl(z') = i_{zw} + 2$)}\\
			&\leq 4r_{i_{zw} + 2} + \nd(x, y) + \nd(y, z) \qquad \text{(by triangle inequality)}\\
			&\leq 4r_{i_{zw} + 2} + (\lambda + 2)r_{i_y} + (\lambda + 2)r_{i_{zw}} \qquad \text{(by Eq. \Cref{eq:dxyz})}\\
			&\leq \left(4 + \frac{\lambda + 2}{25} + \frac{\lambda + 2}{25}\right)r_{i_{zw} + 2} \leq \lambda r_{i_{zw} + 2} \qquad\text{(since $i_y \leq i_{zw}$ and $\lambda = 5^{20}(1 + \eps^{-1})$),}
		\end{split}
 	\end{equation} 
 	 as claimed. Similarly, $w' \in NC[x']$. Therefore, $(z', w') \in \Aug(x, 0, \gamma_2)$.
\end{proof}

\begin{proof}[Proof of \Cref{lm:detour-edge}]
	We find a $P$-jump from the set of cross edges of $E(G)$ if there is a long edge in some prefix of $P$. Otherwise, we claim that there exist $P$-jump from the set of argumented cross edges of nodes in $\mathrm{LNF}$. 
	
	Let $x_1 = \LCA(u_1)$ and $x_2, x_3, \ldots x_l$ be the ancestors of $u_2, u_3, \ldots u_l$ at level $i_1$, respectively. Let $j$ be the largest integer in $[1, l]$ such that $\nd(x_1, x_j) \leq \lambda r_{i_1}$. Since $\nd(\LCA(u_1), v_1) > \lambda r_{i_1}$, $j < l$. Let $(\hat{u}_1, \hat{u}_{j + 1})$ be the original cross edge of $(u_1, u_{j + 1})$. Because of the maximality of $j$, $\lvl(\hat{u}_1) = \lvl(\hat{u}_{j + 1}) > i_1$. Then, $(\hat{u}_1, \hat{u}_{j + 1})$ is a $P$-jump.
	
	If $\lvl(\hat{u}_1) \leq i_1 + 3\log{\lambda}$ then $(\hat{u}_1, \hat{u}_{j + 1}) \in E^*$ by \cref{line:begin-aug-original} -- \ref{line:end-aug-original} in \Cref{alg:VFT-spanner}. Using  \Cref{clm:augmented} with $\gamma_1 = 3\log{\lambda}$ and $\gamma_2 = 5\log{\lambda}$, we obtain $\Pa^2(\Aug(\hat{u}_1, 0, \gamma_2 - \gamma_1 - 2))$ is a subset of $\Aug(x_1, 0, 5\log{\lambda})$. Since $\gamma_2 - \gamma_1 - 2 = 5\log{\lambda} - 3\log{\lambda} - 2 \geq \log{\lambda}$, $\Pa^2(\Aug(\hat{u}_1, 0, \log{\lambda}))$ is a subset of $\Aug(x_1, 0, 5\log{\lambda})$ and hence is in $E^*$. Similarly, $\Pa^2(\Aug(\hat{u}_{j + 1}, 0, \log{\lambda}))$ is also in $E^*$.
	
	The last case is when $\lvl(\hat{u}_1) > i_1 + 3\log{\lambda}$. Let $(\tilde{u}_j, \tilde{u}_{j + 1})$ be the original cross edge of $(u_j, u_{j + 1})$. By \cref{line:begin-aug-original} -- \ref{line:end-aug-original}, $(\tilde{u}_j, \tilde{u}_{j + 1}) \in E^*$. Let $\hat{i}$ and $\tilde{i}$ be the levels of $\hat{u}_{j + 1}$ and $\tilde{u}_{j + 1}$, respectively.
	We claim that $\tilde{i} - 2 \leq \hat{i} \leq \tilde{i} + 2$, and therefore, there exists a good cross edge of $(u_1, u_{j + 1})$ in $E^*$. To do that, we show the following claim:
	
	\begin{claim}
		$\frac{1}{2} \leq \frac{\nd(\hat{u}_1, \hat{u}_{j + 1})}{\nd(\tilde{u}_j, \tilde{u}_{j + 1})} \leq 2$.
	\end{claim}
	  
	\begin{proof}
		We first prove that $\nd(\hat{u}_1, \hat{u}_{j + 1}) \leq 2 {\nd(\tilde{u}_j, \tilde{u}_{j + 1})}$. Since $\lvl(\hat{u}_1) > i_1 + 3\log{\lambda}$, $\nd(\hat{u}_1, \hat{u}_{j + 1}) \geq \lambda^3r_{i_1}$. By triangle inequality, $\nd(u_1, u_j) \leq (\lambda + 32)r_{i_1} \leq \frac{(\lambda + 32)\nd(\hat{u}_1, \hat{u}_{j + 1})}{\lambda^3} \leq \frac{\nd(\hat{u}_1, \hat{u}_{j + 1})}{\lambda}$ because $\lambda ^ 2 \geq 2\lambda \geq \lambda + 32$. By \Cref{it:good-approx-weight} of \Cref{prop:connectivity}, $\nd(\hat{u}_1, \hat{u}_{j+ 1}) \leq \frac{\nd(u_{1}, u_{j + 1})}{1 - 5^{-3}\eps}$. Using the triangle inequality,
		
		\begin{equation}
			\nd(u_{1}, u_{j}) \leq (\lambda + 32)r_{i_1} \leq \frac{\nd(\hat{u}_1, \hat{u}_{j + 1})}{\lambda} \leq \frac{\nd(u_1, u_{j + 1})}{\lambda \cdot (1 - 5^{-3}\eps)} \leq \frac{\nd(u_1, u_j) + \nd(u_j, u_{j + 1})}{\lambda \cdot (1 - 5^{-3}\eps)} \qquad,
		\end{equation}
		
		implying that $\nd(u_1, u_j) \leq \frac{\nd(u_j, u_{j + 1})}{\lambda \cdot (1 - 5^{-3}\eps) - 1} \leq \eps \nd(u_j, u_{j + 1})$. By triangle inequality, $\nd(u_1, u_{j + 1}) \leq \nd(u_{j}, u_{j + 1}) + \nd(u_1, u_{j}) \leq (1 + \eps)\nd(u_{j}, u_{j + 1})$. Then, we have:
		
		\begin{equation}
			\begin{split}
				\nd(\hat{u}_1, \hat{u}_{j + 1}) &\leq (1 + 5^{-3}\eps)\nd(u_1, u_{j + 1}) \qquad\text{(by It. \ref{it:good-approx-weight} of Prop. \ref{prop:connectivity})}\\
				&\leq (1 + 5^{-3}\eps)(1 + \eps)\nd(u_j, u_{j + 1}) \leq (1 + \eps)^2\nd(u_j, u_{j + 1})\\
				& \leq  \frac{(1 + \eps)^2}{1 - 5^{-3}\eps}\nd(\tilde{u}_j, \tilde{u}_{j + 1})  \qquad \text{(by It. \ref{it:good-approx-weight} of Prop. \ref{prop:connectivity})}\\
				&\leq 2\nd(\tilde{u}_j, \tilde{u}_{j + 1}) \qquad,
			\end{split}
		\end{equation} 
	
		as claimed. Similarly, $\nd(\hat{u}_1, \hat{u}_{j + 1}) \geq 1/2 \cdot {\nd(\tilde{u}_j, \tilde{u}_{j + 1})}$. 
	\end{proof}	 
	
	Hence, from \Cref{it:level-weight-relation} of \Cref{prop:connectivity}, we have $\hat{i} \leq \log{\frac{\nd(\hat{u}_1, \hat{u}_{j + 1})}{\lambda}} + 2$ and $\tilde{i} \geq \log{\frac{\nd(\tilde{u}_j, \tilde{u}_{j + 1})}{\lambda}}$. Therefore,
	\begin{equation}
		\hat{i} - \tilde{i} \leq  \log{\frac{\nd(\hat{u}_1, \hat{u}_{j + 1})}{\lambda}} + 2 - \log{\frac{\nd(\tilde{u}_j, \tilde{u}_{j + 1})}{\lambda}} =  \log{\frac{\nd(\hat{u}_1, \hat{u}_{j + 1})}{\nd(\tilde{u}_j, \tilde{u}_{j + 1})}} + 2 < 3 \qquad,
	\end{equation}
	implying that  $\hat{i} \leq \tilde{i} + 2$ since both $\hat{i}$ and $\tilde{i}$ are integers. Similarly, $\hat{i} \geq \tilde{i} - 2$.
	
	Let $(w_1, w_{j + 1})$ be the $2$-cross edge of $(u_1, u_{j + 1})$, i.e., the ancestor at level $\hat{i} + 2$ of $(\hat{u}_1, \hat{u}_{j + 1})$. $(w_1, w_{j + 1})$ is a $P$-jump by definition. We complete our proof by showing that $(w_1, w_{j + 1}) \in E^*$. By \cref{line:begin-aug-original} -- \ref{line:end-aug-original}, $\Aug(\tilde{u}_{j + 1}, 0, 5\log{\lambda}) \subseteq E^*$. Observe that $\lvl(w_{j + 1}) = \hat{i} + 2 \in [\tilde{i}, \tilde{i} + 4]$, implying that $w_{j + 1}$ is an ancester of $\tilde{u}_{j + 1}$ at a level lower than or equal to $\tilde{i} + 4$. Hence, $(w_1, w_{j + 1}) \in \Aug(\tilde{u}_{j + 1}, 0, 5\log{\lambda})$. Furthermore, by \Cref{clm:augmented}, we also obtain $\Pa^2(\Aug(w_1, 0, \log{\lambda}))$ and $\Pa^2(\Aug(w_{j + 1}, 0, \log{\lambda}))$ are in $E^*$.
\end{proof}

To construct a $P$-detour, we keep finding a jump recursively on subpaths of $P$ until we reach a complete end. We later prove that the set of those jumps has a structure called \emph{$P$-stair-jump}.

An oriented cross edge $(x \rightarrow y)$ is a cross edge with direction from $x$ to $y$. We formally define a $P$-stair-jump.

\begin{definition}
	\label{def:detour}
	Given a path $P = \{u_1, u_2, \ldots u_l\}$ of $G$, an ordered set of oriented cross edges  $\overrightarrow{D} = \{(x_1 \rightarrow y_1), (x_2 \rightarrow y_2), \ldots (x_h \rightarrow y_h)\}$ is a \textit{$P$-stair-jump} if:
	
	\begin{itemize}
		\item For any $k$, $(x_k, y_k)$ is a good cross edge of $(u_{a_k}, u_{b_k})$, $a_k < b_k$, $a_1 = 1$ and $a_{k + 1} = b_{k}$ for any $k < |\overrightarrow{D}|$.
		\item $(x_k, y_k)$ is a $P_{a_k}$-jump with $P_j = \{(u_j, u_{j + 1}), (u_{j + 1}, u_{j + 2}), \ldots (u_{l - 1}, u_l)\}$ for every $j \in [1, l]$.
		\item $\lvl(x_1) < \lvl(x_2) < \ldots < \lvl(x_h)$.
	\end{itemize}
\end{definition}

See \Cref{fig:stair} for an illustration. The set $D$ containing undirected versions of cross edges in $\overrightarrow{D}$ is an \textit{undirected $P$-stair-jump}. The node $y_h$ is the \textit{tail} of $D$ and $\overrightarrow{D}$.

\begin{center}
	\begin{figure}[H]
		\includegraphics[width=0.8\textwidth]{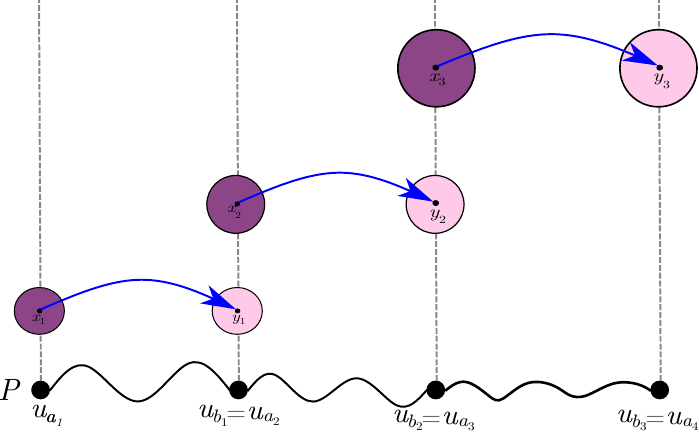}
		\caption{An example of a $P$-stair-jump with three cross edges. Note that $u_{a_1} = u$. Each pink node $y_k$ is a descendant of a purple node $x_{k + 1}$. When we construct a $P$-stair-jump later in the proof of \Cref{lm:detour-exist}, every pink node is incomplete and every purple node is complete.}
		\label{fig:stair}
	\end{figure}
\end{center}

We prove that the total weight of $\overrightarrow{D}$ is approximate the weight of the highest cross edge.

\begin{claim}
	\label{clm:stair-weight}
	Let $u$ and $v$ be two points in $X$ and $P = (u = u_1, u_2, \ldots u_l = v)$ be the shortest path from $u$ to $v$ in $G$. For any $P$-stair-jump $\overrightarrow{D} = \{(x_1 \rightarrow y_1), (x_2 \rightarrow y_2), \ldots (x_h \rightarrow y_h)\}$ with the last point $u'$, $\frac{\lambda}{250(1 + \eps)^2}r_{\lvl(x_h)} \leq \nd(u_1, u') \leq \frac{2\lambda}{1 - \eps}r_{\lvl(x_h)}$.
\end{claim}

\begin{proof}
	By \Cref{def:detour}, we have $\lvl(x_1) < \lvl(x_2) < \ldots \lvl(x_h)$. By \Cref{def:detour}, there exists $\{a_i, b_i\}_{i \in [1, h]}$ such that $a_1 = 1, a_{k + 1} = b_k$ for every $k \in [1, h - 1]$ and $(x_k, y_k)$ is a good cross edge of $(u_{a_k}, u_{b_k})$ for every $k \in [1, h]$. Note that $u' = u_{b_h}$. For every $k \in [1, h]$, since $(x_k, y_k)$ is a cross edge, $\nd(x_k, y_k) \leq \lambda r_{\lvl(x_t)}$. Hence, 
	\begin{equation}
		\sum_{k = 1}^h\nd(x_k, y_k) \leq \lambda \cdot \sum_{k = 1}^h r_{\lvl(x_k)} \leq 2\lambda \cdot r_{\lvl(x_h)} 	
	\end{equation}
	by a geometric sum. By \Cref{it:good-approx-weight} of \Cref{prop:connectivity}, $1 - 5^{-3}\eps \leq \frac{\nd(x_k, y_k)}{\nd(u_{a_k}, u_{b_k})} \leq 1 + 5^{-3}\eps$ for any $k$. Then, using the triangle inequality, we obtain:
	\begin{equation}
		\nd(u_1, u_{b_h}) \leq \sum_{k = 1}^{h}\nd(u_{a_k}, u_{b_k}) \leq \frac{1}{1 - 5^{-3}\eps}\sum_{k = 1}^h\nd(x_k, y_k) \leq \frac{2\lambda r_{\lvl(x_h) }}{1 - 5^{-3}\eps} \leq \frac{2\lambda r_{\lvl(x_h)}}{1 - \eps}.
	\end{equation}
	 For the lower bound, since $G$ is an $(1 + \eps)$-spanner of $X$, we have:
	\begin{equation}
		\begin{split}
			\nd(u_1, u_{b_h}) &\geq \frac{\dist_G(u_1, u_{b_h})}{1 + \eps} = \frac{1}{1 + \eps} \sum_{k = 1}^{h}\dist_G(u_{a_k}, u_{b_k})\\
			 &\geq \frac{1}{1 + \eps}\sum_{k = 1}^h\nd(u_{a_k}, u_{b_k}) \geq \frac{1}{(1 + \eps)(1 + 5^{-3}\eps)}\sum_{k = 1}^h\nd(x_k, y_k) \qquad\text{(by It. \ref{it:good-approx-weight} of Prop. \ref{prop:connectivity})}.
		\end{split}
	\end{equation}
	Since $(x_k, y_k)$ is a $\kappa$-cross edge of $(u_{a_k}, u_{b_k})$ with $0 \leq \kappa \leq 2$, $\nd(x_k, y_k) \geq \left(\frac{\lambda}{6\cdot 5^{2}} - 4\right)r_{\lvl(x_k)} \geq \frac{\lambda}{250} \cdot r_{\lvl(x_k)}$ by \Cref{it:good-cross-edge-long} of \Cref{prop:connectivity}. Thus, we have:
	\begin{equation}
		\nd(u_1, u_{b_h}) \geq \frac{1}{(1 + \eps)(1 + 5^{-3}\eps)}\sum_{k = 1}^h \frac{\lambda}{250} \cdot r_{\lvl(x_k)} \geq \frac{1}{(1 + \eps)^2} \frac{\lambda}{250} \cdot r_{\lvl(x_k)} \qquad,
	\end{equation}
	as claimed.
\end{proof}

In the next lemma, we show that the highest cross edge of a $P$-stair-jump $\overrightarrow{D}$ is approximately equal to the total distance from $u$ to the last point of $\overrightarrow{D}$, meaning that there is a $P$-detour whose level close to the highest cross edge in $\overrightarrow{D}$.

\begin{lemma}
	\label{lm:shortcut-crosspath}
	Let $u$ and $v$ be two points in $X$ and $P = (u = u_1, u_2, \ldots u_l = v)$ be the shortest path from $u$ to $v$ in $G$. For any  $P$-stair-jump $\overrightarrow{D} = \{(x_1 \rightarrow y_1), (x_2 \rightarrow y_2), \ldots (x_h \rightarrow y_h)\}$ with the last point $u'$, if $y_h$ is complete, there exists a good cross edge $(x, y)$ of $(u_1, u')$ in $\Pa^2(\Aug(x_h, 0, \lambda))$ such that $x$ and $y$ are complete.
\end{lemma} 

\begin{proof}
	Let $\{(a_k, b_k)\}_{k \in [1, h]}$ be the notation as in \Cref{def:detour}. We claim that the original cross edge of $(u, u_{b_h})$ locates at a level not too far from $\lvl(x_h)$.
	
	\begin{claim}
		\label{clm:lvl-highest}
		For every integer $k \in [1, h]$, let $i_k$ be the level of the original cross edge of $(u, u_{b_k})$. Then, $\lvl(x_k) - 3 \leq i_k \leq \lvl(x_k) + 3$.
	\end{claim}
	
	\begin{proof}
		Let $e = (u, u_{b_k})$ and $(\check{u}, \check{u}_{b_k})$ be the original cross edge of $e$. From \Cref{clm:stair-weight}, $\frac{\lambda r_{\lvl(x_k)} }{250(1 + \eps)^2} \leq \nd(u_1, u_{b_k}) \leq \frac{2\lambda r_{\lvl(x_k)}}{1 - \eps}$. Combining with \Cref{it:good-approx-weight} of \Cref{prop:connectivity}, we obtain:
		\begin{equation}
			\nd(\check{u}, \check{u}_{b_k}) \leq (1 + 5^{-3}\eps)\nd(u_1, u_{b_k}) \leq \frac{2\lambda r_{\lvl(x_k)} \cdot (1 + 5^{-3}\eps)}{1 - \eps} \qquad .
		\end{equation}
		Hence,  by \Cref{it:level-weight-relation} of \Cref{prop:connectivity}, we have
		\begin{equation}
			\lvl(\check{u}) \leq \log{\frac{\nd(\check{u}, \check{u}_{b_k})}{\lambda}} + 2 \leq \log{\frac{2\lambda r_{\lvl(x_k)} \cdot (1 + 5^{-3}\eps)}{(1 - \eps)\lambda}} + 2 \leq \lvl(x_k) + 3  ~,
		\end{equation}
		implying that $i_k \leq \lvl(x_k)$ since $i_k = \lvl(\check{u})$. On the other hand:
		
		\begin{equation}
			\nd(\check{u}, \check{u}_{b_k}) \geq (1 - 5^{-3}\eps)\nd(u_1, u_{b_k}) \geq (1 - 5^{-3}\eps) \frac{\lambda r_{\lvl(x_k)} }{250(1 + \eps)^2} \cdot  \qquad \text{(by \Cref{clm:stair-weight})} .
		\end{equation}
		Thus, by \Cref{it:level-weight-relation} of \Cref{prop:connectivity}, we have:
		\begin{equation}
			\begin{split}
				\lvl(\check{u}) &\geq \log \frac{\nd(\check{u}, \check{u}_{b_k})}{\lambda} \geq \log \left(\frac{(1 - 5^{-3}\eps) r_{\lvl(x_k)} }{250(1 + \eps)^2}\right) \\
				&\geq \lvl(x_k) + \log \left(\frac{1 - 5^{-3}\eps}{250(1 + \eps)^2}\right) \geq \lvl(x_k) - 3~,
			\end{split}	
		\end{equation}
		as $\eps \leq 1/20$. Therefore, we get $\lvl(\check{u}) \geq \lvl(x_k) - 3$, or $i_k \geq \lvl(x_k) - 3$ as claimed.
	\end{proof}

	Let $(\hat{u}, \hat{u}_{b_h})$ be the original cross edge of $(u, u_{b_h})$.  By \Cref{clm:lvl-highest}, $\lvl(x_h) - 3 \leq \lvl(\hat{u}) \leq  \lvl(x_h) + 3$. Let $(z_5, t_5)$ be the $5$-cross edge of $(u, u_{b_h})$, we prove that $(z_5, t_5)$ is in $\Aug(x_h, 0, \lambda)$, implying that the $7$-cross edge of $(u, u_{b_h})$ is in $\Pa^2(\Aug(x_h, 0, \log \lambda))$. Let $x_h'$ be the ancestor of $x_h$ at level $\lvl(z_5)$. $x_h'$ exists since
	\begin{equation}
		\label{eq:lvx5}
		\lvl(x_h') = \lvl(z_5) = \lvl(\hat{u}) + 5 \geq \lvl(x_h) + 2 > \lvl(x_h).
	\end{equation}
	 The set of cross edges between nodes in $NC[x_h']$ is a subset of $\Aug(x_h, 0, \log \lambda)$ because $\lvl(x_h') = \lvl(\hat{u}) + 5 \leq \lvl(x_h) + 8 \leq \lvl(x_h) + \log \lambda$. Since $(x_h, y_h)$ is a $\kappa$-cross edge of $(u_{a_h}, u_{b_h})$ with $0 \leq \kappa \leq 2$ and $(x_h', t_5)$ is an ancestor of $(x_h, y_h)$, $t_5 \in NC[x_h']$ by \Cref{it:cross-edge-ancestor} of \Cref{prop:connectivity}. We claim that $z_5$ is also in $NC[x_h']$. Let $(\tilde{u}, \tilde{u}_{b_{h - 1}})$ be the originial cross edge of $(u, u_{b_{h - 1}})$. By \Cref{clm:lvl-highest}, $\lvl(\tilde{u}) \leq \lvl(x_{h - 1}) + 3 \leq \lvl(x_h) + 2$, which is at most $\lvl(x_h')$ by \Cref{eq:lvx5}. Hence, $(z_5, x_h')$ is an ancestor of $(\tilde{u}, \tilde{u}_{b_{h - 1}})$ and therefore, $z_5 \in NC[x_h']$ by \Cref{it:cross-edge-ancestor} of \Cref{prop:connectivity}. 
	 
	 We finish our proof by showing that both ends of the $7$-cross edge of $(u, u_{b_h})$ are complete. Since the $7$-cross edge of $(u, u_{b_h})$ is an ancestor of $(z_5, t_5)$, it is enough to show that both $z_5$ and $t_5$ are ancestors of some complete ancestors of $u$ and $u_{b_h}$ by \Cref{clm:parent-complete}. Since $\lvl(z_5) > \lvl(x_h) \geq \lvl(x_1)$ by \Cref{eq:lvx5} and $x_1$ is a complete ancestor of $(u, 0)$, $z_5$ is an complete ancestor of $(u, 0)$. Similarly, $\lvl(t_5) = \lvl(z_5) \geq \lvl(y_h)$. Hence, $t_5$ is an ancestor of a complete ancestor $y_h$ of $(u_{b_h})$.	
\end{proof}

We now prove for any shortest path from $u$ to $v$ in $G$, there is a $P$-detour set in $E^*$. For the case that two points $u$ and $v$ are close, the next claim shows that if the level of the original cross edge of $(u, v)$ is less than the level of the lowest complete ancestor of any point in the shortest path between $u$ and $v$, there exists a good cross edge of $(u, v)$ in $E^*$ (which is also a $P$-detour).

\begin{claim}
	\label{clm:mid-small}
	Let $u$ and $v$ be two points in $X$, $i$ be the level of the original cross edge of $(u, v)$ and $P = (u = u_1, u_2, \ldots u_l = v)$ be a shortest path from $u$ to $v$ in $G$. For any point $u_j \in V(P)$ ($1 \leq j \leq l$), if $i \leq \lvl(\LCA(u_j))$, there is a good cross edge of $e$ in $E^*$.
\end{claim}

\begin{proof}
	Let $i_l, i_r$ be the level of the original cross edges of $(u, u_j)$ and $(u_j, v)$, respectively. We claim that $i_l, i_r \leq i + 2$. Let $(\hat{u}, \hat{v})$ and $(\tilde{u}, \tilde{u}_j)$ be $(u, v)$'s and $(u, u_j)$'s original cross edges, respectively. By \Cref{it:good-approx-weight} of \Cref{prop:connectivity}, we have:
	\begin{equation}
		\begin{split}
			\nd(\tilde{u}, \tilde{u}_j) &\leq (1 + 5^{-3}\eps)\nd(u, u_j) \leq (1 + 5^{-3}\eps)\dist_G(u, u_j)\\ 
			&\leq (1 + 5^{-3}\eps)\dist_G(u, v) \qquad\text{(since $P$ contains $u_j$ and $\dist_G(u, v) = \weight(P)$)}\\
			& \leq (1 + 5^{-3}\eps)(1 + \eps)\nd(u, v) \qquad \text{(since $G$ is a $(1 + \eps)$-spanner of $X$)}\\
			&\leq (1 + \eps)\frac{1 + 5^{-3}\eps}{1 - 5^{-3}\eps}\nd(\hat{u}, \hat{v}) \leq 2 \nd(\hat{u}, \hat{v}) \qquad\text{(by It. \ref{it:good-approx-weight} of Prop. \ref{prop:connectivity})}.
		\end{split}
	\end{equation}
	By \Cref{it:level-weight-relation} of \Cref{prop:connectivity}, $i_l \leq \log{\frac{\nd(\tilde{u}, \tilde{u}_j)}{\lambda}} + 2 \leq \log \frac{2\nd(\hat{u}, \hat{v})}{\lambda} + 2 < i + 3$, implying that $i_l \leq i + 2$ since both $i_l$ and $i$ are integers. Similarly, $i_r \leq i + 2$. Let $\check{u}, \check{u}_j$ and $\check{v}$ be the ancestor of $(u, 0), (u_j, 0)$ and $(v, 0)$ at level $i + 3$. Since $i_l, i_r \leq i + 3$, $(\check{u}, \check{u}_j)$ and $(\check{u}_j, \check{v})$ are ancestors of $(u, u_j)$'s and $(u_j, v)$'s orginal cross edges, respectively. Hence, $\check{u}, \check{v} \in NC[\check{u}_j]$ by \Cref{it:cross-edge-ancestor} of \Cref{prop:connectivity}. Furthermore, $(\check{u}, \check{v})$ is a good cross edge of $(u, v)$ by definition. If $\check{u}_j$ is incomplete, then all cross edges between nodes in $NC[\check{u}_j]$ is in $E^*$ by \cref{line:add-NCx} of \Cref{alg:VFT-spanner}, implying that $(\check{u}, \check{v}) \in E^*$. If $\check{u}_j$ is complete, then $\check{u}_j$ has an incomplete ancestor at level $\lvl(\LCA(u_j)) - 1$ by the definition of lowest complete ancestor. Since $\lvl(\LCA(u_j)) - 1 \geq i - 1 \geq \lvl(\check{u}_j) - 5\log{\lambda}$, by \cref{line:add-NCx}, all (long enough) cross edges in $NC[\check{u}_j]$ are also in $E^*$, implying that $(\check{u}, \check{v}) \in E^*$.
\end{proof}

We now ready to prove there is always a $P$-detour for every shortest path $P$ in $G$.  

\begin{proof}[Proof of \Cref{lm:detour-exist}]
	The idea is to find a $P$-stair-jump $\overrightarrow{D}$. If the tail $y$ of $\overrightarrow{D}$ is complete, there exists a $P$-detour in $E^*$ by \Cref{lm:shortcut-crosspath}. Otherwise, if $y$ is incomplete, we keep adding another jump to $\overrightarrow{D}$. If we cannot find any jump, then the subtree containing $y$, denoted by $T_y$, in the LNF is close to both $u$ and $v$, implying that we can find a detour directly from an ancestor of $(u, 0)$ to an ancestor of $(v, 0)$ in the set of augmented cross edges of nodes in $T_y$.
	
	Let $P = (u = u_1, u_2, \ldots u_l = v)$, $i_1$ be the level of $\LCA(u_1)$ and $v_1$ be the ancestor of $v$ at level $i_1$. We denote by $i$ the level of the original cross edge of $(u, v)$. We claim that there is either a good cross edge of $(u, v)$ or a $P$-stair-jump in $E^*$.
	
	Let $\overrightarrow{D}$ be a list of oriented cross edges. At the beginning, $\overrightarrow{D} = \emptyset$ and we consider $P_1$. If $\nd(\LCA(u_1), v_1) \leq \lambda r_{i_1}$, then the original cross edge of $(u, v)$ has level at most $i_1$. By \Cref{clm:mid-small}, there exists a good cross edge of $(u, v)$ in $E^*$. Otherwise, $\nd(\LCA(u_1), v_1) > \lambda r_{i_1}$, by \Cref{lm:detour-edge}, $E^*$ contains a good cross edge $(x_1, y_1)$ of $(u_1, u_{b_1})$ for some $b_1 > 1$. Let $a_1 = 1$. Furthermore, we have $\Pa^2(\Aug(x_1, 0, \log{\lambda}))$ and $\Pa^2(\Aug(y_1, 0 , \log{\lambda}))$ are also in $E^*$. If $y_1$ is complete, then we found a $P$-stair-jump with one element. Otherwise, we keep doing this process with the path $P_{b_1}$. Assuming that at current step, $\overrightarrow{D} = \{(x_1, y_1), (x_2, y_2), \ldots (x_h, y_h)\}$ with $h > 0$ such that:
	\begin{itemize}
		\item For every $k \in [1, h]$, $(x_k, y_k)$ is a good cross edge of $(u_{a_k}, u_{b_k})$ with $a_k < b_k$, $a_1 = 1$ and $a_{k + 1} = b_{k}$ if $k < h$.
		\item $y_k$ is incomplete for every $k \in [1, h]$.
		\item $\Pa^2(\Aug(x_k, 0, \log\lambda))$ and $\Pa^2(\Aug(y_k, 0, \log\lambda))$ are in $E^*$ for every $k \in [1, h]$.
	\end{itemize}
	
	See \Cref{fig:stair} for an illustration. For every $k \in [1, h]$, let $i_k = \lvl(x_k)$ ($= \lvl(y_k)$). Let $j = b_h$, $g_j = \lvl(\LCA(u_j))$ and $v_j$ be the ancestor of $(v, 0)$ at level $g_j$. If $\nd(\LCA(u_j), v_j) \leq \lambda r_{g_j}$, we claim that there is an good cross edge between $u$ and $v$ in $E^*$.
	
	\begin{claim}
		\label{clm:short-tail-cross}
		If $\nd(\LCA(u_{j}), v_j) \leq \lambda r_{g_j}$, then there exists a good cross edge $(x, y)$ of $(u, v)$ in $E^*$.
	\end{claim}
	
	\begin{proof}
		By \Cref{clm:stair-weight}, $ \frac{\lambda}{250(1 + \eps)^2}r_{\lvl(x_h)} \leq \nd(u_1, u_j) \leq \frac{2\lambda r_{\lvl(x_h)}}{1 - \eps}$. Let $(\hat{u}, \hat{v})$ be the original cross edge of $(u, v)$ and $i$ is the level of $\hat{u}$. We have $\nd(\hat{u}, \hat{v}) \leq \lambda r_i$. We show that $i \leq g_j + 2$. Observe that:
		
		\begin{equation}
			\label{eq:hatuv1}
			\begin{split}
				\nd(\hat{u}, \hat{v}) &\leq (1 + 5^{-3}\eps)\nd(u, v) \leq (1 + 5^{-3}\eps)(\nd(u, u_j) + \nd(u_j, v))\qquad\text{(by Eq. \ref{it:good-approx-weight} of Prop. \ref{prop:connectivity})}\\
				&\leq (1 + 5^{-3}\eps)\left( \frac{2\lambda r_{i_h}}{1 - \eps} + \nd(u_j, v)\right)  \qquad\text{(by \Cref{clm:stair-weight})}
			\end{split}
		\end{equation}
	
		Since $\nd(\LCA(u_j), v_j) \leq \lambda r_{i_j}$, $(\LCA(u_j), v_j)$ is an ancestor of the original cross edge of $(u_j, v)$, denoted by $(\tilde{u}_j, \tilde{v})$. Let $\tilde{i} = \lvl(\tilde{v})$, then $\nd(\tilde{u}, \tilde{v}) \leq \lambda r_{\tilde{i}} \leq \lambda r_{g_j}$ by the minimality of the level of the original cross edge. Using \Cref{it:good-approx-weight} of \Cref{prop:connectivity}, we get:
		\begin{equation}
			\label{eq:hatuv2}
			\begin{split}
				\nd(u_j, v) &\leq \frac{\nd(\tilde{u}_j, \tilde{v})}{1 - 5^{-3}\eps} \leq \frac{\lambda r_{g_j}}{1 - 5^{-3}\eps} \leq \frac{\lambda r_{g_j}}{1 - \eps}\qquad.
			\end{split}
		\end{equation}
		Plugging in \Cref{eq:hatuv1}, we have: 
		\begin{equation}
			\label{eq:dhatuv}
			\nd(\hat{u}, \hat{v}) \leq \frac{\lambda(1 + 5^{-3}\eps)}{1 - \eps}(2r_{i_h} + r_{g_j})
		\end{equation}
		Since $y_h$ is incomplete, we have $i_h = \lvl(x_h) = \lvl(y_h) \leq g_j - 1$, implying that $2r_{i_h} \leq r_{g_j}$. Then, by \Cref{eq:dhatuv}, $\nd(\hat{u}, \hat{v}) \leq \frac{2(1 + 5^{-3}\eps)\lambda r_{g_j}}{1 - \eps}$. From \Cref{it:level-weight-relation} of \Cref{prop:connectivity}, 
		\begin{equation}
			i \leq \log \frac{\nd(\hat{u}, \hat{v})}{\lambda} + 2 \leq\log \frac{2(1 + 5^{-3}\eps)r_{g_j}}{1 - \eps} + 2\leq g_j + 3 \qquad\text{(as $\eps \leq 1/20$)}
		\end{equation}
	
		If $i \leq g_j$, by \Cref{clm:mid-small}, there is a good cross edge of $(u, v)$ in $E^*$. Otherwise, $i > g_j$. Let $(\hat{u}_3, \hat{v}_3)$ be the ancestor of $(\hat{u}, \hat{v})$ at level $i + 3$. Let $\hat{u}_j$ be the ancestor of $u_j$ at level $i + 3$. Since $i + 3 \leq g_j + 6 \leq g_j + 5\lambda$, $\Aug_{i + 3}(\LCA(u_j)) \subseteq \Aug(\LCA(u_j), 0, 5\log \lambda) \subseteq E^*$ by \cref{line:Dx-def} -- \ref{line:add-NCx} of \Cref{alg:VFT-spanner}. In other words, all cross edges between nodes in $NC[\hat{u}_j]$ are in $E^*$. From \Cref{clm:lvl-highest}, the level of the original cross edge of $(u, u_j)$ is at most $\lvl(x_h) + 3 = i_h + 3 \leq g_j + 2 \leq i + 2$, then $(\hat{u}_3, \hat{u}_j)$ is an ancestor of the original cross edge of $(u, u_j)$. By \Cref{it:cross-edge-ancestor} of \Cref{prop:connectivity}, $\hat{u}_3 \in NC[\hat{u}_j]$. Since $\nd(\LCA(u_j), v_j) \leq \lambda r_{g_j}$ and $(\hat{u}_j, \hat{v}_3)$ is an ancestor of $(\LCA(u_j), v_j)$, $\hat{v}_3 \in NC[\hat{u}_j]$. Since $(\hat{u}_3, \hat{v}_3)$ is a good cross edge of $(u, v)$ and $\hat{u}_3, \hat{v}_3 \in NC[\hat{u}_j]$, $(\hat{u}_3, \hat{v}_3) \in E^*$.
		
	\end{proof}
	
	 By \Cref{clm:short-tail-cross}, if $\nd(\LCA(u_j), v_j) \leq \lambda r_{i_j}$, there exists a good cross edge of $(u, v)$ in $E^*$, which gives us the lemma. Otherwise, by \Cref{lm:detour-edge}, there exists a $P_j$-jump $(x_{h + 1} , y_{h + 1})$ with $x_h$ is a complete ancestor of $(u_j, 0)$ and $y_h$ is an ancestor of $(u_{j'}, 0)$ for some $j' > j$. Set $a_{h + 1} = j$ and $b_{h + 1} = j'$. We append $(x_{h + 1} \rightarrow y_{h + 1})$ to $\overrightarrow{D}$. If $y_{h + 1}$ is incomplete, we continue the process recursively with $P_{j'}$. Otherwise $\overrightarrow{D}$ is a $P$-detour-path with complete tail. From \Cref{lm:shortcut-crosspath}, there exists a good cross edge $(x, y)$ of $(u, u_{j'})$ such that both $x$ and $y$ are complete. Hence, $(x, y)$ is a $P$-detour by definition.
\end{proof}

%% file: fast_implementation.tex
\section{Fast Implementation}
\label{sec:fast}

In this section, we describe an implementation of \Cref{alg:VFT-spanner} in time $O(nf + n\log{n})$, which is asymptotically optimal. We consider the running time of each step in \Cref{alg:VFT-spanner}:

\begin{itemize}
	\item \Cref{line:net-tree}: Constructing a light spanner can be done in $O(n\log{n})$ time \cite{FS16}. A net tree $T$ is constructed in $O(n\log{n})$ \cite{HPM06}.
	\item \Cref{line:begin-aug-original}--\ref{line:end-aug-original}: The total number of cross edges in $T$ is $O_{\eps, d}(n)$. Hence, finding all the original cross edges and adding them to $E^*$ cost $O_{\eps, d}(n)$ time. For each original cross edge $(\hat{u}, \hat{v})$, we add (long enough) cross edges in $Aug(\hat{u}, 0, 5\log{\lambda})$ and $Aug(\hat{v}, 0, \log{\lambda})$ to $E^*$. Since each of these sets has a constant size and can be found in constant time, the for loop in \cref{line:begin-aug-original}--\ref{line:end-aug-original} can be implemented in $O(n)$ time.
	\item \Cref{line:level-iter}--\ref{line:H-update}: In this (nested) for loop, we have two for loops. We indeed do not perform $\zeta$ loops since there are many levels with no cross edge. Since the number of cross edges is at most $O(n)$, we actually only need to consider the running time for $O(n)$ loops.
	\begin{itemize}
		\item \Cref{line:begin-add-incomplete}--\ref{line:end-add-incomplete}: For each node $x$, finding whether $x$ has a small descendant $w$ such that $w$ has at most $f$ leaves (recall that such $w$ is called incomplete) needs constant time. Adding cross edges between nodes in $NC[x]$ to $E^*$ also requires only constant time since the size of $NC[x]$ is constant due to the packing bound (\Cref{lm:packing-bound}).
		\item \Cref{line:edge-process}--\ref{line:H-update}: Since $H$ has at most $O(nf)$ edges, adding egdes to $H$ in \cref{line:H-update} only need $O(nf)$ time in total ($O(f)$ on average). Let $U$ be the upper bound time complexity for \Cref{alg:Surrogate}, the amount of time needed to run this for loop is $2U + O(f)$. 
	\end{itemize}
\end{itemize}

In total, the time complexity of \Cref{alg:VFT-spanner} is $O(nf + n\log{n} + nU)$. The running time of \Cref{alg:Surrogate} is therefore crucial if we want to achieve optimal time. We then focus on how to implement \Cref{alg:Surrogate} such that its amortized cost is $O(f + \log{n})$, meaning the total running time of all calls to \Cref{alg:Surrogate} is $O(n(\log{n} + f))$.

We generalize the set $\ball(u, 4r_i)$ as the \emph{pool set} of $(u, i)$, denoted by $P(u, i)$, and $\ball(u, 16r_i)$ as the \emph{extended pool set} of $(u, i)$, denoted by $P^+(u, i)$. We re-write \Cref{alg:Surrogate} as \Cref{alg:Implement-Surrogate}.

\begin{algorithm}[!htp]
	\caption{Implement-SelectSurrogate}\label{alg:Implement-Surrogate}
	\KwIn{A node $x \in V(T)$}
	\KwOut{A surrogate set $S(x)$}
	\eIf{$x$ is small \label{line:imp-check-small}}{	
		$S(x)\leftarrow$ arbitrary $f+1$ leaves of $T(x)$ if $T(x)$ has at least $f+1$ leaves; otherwise $S(x)\leftarrow $ all leaves of $T(x)$\label{line:imp-add-leaf-S}\;}
	{
		$S' \leftarrow \{u \in X\cap P^+(x): \deg_H(u) \geq c_2 \cdot f \mbox{ and $u$ is not saturated}\}$\label{line:imp-def-Sprime}\;
		$S'' \leftarrow \{u \in X\cap P(x): \deg_H(u) <  c_2 \cdot f \}$\label{line:imp-def-Sdprime}\;
		\eIf{$|S'| \geq f+1$}{
			$S(x)\leftarrow$  arbitrary $f+1$ vertices in $S'$\;
		}{
			$S(x)\leftarrow $ $S'\cup \{\mbox{arbitrary $f+1 - |S'|$ vertices in $S''$}\}$\label{line:imp-general-surrogate}\;    
		}
	}
	\Return $S(x)$\;
\end{algorithm}

Given a node $x = (u, i)$, we choose the set $P(x)$ and $P^+(x)$ such that \Cref{alg:VFT-spanner} with the call to \Cref{alg:Implement-Surrogate} in \cref{line:Sxdef} and \cref{line:Sydef} still returns a bounded degree, bounded lightness $(1 + \eps)$-VFT spanner. We keep the same setting for the pool set as in \Cref{alg:Surrogate}, i.e., $P(u, i) = \ball(u, 4r_i)$. To maintain a set of clean and semi-saturated points, we use two lists attached to each node in the tree. For each node $(u, i)$, we store a list $\mathrm{Clean}(u, i)$ of at most $4f + 4$ clean points in $P(x)$. We maintain this set by following exactly the proof of \Cref{lm:complete-ball}: when a point in $\mathrm{Clean}(u, i)$ becomes semi-saturated, we look down to the path of $T$ from $(u, i)$ to one of its leaf with a high current degree, call $v$. We denote this path by $P = \{x_i, x_{i - 1}, \ldots x_0\}$ with $x_k$ is the ancestor of $(v, 0)$ at level $k$ ($x_i = (u, i)$). Then, we find the clean points near that branch to add to $\ball(u, 4r_i)$. 

A problem with this approach is that we may have to search at some level much lower than $i$ to find enough clean points. This problem is due to the fact that for many level $j$, there might be no level-$j$ edge incident to some point in $P(x_j)$. However, we do not have to search over all the levels until getting enough $4f + 4$ clean points. For each node $x$, we store a pointer pointing to its descendant $x'$ such that the subtree rooted $x'$ also contains $x$'s highest-degree leaf and there exists a level-$j$ edge incident to some point in $P(x)$. Hence, at the time a point in $P(x)$ turns from clean to semi-saturated, we only need to look down at most $4f + 4$ nodes to find a substitution in $P(x)$. Since each point only turns from clean to semi-saturated once, we only need $O_{\eps, d}(nf)$ time in total for updating from clean to semi-saturated.

The harder task is to maintain the list of semi-saturated points in $P^{+}(x)$. We do not have any information about how many semi-saturated points are in $P^{+}(x)$. Since semi-saturated points in $P^{+}(x)$ are prioritized over the clean points in $P(x)$, we need to make sure that all up to $f + 1$ semi-saturated points in $P^{+}(x)$ can be found in an efficient way. However, there is no information about where should we look for semi-saturated points in $P^{+}(x)$ if we keep setting $P^{+}(x) = \ball(u, 16r_i)$. Observe that it is easier if we can compute the set $P^+(x)$ recursively from the extended pool of some below nodes. Then, we want to modify $P^+(x)$ such that it still preserves the crucial properties needed to prove the maximum degree, bounded lightness, and fault-tolerant property but can be computed recursively. 

From \Cref{rm:ext-pool}, to ensure the correctness of \Cref{alg:VFT-spanner}, it is sufficient that $\ball(x, 12r_i) \subseteq P^+(x) \subseteq \ball(x, 16r_i)$. Then, we define $P^+(x)$ recursively as follow:

\begin{itemize}
	\item If $i \leq 3$, $P^+(x) = \ball(x, 16r_i)$
	\item If $i > 3$, let $D_{i - 2}(x)$ be the set of node $y$ at level $i - 2$ such that $\ball(y, 12r_{i - 2}) \cap \ball(x, 12r_i) \neq \emptyset$. Then, $P^+(x) = \bigcup_{y \in D_{i - 2}(x)}P^+(y)$.    
\end{itemize}

We prove $P^+(x)$ possesses the required property in \Cref{rm:ext-pool}. 

\begin{claim}
	\label{clm:extend-ball-in-ring}
	For every node $x$ at level $i$, $\ball(x, 12r_i) \subseteq P^+(x) \subseteq \ball(x, 16r_i)$.
\end{claim}

\begin{proof}
	If $i \leq 3$, $P^+(x) = \ball(x, 16r_i)$. \Cref{clm:extend-ball-in-ring} holds trivially. Assume that \Cref{clm:extend-ball-in-ring} holds for every node at any level lower than $i$. Let $C^+(x)$ be the set of nodes at level $i - 2$ such that for every $y \in C^+(x)$, $\ball(y, 12r_{i - 2})$ intersects $\ball(x, 12r_i)$. By \Cref{clm:leaf-dist}, for every point $w$ in $X$, the distance between $w$ and its ancestor at level $i - 2$ is at most $5r_{i - 2}/4$. Hence, the union of $\ball(y, 12r_{i - 2})$ over all node $y$ at level $i - 2$ covers the whole space $X$, implying that the union of $\ball(y, 12r_{i - 2}) \cap \ball(x, 12r_i)$ over all $y \in C^+(x)$ covers $ \ball(x, 12r_i)$. Thus, $P^+(x)$ contains $\ball(x, 12r_i)$.
	
	We complete our proof by showing $\nd(x, w) \leq 16r_i$ for every $w \in \bigcup_{y \in C^+(x)}P^+(y)$. For every $w \in P^+(x)$, let $y$ be the node in $C^+(x)$ that $w \in P^+(y)$. By our induction hypothesis, $w \in \ball(y, 16r_{i - 2})$. Since $\ball(y, 12r_{i - 2}) \cap \ball(x, 12r_i) \neq \emptyset$, $\nd(x, y) \leq 12r_{i - 2} + 12r_i$. By the triangle inequality,
	\begin{equation*}
		\nd(x, w) \leq \nd(x, y) + \nd(y, w) \leq 12r_{i - 2} + 12r_i + 16r_{i - 2} \leq r_i(12 + 12/25 + 16/25) \leq 16r_i \qquad,
	\end{equation*}
	implying that $w \in \ball(x, 16r_i)$. Hence, $P^+(x) \subseteq \ball(x, 16r_i)$.
\end{proof}

Now, we describe how to maintain $P^+(x)$. For each node $x$, we store a list $C^+(x)$ of nodes at level $i - 2$ such that for every $y \in C^+(x)$, $\ball(x, 12r_i) \cap \ball(y, 12r_{i - 2}) \neq \emptyset$. Each node in $C^+(x)$ is called an \emph{extended child} of $x$ and $x$ is an \emph{extended parent} of every node in $C^+(x)$ (a node might have multiple extended parents). Similarly, a node $x$ is called an \emph{extended ancestor} of $y$ if $x$ is the extended parent of either $y$ or of some extended ancestor of $y$. The node $y$ is then called an \emph{extended descendant} of $x$. 
Observe that $x$ does not have any extended descendant  at a level $j$ if $i - j$ is odd. Similarly, there is no extended ancestor of $x$ at level $k$ if $k - i$ is odd. 
We construct an extended graph $T^+$ whose nodes are $V(T)$ plus $O(n)$ additional nodes that we will define later.

For any two nodes $x$ and $y$ in $V(T)$, $(x, y) \in E(T^+)$ if and only if $x$ is an extended parent of $y$. The idea to compute $P^+(x)$ is to merge the lists of semi-saturated points from some extended descendants of $x$ plus the new semi-saturated points. However, storing all the semi-saturated points in $P^+(y)$ for all $y \in V(T)$ will increase the time needed to find semi-saturated points since the extended pool of all extended children might share many points in common. Moreover, merging two lists with multiplicity in sublinear time is impossible. Thus, we only store the list of semi-saturated points in the extended pool of \emph{leaf} nodes of $T^+$, meaning nodes with no extended child.
\\\\
\textbf{Information stored on each node: } For each node $x$, we store a list of $cf$ semi-saturated points in $P^+(x)$, denoted by $\mathrm{SSList}(x)$. Here, $c = \eps^{-O(d)}$ is the constant in \Cref{clm:bdd-new-ss}. If $x$ is a leaf, we store all semi-saturated points in $x$. The lists of semi-saturated points of each node must be computed upward. We set $\mathrm{SSList}(x)$ to be the union (up-to $cf$ points) of all $\mathrm{SSList}(y)$ such that $y \in C^+(x)$. During some iterations, there might be some points in $P^+(x)$ turning to semi-saturated and some turning to saturated. While the former ones only require a local modification on some nodes, the latter ones are more challenging to update. 
\\\\
\textbf{Update when a point becomes semi-saturated: } When some point $u$ turns into semi-saturated, $u$ might belong to some extended pools of nodes at a much lower level. However, it is unnecessarily costly to update every $P^+(y)$ such that $u \in P^+(y)$. In fact, a new semi-saturated point at level $i$ only affects the extended pool sets of nodes at level $i$ and $i + 1$ since for nodes at higher levels, we update their extended pools recursively. Hence, for each node $x$ at level $i$ in $V(T)$, we create an extended child $x'$ of $x$ and add $(x, x')$ to $T^+$. $P^+(x')$ contains points in $P^+(x)$ which turn into semi-saturated during iteration $i$ or $i - 1$. Observe that $x'$ is a leaf of $T^+$. We show that the total points in $P^+(x')$ is at most $\eps^{-O(d)}f$.

\begin{claim}
	\label{clm:bdd-new-ss}
	Let $x$ be a node at level $i$ in $V(T)$. Then, the number of points in $P^+(x)$ becoming semi-saturated during iteration $i$ or $i - 1$ is at most $cf$ where $c = \eps^{-O(d)}$. 
\end{claim} 

\begin{proof}
	We only consider the number of points in $P^+(x)$ turning into semi-saturated in iteration $i$, those in iteration $i - 1$ will be counted similarly. A point $u \in P^+(x)$ becomes semi-saturated during iteration $i$ only if it is a surrogate of some node $y$ at level $i$. Using the triangle inequality, we have:
	\begin{equation}
		\nd(x, y) \leq \nd(x, u) + \nd(u, y) \leq 16r_i + 16r_i = 32r_i.
	\end{equation}
	Hence, $y \in \ball(x, 32r_i)$. By packing bound (\Cref{lm:packing-bound}), there are at most $2^{O(d)}$ such $y$.
	Furthermore, from \Cref{lm:bounded-surrogate}, for every node $y$ at level $i$, there are at most $\eps^{-O(d)}f$ points which are used as surrogates of $y$. Hence, the total number of surrogates of $y$ becoming semi-saturated is $\eps^{-O(d)}$.
	Therefore, the total number of points in $P^+(x)$ turning into semi-saturated during iteration $i$ is at most the number of nodes in $\ball(x, 32r_i)$ times the maximum number of surrogates of a node turning into semi-saturated, which is $\eps^{-O(d)}$.
\end{proof}

\noindent\textbf{Update when a point becomes saturated: } A node $y$ is marked as DELETED if all semi-saturated points in all leaves of $y$ have become saturated. For each node $x$ at level $i$, we add a pointer pointing to a leaf of $x$, denoted by $\rho(x)$, such that $\rho(x)$ is not deleted. 
Initially, $\rho(x) = \rho(z)$ for some $z \in C^+(x)$ such that $\rho(z) \neq \mathrm{NULL}$.
When a point $u$ becomes saturated after iteration $i$, we update $\mathrm{SSList}(x)$ for every $x$ at level $i$ and $i - 1$ such that $\mathrm{SSList}(x)$ contains $u$. By packing bound (\Cref{lm:packing-bound}), there are a constant number of such nodes $x$. We then focus on updating only one node $x$ at level $i$ such that $\mathrm{SSList}(x)$ contains $u$. 

We use $\rho(x)$ to find an extended leaf of $x$ that still has more semi-saturated points. Let $z = \rho(x)$. If there exists a semi-saturated point in $\mathrm{SSList}(z)$, we update $\mathrm{SSList}(x)$ by adding that point. However, if there is no semi-saturated point in $\mathrm{SSList}(z)$, a recursive procedure named $\mathrm{NodeDeletion}(z)$ is invoked to delete $z$. The procedure is as follows:

For every node $x'$ at level $i$ or $i - 1$ such that $\rho(x') = z$, we update $\rho(x')$. First, we find an extended parent $z'$ of $z$ such that $x'$ is the ancestor of $z'$. Second, we update $\rho(z')$ by finding an extended child of $z'$ which is not deleted and set $\rho(z')$ to be that child. If all extended children of $z'$ are deleted, we mark $z'$ as deleted and recursively run $\mathrm{NodeDeletion}(z')$. Otherwise, we set $\rho(x') = \rho(z')$. We later argue that each call to procedure $\mathrm{NodeDeletion}(\cdot)$ runs in $O(\log{n})$.
\\\\
\textbf{Running time analysis of all updates: } First, we prove that the total number of nodes in $T^+$ is $|V(T)| + O_{d}(n)$.

\begin{claim}
	The total number of nodes in $T^+$ is $|V(T)| + O_{d}(n)$.
\end{claim}

\begin{proof}
	We prove that the size of the set containing extended leaves added to $T^+$, call $A$, is $O(n)$. Recall that for each node $x$ at level $i$, we add an extended leaf to $x$ at level $i - 1$ containing all points in $P^+(x)$ that become semi-saturated during iteration $i - 1$ or $i$. The leaf corresponding to $x$, call $x'$, is an artificial node that contains all points in $P^+(x)$ turning into semi-saturated during iteration $i$ or $i - 1$. We call $x'$ an \emph{artificial leaf}. Since each point $u \in X$ only turns into semi-saturated once, it is easy to verify that there are $2^{O(d)}$ artificial leaves $y$ such that $u \in P^+(y)$. Let $A_u$ be the set of those artificial leaves, we have:
	\begin{equation}
		\begin{split}
			|A| \leq \sum_{u \in X}|A_u| \leq \sum_{u \in X}2^{-O(d)} = n\cdot 2^{O(d)}.
		\end{split}
	\end{equation}
	Therefore, $|V(T^+)| = |V(T)| + |A| \leq  |V(T)|  + O(n)$.
\end{proof}

In our construction, we only focus on \emph{relevant nodes}. A relevant node is either a leaf or a node with at least $2$ extended children in $T^+$. Observe that a non-leaf node $x$ at level $i$ is relevant if there is a node $y$ at level $i - 2$ such that $\nd(x, y) \leq 12r_{i - 2} + 12r_i = O(r_{i - 2})$. The number of such pairs $(x, y)$ is bounded by $O(n)$ by using the similar technique as in bounding the number of cross edges (See Theorem $5.3$ \cite{CGMZ16}). Thus, the total number of nodes in $T^+$ is $O(n)$. The update procedure when a point becomes semi-saturated therefore needs $O(n)$ time in total. It remains to prove that $\mathrm{NodeDeletion}(\cdot)$ can be done in $O(nf + n\log{n})$ time in total. 

Observe that for each node $z$, $\mathrm{NodeDeletion}(z)$ only updates $z$'s extended parents and $z$'s extended ancestors at level $i$ and $i - 1$ with $i$ is the current iteration in \Cref{line:level-iter} of \Cref{alg:VFT-spanner}. The following observation of the packing bound implies that $\mathrm{NodeDeletion}(z)$ only updates a constant number of nodes. 

\begin{observation}
	For every node $x$ at level $i$ of $T^+$, we have:
	\begin{enumerate}
		\item \label{it:constant-child} The number of $x$'s extended children is $O(1)$.
		\item \label{it:constant-anc} For every level $k > i$, the number of extended ancestors of $x$ at level $k$ is $O(1)$. 
	\end{enumerate}
\end{observation}

We complete our implementation by showing that for each node, all of its extended ancestors can be found in $O(\log{n})$ time.

For each node $x$ at level $i$, we store the extended ancestors of $x$ in $T^+$ at level $i + 1, i + 2, i + 4, \ldots i + 2^h$ for $h \in [0, \log n]$. Since $x$ has a constant number of ancestors at a certain level, we only need time (and space) $O(\log n)$ for each node $x$. Therefore, the total time complexity for this step is $O(n \log n)$. For each node $x$, computing the ancestor of $x$ at level $k > i$ is done in $O(\log n)$ time.

In total, for each node in $T^+$, we will delete it once, and each deletion procedure takes time $O(\log n)$. The construction of $T^+$ requires time $O(|V(T^+)|f)$ since each node $x$ in $T^+$ has at most $cf$ points in its surrogates list $\mathrm{SSList}(x)$. Therefore, the total running time of computing all $P^+(\cdot)$ is then $O(n\log n + nf)$.